\documentclass[a4paper]{llncs}

\usepackage{amsmath}
\usepackage[T1]{fontenc}
\usepackage{amsfonts}
\usepackage{amssymb}
\usepackage{float}
\usepackage{url}
\usepackage{tikz}
\usetikzlibrary{automata}
\usetikzlibrary{arrows}
\usetikzlibrary{shapes}
\usetikzlibrary{decorations.pathmorphing}
\usetikzlibrary{fit}
\usepackage{algorithm}
\usepackage{algorithmic}
\usepackage{hyperref}
\usepackage[margin=1.5cm]{geometry}
\usepackage{multicol}

\allowdisplaybreaks[1]

\tikzstyle{every picture}=[
  >=stealth', shorten >=1pt, node distance=1.44cm,auto,bend angle=45,initial text=,
  every state/.style={inner sep=0.75mm, minimum size=1mm},font=\scriptsize,
]

\makeatletter
\newcommand{\subalign}[1]{%
  \vcenter{%
    \Let@ \restore@math@cr \default@tag
    \baselineskip\fontdimen10 \scriptfont\tw@
    \advance\baselineskip\fontdimen12 \scriptfont\tw@
    \lineskip\thr@@\fontdimen8 \scriptfont\thr@@
    \lineskiplimit\lineskip
    \ialign{\hfil$\m@th\scriptstyle##$&$\m@th\scriptstyle{}##$\crcr
      #1\crcr
    }%
  }
}
\makeatother

    \everymath{\displaystyle}

\begin{document} 

  \title{Constrained Expressions and their Derivatives}
  
  \author{
    Jean-Marc Champarnaud \and
    Ludovic Mignot
    \and Florent Nicart
  } 

  \institute{
    LITIS, Universit\'e de Rouen, 76801 Saint-\'Etienne du Rouvray Cedex, France\\
     \email{\{jean-marc.champarnaud,ludovic.mignot,florent.nicart\}@univ-rouen.fr}\\
  }
  
  \maketitle
  
  
  \begin{abstract} 
This paper proposes an extension to classical regular expressions by the addition of two operators allowing the inclusion of boolean formulae from the zeroth order logic. These expressions are called \emph{constrained expressions}. 
The associated language is defined thanks to the notion of interpretation and of realization.

We show that the language associated when both interpretation and realization are fixed is stricly regular and can be not regular otherwise.

    Furthermore, we use an extension of Antimirov partial derivatives in order to solve the membership test in the general case.
Finally, we show that once the interpretation is fixed, the membership test of a word in the language denoted by a constrained expression can be undecidable whereas it is always decidable when the interpretation is not fixed. 
  \end{abstract}

\section{Introduction}\label{se:int}

	Regular expressions are a convenient formalism to denote in a finite and concise way regular languages that are potentially infinite. Based on only three simple operators (sum, catenation and iteration), they are extremely easy to manipulate and are widely used in numerous domains, such as pattern matching, specification or schema validation. However, their expressive power is restricted to the class of regular languages and many attempts have been made to extend the class of their denoted languages while trying to keep their simplicity.
  
  Several approaches exist in order to make this expressive power larger; \emph{e.g.} by adding new operators~\cite{CDJM13} or by modifying the way symbols are combined~\cite{Brz68,Semp00}. In the latter case, the expressive power of the so-called regular-like expressions is increased to the linear languages, that is, a strict subclass of context-free languages in the Chomsky hierarchy~\cite{Cho56}.
  Other concepts have been added to expressions, such as the mechanism of capturing variables~\cite{CSY03,SL12}.
  
Our approach, although it is based on the introduction of two new operators, is quite different
These two operators, respectively $\mid$ and $\dashv$, establish a link with the first order logic without quantifiers (a.k.a. zeroth order logic), allowing us to easily describe non-regular languages, using both predicates and variables in order to evaluate  what we call \emph{constrained expressions}.

Given an expression $E$ and a boolean formula $\phi$, we define the expression $E\mid \phi$ ($E$ such that $\phi$) that denotes $L(E)$ when $\phi$ is satisfied and the empty set otherwise. Given a word $\alpha$, based on symbols and variables, and an expression $E$, we define the expression $\alpha \dashv E$ denoting $\alpha$ if it is in $L(E)$ and the empty set otherwise. The addition of these two operators allows us to go beyond regular languages.
Also, it turns out that constrained expressions allow us to implement  the concept of comprehension over regular expressions. Recall that the comprehension axiom can be stated as follows: for any set $A$, for any property $\phi$, there exists a set $B$ defined for any element $x$ by $x\in B$ $\Leftrightarrow$ $x\in A \wedge \phi(x)$.

  In our formalism, variables are used as a combination of the following two concepts.
  In~\cite{CSY03}, variables are used to formalize the notion of practical regular expressions (a.k.a. regex) with backreferences: in addition to the classical operators of regular expressions, variables can be used to submatch some parts of an expression. As an example, in the expression $E=(a^*)b\backslash 1$, the variable $\backslash 1$ is interpreted as a copy of the match of the expression $a^*$; therefore the language denoted by $E$ is the set $\{a^nba^n\mid n\in\mathbb{N}\}$, that is, not a regular language.
  In~\cite{SL12}, variables are used not to extend the representation power of expressions (since the denoted languages are still regular), but to efficiently solve the submatching problem, that is, to split a word according to the subexpressions of an expression it is denoted by.
  Our formalism is a different application of the concept of variables: we match particular subexpressions that we can repeat, we filter \emph{a posteriori} and we obtain languages which are not necessarily regular. As an example, let us consider the expression $((x\dashv a^*)\cdot (y\dashv b^*)\cdot (z\dashv c^*) \mid P(x,y,z) )$. This expression denotes the set of the words that can be written $w_1 w_2 w_3$ such that $w_1$ (resp. $w_2$, $w_3$) is only made of $a$ (resp. $b$, $c$) and such that the words $w_1$, $w_2$ and $w_3$ satisfy a property $P$. As an example, if $P$ is the property "$w_1$, $w_2$ and $w_3$ admit the same length", then the language (under this interpretation) is the set $\{a^nb^nc^n\mid n\in\mathbb{N}\}$.

  The aim of this paper is to define these new operators, and to show how to interpret them. We also show how to solve the membership problem, that is, to determine whether a given word belongs to the language denoted by a given constrained expression. 
  In order to perform this membership test, we present a mechanism that first associates to the variable some subword of the word to be matched, and then evaluates the boolean formula that may appear during the run.
  
  The membership test for classical regular expressions can be performed \emph{via} the computation of a finite state machine, an automaton. However, we cannot apply this technique here, since we deal with non-regular languages and therefore with infinite state machines. Nevertheless, we apply a well-known method, the expression derivation~\cite{Ant96,Brz64}, in order to reduce the membership problem of any word to the membership test of the empty word. Once this reduction made, we study the decidability of this problem.
  
  Section~\ref{se:pre} is a preliminary section where  we recall some basic definitions of formal language theory, such as languages or expression derivation. We also introduce the notion of zeroth-order logic that we use in the rest of the paper.
  Section~\ref{sec:cons expr lang and der} defines the constrained expressions and the different languages they may denote. We also define in this section the way we derive them.
  Section~\ref{sec:eps membership test} is devoted to illustrating the link between the empty word membership test and a satisfiability problem.
  In Section~\ref{sec:decidab}, we show that the satisfiability problem we use is decidable in the general case and that it is not in a particular subclass of our evaluations.

\section{Preliminaries}\label{se:pre}

  \subsection{Languages and Expressions}
  
  Let $\Sigma$ be an alphabet. We denote by $\Sigma^*$ the free monoid generated by $\Sigma$ with $\cdot$ the catenation product and $\varepsilon$ its identity. Any element in $\Sigma^*$ is called a \emph{word} and $\varepsilon$ the empty word. A \emph{language over} $\Sigma$ is a subset of $\Sigma^*$. 
  
  A language over $\Sigma$ is \emph{regular} if and only if it belongs to the family $\mathrm{Reg}(\Sigma)$ which is the smallest family containing all the subsets of $\Sigma$ and closed under the following three operations:
  \begin{itemize}
    \item \emph{union}: $L\cup L'=\{w\in\Sigma^*\mid w\in L\vee w\in L'\}$,
    \item \emph{catenation}: $L\cdot L'=\{w\cdot w'\in\Sigma^*\mid w\in L\wedge w'\in L'\}$,
    \item \emph{Kleene star}: $L^*=\{w_1\cdots w_k\in\Sigma^*\mid k\geq 0 \wedge \forall 1\leq j\leq k,\ w_j\in L\}$.
  \end{itemize}
  
  A \emph{regular expression} $E$ \emph{over} $\Sigma$ is inductively defined as follows:
  \begin{align*}
    E&=a, & E&=\varepsilon,& E&=\emptyset,\\    
    E&=(E_1)+(E_2),& E&=(E_1)\cdot (E_2),& E&=(E_1)^*,
  \end{align*}
  where $a$ is any symbol in $\Sigma$ and $E_1$ and $E_2$ are any two regular expressions over $\Sigma$. Parentheses can be omitted when there is no ambiguity.  The \emph{language denoted by} $E$ is the language $L(E)$ inductively defined by:
  \begin{align*}
    L(a)&=\{a\},& L(\varepsilon)&=\{\varepsilon\},& L(\emptyset)&=\emptyset,\\
    L(E_1+E_2)&=L(E_1)\cup L(E_2),& L(E_1\cdot E_2)&=L(E_1)\cdot L(E_2),& L(E_1^*)&=(L(E_1))^*,
  \end{align*}
  where $a$ is any symbol in $\Sigma$ and $E_1$ and $E_2$ are any two regular expressions over $\Sigma$. It is well known that language denoted by a regular expression is regular.
  
  Given a word $w$ and a language $L$, the \emph{membership problem} is the problem defined by "Does $w$ belong to $L$". Many methods exist in order to solve this problem: 
  as far as regular languages, given by regular expressions, are concerned, a finite state machine, called an \emph{automaton}, can be constructed with a polynomial time complexity w.r.t. the size of the expression, that can decide with a polynomial time complexity w.r.t. the size of the expression if a word $w$ belongs to the language denoted by the expression~\cite{Glu60,IY03,MY60,Tho68}. See~\cite{GH14} for an exhaustive study of these constructions and of their descriptional complexities.
  
  As far as regular expressions are concerned, the computation of a whole automaton is not necessary; the very structure of regular expressions is sufficient. 
  Considering the residual $w^{-1}(L)$ of $L$ w.r.t. $w$, that is, the set
  $\{w'\in\Sigma^*\mid ww'\in L\}$, the membership test $w\in L$ is equivalent to the membership test $\varepsilon \in w^{-1}(L)$. This operation of quotient can be performed directly through regular expressions using the \emph{partial derivation}~\cite{Ant96}, which is an extension of the derivation~\cite{Brz64}.
  
  \begin{definition}[\cite{Ant96}]
    Let $E$ be a regular expression over an alphabet $\Sigma$. The \emph{partial derivative of} $E$ w.r.t. a word $w$ in $\Sigma^*$ is the set $\partial_w(E)$ inductively defined as follows:  
  \begin{align*}
      \partial_w(E)=
        \begin{cases}
            \{E\} & \text{ if } w=\varepsilon,\\
            \{\varepsilon\} & \text{ if } E=w=a\in\Sigma,\\
            \emptyset & \text{ if } w=a\in\Sigma\wedge E\in\Sigma\setminus\{a\}\cup\{\emptyset,\varepsilon\},\\
            \partial_w(F)\cup \partial_w(G) & \text{ if } w=a\in\Sigma\wedge E=F+G,\\
            \partial_w(F)\odot G \cup (\partial_w(G)\mid \varepsilon\in L(F)) & \text{ if }w=a\in\Sigma\wedge E=F\cdot G,\\
            \partial_w(F)\odot F^* & \text{ if } w=a\in\Sigma\wedge E=F^*,\\
            \partial_{w'}(\partial_a(E)) & \text{ if } w=aw'\wedge a\in\Sigma\wedge w'\in\Sigma^*,
        \end{cases}
  \end{align*}
  where $F$ and $G$ are any two regular expressions over $\Sigma$ and where for any set $\mathcal{E}$ of regular expressions, for any regular expression $E'$, $\mathcal{E}\odot E'=\bigcup_{E\in\mathcal{E}} \{E\cdot E'\}$ and for any word $w'$ in $\Sigma^*$, $\partial_w(\mathcal{E})=\bigcup_{E\in\mathcal{E}}\partial_w(E)$.
  \end{definition}
  
  Furthermore, the membership test of $\varepsilon$ is syntactically computable. Considering the predicate $\mathrm{Null}(E)=(\varepsilon\in L(E))$, it can be checked that:
  \begin{align*}
      \mathrm{Null}(E)=
        \begin{cases}
            1 & \text{ if } E=\varepsilon,\\
            0 & \text{ if } E\in\Sigma\cup\{\emptyset\},\\
            \mathrm{Null}(F)\vee \mathrm{Null}(G) & \text{ if } E=F+G,\\
            \mathrm{Null}(F)\wedge \mathrm{Null}(G) & \text{ if }E=F\cdot G,\\
            1 & \text{ if } E=F^*,
         \end{cases}   
  \end{align*}
  where $F$ and $G$ are any two regular expressions over $\Sigma$.
    
  Consequently, from these definitions, the membership test of $w$ in $L(E)$ can be performed as follows:
  
  \begin{proposition}[\cite{Ant96}]
    Let $E$ be a regular expression over an alphabet $\Sigma$ and $w$ be a word in $\Sigma^*$. The following two conditions are equivalent:
    \begin{enumerate}
      \item $w\in L(E)$
      \item $\exists E'\in \partial_w(E)$, $\varepsilon\in L(E')$.
    \end{enumerate}
  \end{proposition}
  
  Derivation and  partial derivation have already been used in order to perform the membership test over extensions of regular expressions~\cite{CCM11b,CCM12c,CCM14,CJM13}, expressions denoting non-necessarily regular languages~\cite{CDJM13}, guarded strings~\cite{ABM12} or even context-free grammars~\cite{MDS11}.  
  In the rest of this paper, 
  we extend regular expressions by introducing new operators based on boolean formulae  
  in order to increase the expressive  power of expressions. 
  Let us first recall some well-known definitions of logic.

  \subsection{Zeroth-Order Logic}
  
  The notion of constrained expression introduced in this paper is expressed through the formalism of zeroth-order logic, that is, first order logic without quantifiers (see primitive recursive arithmetic in~\cite{Sko67} for an example of the difference between the expressiveness of propositional logic and zeroth-order logic).
  
  More precisely, we consider two $\mathbb{N}$-indexed families $\mathcal{F}=(\mathcal{F}_k)_{k\in\mathbb{N}}$ and $\mathcal{P}=(\mathcal{P}_k)_{k\in\mathbb{N}}$ of disjoint sets, where for any integer $k$ in $\mathbb{N}$, $\mathcal{F}_k$ is a set of $k$-ary function symbols and $\mathcal{P}_k$ is a set of $k$-ary predicate symbols. The family $\mathcal{F}$  is combined with a set of variables in order to obtain a set of terms. This set of terms is combined with the family $\mathcal{P}$ in order to obtain boolean formulae.
  
  Given a set $X$ of variables, a \emph{term} $t$ \emph{over} $(\mathcal{F},X)$ is inductively defined by:
  \begin{align*}
        t & =x \text{ with } x\in X,\\
        t & =\mathrm{f}(t_1,\ldots,t_k),
  \end{align*}
  where $k$ is any integer, $\mathrm{f}$ is any element in $\mathcal{F}_k$, $t_1,\ldots,t_k$ are any $k$ terms over $(\mathcal{F},X)$. 
    We denote by $\mathcal{F}(X)$ the set of the terms over $(\mathcal{F},X)$.
    
    A \emph{subterm} of a term $t$ is a term in the set $\mathrm{Subterm}(t)$ inductively computed as follows:
  \begin{align*}
        \mathrm{Subterm}(x) & =\{x\},\\
        \mathrm{Subterm}(\mathrm{f}(t_1,\ldots,t_k)) & =\{\mathrm{f}(t_1,\ldots,t_k)\}\cup \bigcup_{j\in\{1,\ldots,k\}}\mathrm{Subterm}(t_j),
  \end{align*}
  where $x$ is any element in $X$, $k$ is any integer, $\mathrm{f}$ is any function symbol in $\mathcal{F}_k$ and $t_1$, $\ldots$, $t_k$ are any $k$ terms in $\mathcal{F}(X)$.
     
    A \emph{boolean formula} $\phi$ \emph{over} $(\mathcal{P},\mathcal{F}(X))$ is inductively defined by:
  \begin{align*}
        \phi & =P(t_1,\ldots,t_k),\\
        \phi & =\mathrm{o}(\phi_1,\ldots,\phi_{k'}),
  \end{align*}
  where $k$ and $k'$ are any two integers, $P$ is any element in $\mathcal{P}_k$, $t_1,\ldots,t_k$ are any $k$ terms in $\mathcal{F}(X)$, $\mathrm{o}$ is any
    $k'$-ary boolean operator associated with a mapping from $\{0,1\}^{k'}$ to $\{0,1\}$
  and $\phi_1,\ldots,\phi_{k'}$ are any $k'$ boolean formulae over $(\mathcal{P},X)$. We denote by $\mathcal{P}(\mathcal{F}(X))$ the set of boolean formulae over $(\mathcal{P},\mathcal{F}(X))$.
  
  Given a formula $\phi$ in $\mathcal{P}(\mathcal{F}(X))$, a term $t$ in $\mathcal{F}(X)$ and a symbol $x$ in $X$, we denote by $\phi_{x\leftarrow t}$ the \emph{substitution of} $x$ by $t$ in $\phi$, which is the boolean formula inductively defined by:
  \begin{align*}
        (\mathrm{o}(\phi_1,\ldots,\phi_{k'}))_{x\leftarrow t} & =\mathrm{o}((\phi_1)_{x\leftarrow t},\ldots,(\phi_{k'})_{x\leftarrow t}),\\
        (P(t_1,\ldots,t_k)_{x\leftarrow t}) & =P((t_1)_{x\leftarrow t},\ldots,(t_k)_{x\leftarrow t}),
  \end{align*}
  where $k$ and $k'$ are any two integers, $P$ is any element in $\mathcal{P}_k$, $t_1,\ldots,t_k$ are any $k$ terms in $\mathcal{F}(X)$, $\mathrm{o}$ is any
    $k'$-ary boolean operator associated with a mapping from $\{0,1\}^{k'}$ to $\{0,1\}$, $\phi_1,\ldots,\phi_{k'}$ are any $k'$ boolean formulae over $(\mathcal{P},X)$ and where for any term $t'$ in $\mathcal{F}(X)$, $t'_{x\leftarrow t}$ is the \emph{substitution of} $x$ by $t$ in $t'$, which is the term inductively defined by:
  \begin{align*}
        (\mathrm{f}(t_1,\ldots,t_k))_{x\leftarrow t} & =\mathrm{f}((t_1)_{x\leftarrow t},\ldots,(t_k)_{x\leftarrow t}),\\
        y_{x\leftarrow t} & =
          \begin{cases}
      y & \text{ if }x\neq y,\\
      t & \text{ otherwise,}
      \end{cases}
  \end{align*}
  where $y$ is any element in $X$, $k$ is any integer, $\mathrm{f}$ is any function symbol in $\mathcal{F}_k$ and $t_1$, $\ldots$, $t_k$ are any $k$ terms in $\mathcal{F}(X)$.
  
  \begin{example}\label{ex formule bool}
    Let us consider the two families $\mathcal{F}=(\mathcal{F}_k)_{k\in\mathbb{N}}$ and $\mathcal{P}=(\mathcal{P}_k)_{k\in\mathbb{N}}$ defined by
  \begin{align*}
        \mathcal{F}_k&=
          \begin{cases}
            \{f\} & \text{ if }k=1,\\
            \{g,h\} & \text{ if }k=2,\\
            \emptyset & \text{ otherwise}
          \end{cases},&
        \mathcal{P}_k&=
          \begin{cases}
            \{P\} & \text{ if }k=1,\\
            \{Q,R\} & \text{ if }k=2,\\
            \emptyset & \text{ otherwise,}
          \end{cases}
  \end{align*}
  and the set $X=\{x,y,z\}$.        
        As an example, the set of boolean formulae over $(\mathcal{P},\mathcal{F}(X))$ contains the formulae:
    \begin{itemize}
      \item $\phi_1=P(x) \vee Q(x,z)$,
      \item $\phi_2=Q(y,f(x)) \wedge R(h(z,z),g(f(y),z))$,
      \item $\phi_3=\neg P(f(g(x,x)))$,
      \item $\phi_4=\underline{?:} (\phi_1,\phi_2,\phi_3)$, where $\underline{?:}$ is the ternary  boolean operator corresponding to the \emph{If-Then-Else}-like conditional expression, generally written $(\phi_1\ ?\ \phi_2\ :\ \phi_3)$.
    \end{itemize}
    \qed
  \end{example}
  
  After having defined the syntactic part of the logic formulae we use, we show how to evaluate these formulae, \emph{i.e.} how to define the semantics of the logic formulae. The boolean evaluation of a formula is performed in two steps. First, an \emph{interpretation} defines a domain and associates the function and predicate symbols with functions; then, each variable symbol is associated with a value from the domain by a \emph{realization}. 
  
  \begin{definition}[Interpretation]\label{def interpretation}
    Let $\mathcal{F}=(\mathcal{F}_k)_{k\in\mathbb{N}}$ and $\mathcal{P}=(\mathcal{P}_k)_{k\in\mathbb{N}}$ be two families of disjoint sets. An \emph{interpretation} $I$ over $(\mathcal{P},\mathcal{F})$ is a tuple $(\mathfrak{D},\mathfrak{F})$ where:
    \begin{itemize}
      \item $\mathfrak{D}$ is a set, called the \emph{interpretation domain} of $I$,
      \item $\mathfrak{F}$ is a function:
        \begin{itemize}
          \item from $\mathcal{P}_k$ to $2^{\mathfrak{D}^k}$,
          \item and from $\mathcal{F}_k$ to $2^{\mathfrak{D}^{k+1}}$ such that for any function symbol $\mathrm{f}$ in $\mathcal{F}_k$, for any two elements $(e_1,\ldots,e_k,e_{k+1})$ and $(e'_1,\ldots,e'_k,e'_{k+1})$ in $\mathfrak{F}(\mathrm{f})$, $(e_1,\ldots,e_k)=(e'_1,\ldots,e'_k)$ $\Rightarrow$ $e_{k+1}=e'_{k+1}$, and such that for any $k$ elements $(e_1,\ldots,e_k)$ in $\mathfrak{D}^k$, there exists $e_{k+1}$ in $\mathfrak{D}$ such that $(e_1,\ldots,e_k,e_{k+1})\in \mathfrak{F}(\mathrm{f})$  
        \end{itemize}
        called the \emph{interpretation function}.
    \end{itemize}
  \end{definition}
  
  \begin{definition}[Realization]
    Let $\mathcal{P}$ and $\mathcal{F}$ be two families of disjoint sets, and $I=(\mathfrak{D},\mathfrak{F})$ an interpretation over $(\mathcal{P},\mathcal{F})$. Let $X$ be a set. An $X$-\emph{realization} $\mathrm{r}$ over $I$ is a function from $X$ to $\mathfrak{D}$.
  \end{definition}
  
  Once an interpretation $I$ and a realization $\mathrm{r}$ given, a term can be evaluated as an element of the domain and a formula as a boolean \emph{via} the function $\mathrm{eval}_{(I,\mathrm{r})}$, the $(I,\mathrm{r})$\emph{-evaluation}:
  
  \begin{definition}[Term Evaluation]
    Let $\mathcal{P}$ and $\mathcal{F}$ be two families of disjoint sets and $I=(\mathfrak{D},\mathfrak{F})$ an interpretation over $(\mathcal{P},\mathcal{F})$. Let $X$ be a set. Let $\mathrm{r}$ be an $X$-realization over $I$. Let $t$ be a term in $\mathcal{F}(X)$. The $(I,\mathrm{r})$-\emph{evaluation} of $t$ is the element $\mathrm{eval}_{I,\mathrm{r}}(t)$ in $\mathfrak{D}$ defined by:
    \begin{align*}
      \mathrm{eval}_{I,\mathrm{r}}(t)&=
        \begin{cases}
            \mathrm{r}(x) & \text{ if }t=x\wedge x\in X,\\
            x_{k+1} & \text{ if }t=\mathrm{f}(t_1,\ldots,t_k)\ \wedge\ (\mathrm{eval}_{I,\mathrm{r}}(t_1),\ldots,\mathrm{eval}_{I,\mathrm{r}}(t_k),x_{k+1})\in \mathfrak{F}(\mathrm{f})\\
          \end{cases}
    \end{align*}
    where $k$ is any integer, $\mathrm{f}$ is any function symbol in $\mathcal{F}_k$, and $t_1,\ldots,t_k$ are any $k$ elements in $\mathcal{F}(X)$.
  \end{definition}
  
  \begin{definition}[Formula Evaluation]
    Let $\mathcal{P}$ and $\mathcal{F}$ be two families of disjoint sets and $I=(\mathfrak{D},\mathfrak{F})$ an interpretation over $(\mathcal{P},\mathcal{F})$. Let $X$ be a set. Let $\mathrm{r}$ be an $X$-realization over $I$. Let $t$ be a term in $\mathcal{F}(X)$. Let $\phi$ be a boolean formula in $\mathcal{P}(\mathcal{F}(X))$. The $(I,\mathrm{r})$-\emph{evaluation} of $\phi$ is the boolean $\mathrm{eval}_{I,\mathrm{r}}(\phi)$ inductively defined by:
  \begin{align*}
        \mathrm{eval}_{I,\mathrm{r}}(P(t_1,\ldots,t_k)) & =
          \begin{cases}
            1 & \text{ if } (\mathrm{eval}_{I,\mathrm{r}}(t_1),\ldots,\mathrm{eval}_{I,\mathrm{r}}(t_k)) \in \mathfrak{F}(P),\\
            0 & \text{ otherwise,}
          \end{cases}\\
        \mathrm{eval}_{I,\mathrm{r}}(\mathrm{o}(\phi_1,\ldots,\phi_k)) & =\mathrm{o}'(\mathrm{eval}_{I,\mathrm{r}}(\phi_1),\ldots,\mathrm{eval}_{I,\mathrm{r}}(\phi_k)),
  \end{align*}
  where $k$ is any integer, $P$ is any predicate symbol in $\mathcal{P}_k$, $t_1,\ldots,t_k$ are any $k$ elements in $\mathcal{F}(X)$, $\mathrm{o}$ is any $k$-ary boolean operator associated with a mapping $\mathrm{o}'$ from $\{0,1\}^k$ to $\{0,1\}$ and $\phi_1,\ldots,\phi_k$ are any $k$ boolean formulae over $(\mathcal{P},X)$.
  \end{definition}

  \begin{example}\label{ex inter real}
    Let us consider Example~\ref{ex formule bool}. Let $\Sigma=\{a,b\}$ be an alphabet. Let $I$ be the interpretation $(\mathfrak{D},\mathfrak{F})$ over $(\mathcal{P},\mathcal{F})$ and $\mathrm{r}$ be the $X$-realization over $I$ defined by:
  \begin{align*}
          \mathfrak{D} & =\Sigma^*, & \mathfrak{F}(g) & =\{(w_1,w_2,w_1\cdot w_2)\in \mathfrak{D}^3\},\\
          \mathfrak{F}(P) & =\{w\in \mathfrak{D}\mid w\neq\varepsilon\}, & \mathfrak{F}(h) & =\{(w_1,w_2,w_1w_1w_2)\in \mathfrak{D}^3\},\\
          \mathfrak{F}(Q) & =\{(w_1,w_2)\in \mathfrak{D}^2\mid w_1=w_2\},&\mathrm{r}(x) & =aa,\\
          \mathfrak{F}(R) & =\{(w_1,w_2)\in \mathfrak{D}^2\mid w_1=\mathrm{rev}(w_2)\},&\mathrm{r}(y) & =bb,\\
          \mathfrak{F}(f) & =\{(w_1,\mathrm{rev}(w_1))\in \mathfrak{D}^2\},&\mathrm{r}(z) & =\varepsilon,
  \end{align*}
  where for any word $w$ in $\Sigma^*$, $\mathrm{rev}(w)$ is the word defined by:
  \begin{align*}
      \mathrm{rev}(w)=
        \begin{cases}
            w & \text{ if } w=\varepsilon,\\
            \mathrm{rev}(w')x & \text{ if }w=xw'\wedge w'\in\Sigma^*\wedge x\in\Sigma.
          \end{cases}
  \end{align*}
  Then:    
    \begin{align*}
      \mathrm{eval}_{(I,\mathrm{r})}(\phi_1) & =\mathrm{eval}_{(I,\mathrm{r})}(P(x) \vee Q(x,z))\\
      & =(aa\neq\varepsilon)\ \mathrm{Or}\ (aa==\varepsilon)\\
      & =1\\
      \mathrm{eval}_{(I,\mathrm{r})}(\phi_2) & =\mathrm{eval}_{(I,\mathrm{r})}(Q(y,f(x)) \wedge R(h(z,z),g(f(y),z))\\
      & =(bb==aa)\ \mathrm{And}\ (\varepsilon == \mathrm{rev}(bb))\\
      & =0\\      
      \mathrm{eval}_{(I,\mathrm{r})}(\phi_3) & =\mathrm{eval}_{(I,\mathrm{r})}(\neg P(f(g(x,x))))\\
      & =\mathrm{Not}\ (aaaa\neq\varepsilon)\\
      & =0\\      
      \mathrm{eval}_{(I,\mathrm{r})}(\phi_4) & =(\mathrm{eval}_{(I,\mathrm{r})}(\phi_1)\ \mathrm{Implies}\ \mathrm{eval}_{(I,\mathrm{r})}(\phi_2))\\
      & \ \ \ \ \ \mathrm{And}\ (\mathrm{Not}\ (\mathrm{eval}_{(I,\mathrm{r})}(\phi_1))\ \mathrm{Implies}\ \mathrm{eval}_{(I,\mathrm{r})}(\phi_3))\\
      & =(1\ \mathrm{Implies}\  0)\ \mathrm{And}\ (0\ \mathrm{Implies}\ 0)\\
      & =0
    \end{align*}
    \qed       
  \end{example}

\section{Constrained Expressions, their Languages and Derivatives}\label{sec:cons expr lang and der}

  In this section, zeroth-order logic is combined with classical regular expressions in order to define \emph{constrained expressions}. The language denoted by these expressions is not necessarily regular. We extend the membership problem for constrained expressions using partial derivatives and then show that it is equivalent to a satisfiability problem.

\subsection{Constrained Expressions and their Languages}

  Whereas regular expressions are defined over a unique symbol alphabet, constrained expressions deal with zeroth-order logic and therefore include function, predicate and variable symbols. Hence the notion of alphabet is extended to the notion of \emph{expression environment} in order to take into account
all these symbols.
  
  \begin{definition}[Expression Environment]
    An \emph{expression environment} is a $4$-tuple $(\Sigma,\Gamma,\mathcal{P},\mathcal{F})$ where:
    \begin{itemize}
      \item $\Sigma$ is an alphabet, called the \emph{symbol alphabet},
      \item $\Gamma$ is an alphabet, called the \emph{variable alphabet},
      \item $\mathcal{P}$ is a $\mathbb{N}$-indexed family of disjoint sets, called the \emph{family of predicate symbols},
      \item $\mathcal{F}$ is a $\mathbb{N}$-indexed family of disjoint sets, called the \emph{family of function symbols} such that $\Sigma\cup\{\varepsilon\}\subset\mathcal{F}_0$ and $\{\cdot\} \in \mathcal{F}_2$.
    \end{itemize}
  \end{definition}
  
  Once this environment stated, we can syntactically define the set of \emph{constrained expressions}, by adding two new operators to regular operators: the first operator, $\mid$, is based on the combination of an expression $E$ and of a boolean formula $\phi$, producing the expression $E\mid \phi$; the second operator, $\dashv$, links a word $\alpha$ composed of variable and letter symbols to an expression $E$, producing the expression $\alpha \dashv E$.
  Notice that the following definitions use extended boolean operators, such as intersection or negation. However, we will extend the membership test for expressions only using the sum operator.

  \begin{definition}[Constrained Expression]
    Let $\mathcal{E}=(\Sigma,\Gamma,\mathcal{P},\mathcal{F})$ be an expression environment. A \emph{constrained expression} $E$ \emph{over} $\mathcal{E}$ is inductively defined by:
    \begin{align*}
      E&=\alpha,&&& E&=\emptyset,\\
      E&=\mathrm{o}(E_1,\ldots,E_k),& E&=(E_1)\cdot (E_2),& E&=(E_1)^*,\\
      E&=(E_1)\mid(\phi),&&& E&=(\alpha) \dashv (E_1),
    \end{align*}
    where $k$ is any integer, $\mathrm{o}$ is any $k$-ary boolean operator associated with a mapping from $\{0,1\}^k$ to $\{0,1\}$, $E_1,\ldots,E_k$ are any $k$ constrained expressions over $\mathcal{E}$, $\alpha$ is any word in $(\Sigma\cup \Gamma)^*$ and $\phi$ is a boolean formula in $\mathcal{P}(\mathcal{F}(\Gamma))$.
  \end{definition}  
\noindent  
  Parenthesis can be omitted when there is no ambiguity.
  
  Any boolean formula that appears in a constrained expression over an environment $\mathcal{E}=(\Sigma,\Gamma,\mathcal{P},\mathcal{F})$ is, by definition, a formula in $\mathcal{P}(\mathcal{F}(\Gamma))$. Furthermore, since we want a constrained expression to denote  a subset of $\Sigma^*$, variable symbols in $\Gamma$ have to be evaluated as words in $\Sigma^*$. Moreover, classical symbols, like $\varepsilon$ or $a$ in $\Sigma$, have to be considered as $0$-ary functions in the interpretation. All these considerations imply some specializations of the notions of
interpretation and realization, defined as follows.
  
  \begin{definition}[Expression Interpretation]
    Let $\mathcal{E}=(\Sigma,\Gamma,\mathcal{P},\mathcal{F})$ be an expression environment. An \emph{expression interpretation} $I$ over $\mathcal{E}$ is an interpretation $(\mathfrak{D},\mathfrak{F})$ over $(\mathcal{P},\mathcal{F})$ satisfying the following three conditions:
    \begin{enumerate}
      \item $\mathfrak{D}=\Sigma^*$,
      \item for any symbol $\alpha$ in $\Sigma\cup\{\varepsilon\}\subset \mathcal{F}_0$, $\mathfrak{F}(\alpha)=\{\alpha\}$, 
      \item $\mathfrak{F}(\cdot)=\{(u,v,w)\in(\Sigma^*)^3\mid w=uv\}$.
    \end{enumerate}
  \end{definition} 
  
  The two new operators appearing in a constrained expression are used to extend the expressive power of regular expressions. The expression $E\mid\phi$ denotes the set of words that $E$ may denote whenever the formula $\phi$ is satisfied. The expression $\alpha \dashv E$ denotes the set of words that $E$ may denote and that can be "matched" by $\alpha$. In order to perform this matching, we extend any $\Gamma$-realization over an expression interpretation as a morphism from $(\Sigma\cup\Gamma)^*$ to $\Sigma^*$. 
  
    Let $\mathcal{E}=(\Sigma,\Gamma,\mathcal{P},\mathcal{F})$ be an expression environment. Let $I$ be an expression interpretation over $\mathcal{E}$ and $\mathrm{r}$ be a $\Gamma$-realization over $I$. The domain of the realization $\mathrm{r}$ can be extended to $(\Sigma\cup\Gamma)^*$ as follows. For any word $\alpha$ in $(\Sigma\cup\Gamma)^*$, $\mathrm{r}(\alpha)$ is the word in $\Sigma^*$ inductively computed by:
      \begin{align*}
      \mathrm{r}(\alpha)=
        \begin{cases}
          \varepsilon & \text{if } \alpha=\varepsilon,\\
          a\mathrm{r}(\alpha') & \text{if } \alpha=a\alpha'\ \wedge a\in \Sigma,\\
          \mathrm{r}(x)\mathrm{r}(\mathrm{\alpha'}) & \text{ if } \alpha=x\alpha' \ \wedge\ x\in\Gamma.
        \end{cases}
      \end{align*}
    
  Using this extension, we can now formally define the different languages that a constrained expression may denote. We consider the following three cases where first both the interpretation and
the realization are fixed, then only the interpretation is fixed
and finally nothing is fixed.
  
  \begin{definition}[(I,r)-Language]\label{def i r lang}
    Let $\mathcal{E}=(\Sigma,\Gamma,\mathcal{P},\mathcal{F})$ be an expression environment. Let $I$ be an expression interpretation over $\mathcal{E}$ and $\mathrm{r}$ be a $\Gamma$-realization over $I$. Let $E$ be a constrained expression over $\mathcal{E}$.
    The $(I,\mathrm{r})$-\emph{language denoted by} $E$ is the language $L_{I,\mathrm{r}}(E)$ inductively defined by:
    \begin{align*}
      L_{I,\mathrm{r}}(\alpha)&=\{\mathrm{r}(\alpha)\},\\
      L_{I,\mathrm{r}}(\emptyset)&=\emptyset,\\
      L_{I,\mathrm{r}}(\alpha\dashv E_1)&=\{\mathrm{r}(\alpha) \mid \mathrm{r}(\alpha)\in L_{I,\mathrm{r}}(E_1)\},\\
      L_{I,\mathrm{r}}(\mathrm{o}(E_1,\ldots,E_k))&=\mathrm{o}'(L_{I,\mathrm{r}}(E_1),\ldots,L_{I,\mathrm{r}}(E_k)),\\
      L_{I,\mathrm{r}}(E_1\cdot E_2)&=L_{I,\mathrm{r}}(E_1)\cdot L_{I,\mathrm{r}}(E_2),\\
      L_{I,\mathrm{r}}(E_1^*)&=(L_{I,\mathrm{r}}(E_1))^*,\\
      L_{I,\mathrm{r}}(E_1 \mid \phi)&=
        \begin{cases}
            L_{I,\mathrm{r}}(E_1) & \text{ if } \mathrm{eval}_{(I,\mathrm{r})}(\phi),\\
            \emptyset & \text{otherwise,}
        \end{cases}
    \end{align*}
    where $k$ is any integer, $\mathrm{o}$ is any $k$-ary boolean operator, $\mathrm{o}'$ is the language operator associated with $\mathrm{o}$, $E_1,\ldots,E_k$ are any $k$ constrained expression over $\mathcal{E}$, $\alpha$ is any word in $(\Sigma\cup \Gamma)^*$ and $\phi$ is any boolean formula in $\mathcal{P}(\mathcal{F}(\Gamma))$.
  \end{definition}
  
  We denote by $\mathrm{Real}_\Gamma(I)$ the set of the $\Gamma$-realizations over an interpretation $I$.
    
  \begin{definition}[I-Language]
    Let $\mathcal{E}=(\Sigma,\Gamma,\mathcal{P},\mathcal{F})$ be an expression environment. Let $E$ be a constrained expression over $\mathcal{E}$. Let $I$ be an expression interpretation over $\mathcal{E}$. The $I$-\emph{language denoted by} $E$ is the language $L_{I}(E)$ defined by:
    \begin{align*}
      L_{I}(E)=\bigcup_{\mathrm{r}\in \mathrm{Real}_\Gamma(I)}L_{I,\mathrm{r}}(E).
    \end{align*}
  \end{definition}
  
  Given an expression environment $\mathcal{E}=(\Sigma,\Gamma,\mathcal{P},\mathcal{F})$, we denote by $\mathrm{Int}(\mathcal{E})$ the set of the expression interpretations over $\mathcal{E}$.
    
  \begin{definition}[Language]
    Let $\mathcal{E}=(\Sigma,\Gamma,\mathcal{P},\mathcal{F})$ be an expression environment. Let $E$ be a constrained expression over $\mathcal{E}$. The \emph{language denoted by} $E$ is the language $L(E)$ defined by:
    \begin{align*}
      L(E)=\bigcup_{I\in \mathrm{Int}(\mathcal{E})}L_{I}(E).
    \end{align*}
  \end{definition}
  
  \begin{example}\label{ex exp cons lang}
    Let $ \mathcal{E}=(\Sigma,\Gamma,\mathcal{P},\mathcal{F})$ be the expression environment defined by:
    \begin{itemize}
      \item $\Sigma=\{a,b,c\}$,
      \item $\Gamma=\{x,y,z\}$,
      \item $\mathcal{P}=\mathcal{P}_2=\{\lessdot,\sim\}$,
      \item $\mathcal{F}_0=\Sigma\cup\{\varepsilon\}$, $\mathcal{F}_1=\{\mathrm{f}\}$, $\mathcal{F}_2=\{\cdot\}$.
    \end{itemize}    
    Let us consider the constrained expressions $E_1=xb^*y\mid \mathopen{\sim}(\mathrm{f}(x),\mathrm{f}(y))$ and $E_2= (ab)^*x \dashv (yx\mid \lessdot (y,x))$. 
    Let $\mathrm{I}=(\Sigma^*,\mathfrak{F}_1)$ be the expression interpretation defined by:
    \begin{itemize}
      \item $\mathfrak{F}(\lessdot)=\{(u,v)\mid |u|\leq |v| \}$,
      \item $\mathfrak{F}(\sim)=\{(u,u)\}$,
      \item $\mathfrak{F}(\alpha)=\{\alpha\}$, for any $\alpha$ in $\mathcal{F}_0$, 
      \item $\mathfrak{F}(\mathrm{f})=\{(u,a^{|u|_a})\}$
      \item $\mathfrak{F}(\cdot)=\{(u,v,u\cdot v)\}$.
    \end{itemize}
    
    In other words, the evaluation of an expression w.r.t. $\mathrm{I}_1$ considers that:
    \begin{itemize}
      \item $\lessdot(u,v)$ is true if and only if $u$ is shorter than $v$,
      \item $\sim(u,v)$ is true if and only if $u=v$,
      \item $\mathrm{f}$ is a function that changes all the symbols in $u$ that are different from $a$ into a symbol $\varepsilon$.
    \end{itemize}
    
    By abuse of notation, let us syntactically apply the interpretation as follows:
    \begin{align*}
      [E_1]_{\mathrm{I}}&=xb^*y\mid a^{|x|_a}=a^{|y|_a},& [E_2]_{\mathrm{I}}&=(ab)^*x \dashv (yx\mid |y|\leq |x|)
    \end{align*}
    
    Let us consider the $\mathrm{I}$-languages denoted by these expressions:
    \begin{itemize}
      \item $L_\mathrm{I}(E_1)$ is the set of words $ub^nv$ with $n\geq 0$ and $|u|_a=|v|_a$,
      \item $L_\mathrm{I}(E_2)$ is the set of words $(ab)^n u$ with $n\geq 0$ and $2n\leq |u|$. 
    \end{itemize}
    
    As an example, the word $ababbbaa$ belongs to:
    \begin{itemize}
      \item $L_\mathrm{I}(E_1)$ since it can be obtained by considering the realization $\mathrm{r}_1$ associating $aba$ with $x$ and $aa$ with $y$: 
    \begin{align*}
        [E_1]_{\mathrm{I},\mathrm{r}_1}=aba b^*aa \mid (a^{2}=a^{2}) 
    \end{align*}
that is equivalent to  $abab^*aa$,
      \item $L_\mathrm{I}(E_2)$ since it can be obtained by considering the realization $\mathrm{r}_2$ associating $bbaa$ with $x$ and $abab$ with $y$: 
    \begin{align*}
        [E_2]_{\mathrm{I},\mathrm{r}_2}=(ab)^*bbaa \dashv (ababbbaa\mid 4\leq 4)
    \end{align*}
    that is equivalent to $(ab)^*bbaa \dashv (ababbbaa)$ and finally to $ababbbaa$.
    \end{itemize}
    
    In other words, the word $ababbbaa$ belongs to $L_\mathrm{I,\mathrm{r}_1}(E_1)\subset L_\mathrm{I}(E_1)$ and to $L_\mathrm{I,\mathrm{r}_2}(E_2)\subset L_\mathrm{I}(E_2)$.
    
    Notice that the word $ababbbaa$ is not in the $(\mathrm{I},\mathrm{r}_2)$-language denoted by $[E_1]_{\mathrm{I},\mathrm{r}_2}=bbaab^*abab\mid a^{2}=a^{2}$ that is equivalent to $bbaab^*abab$ nor in the $(\mathrm{I},\mathrm{r}_1)$-language denoted by $[E_2]_{\mathrm{I},\mathrm{r}_1}=(ab)^*aba \dashv (aaaba\mid 2\leq 3)$ that is equivalent to $(ab)^*aba \cap aaaba$ and finally to $\emptyset$. 
  \end{example}
    \begin{example}\label{ex anbncn}
      Let us consider the expression environment $\mathcal{E}$ of Example~\ref{ex exp cons lang}.
      Let $E$ be the constrained expression defined as follows:
      \begin{align*}
        E &= ((x\dashv a^*)\cdot (y\dashv b^*)\cdot (z\dashv c^*))\mid \sim(x,y)\wedge \sim (y,z)
      \end{align*}
      Let us consider an expression interpretation $J=(\Sigma^*,\mathfrak{G})$ that satisfies
      \begin{align*}
        \mathfrak{G}(\sim)&=\{(w_1,w_2)\in\Sigma^*\mid |w_1|=|w_2|\}
      \end{align*}
      Then
      \begin{align*}
        L_J(E) &=\{w_aw_bw_c \in\Sigma^*\mid (\forall \alpha\in\{a,b,c\}, w_\alpha\in\alpha^*) \wedge |w_a|=|w_b|=|w_c|\}\\
               &=\{a^nb^nc^n\mid n\in\mathbb{N}\}
      \end{align*}
    \end{example}

  \subsection{The $(I,r)$-Language of a Constrained Expression is Regular}
  
  Whenever an interpretation and a realization are fixed, the language denoted by a constrained expression is a regular one. The proof is based on the computation of an equivalent regular expression.
  
  \begin{definition}[Regularization]\label{def i r reg exp}
    Let $\mathcal{E}=(\Sigma,\Gamma,\mathcal{P},\mathcal{F})$ be an expression environment and $E$ be a constrained expression over $\mathcal{E}$. Let $I$ be an expression interpretation over $\mathcal{E}$ and $\mathrm{r}$ be a $\Gamma$-realization over $I$. The $(I,\mathrm{r})$-\emph{regularization} of $E$ is the regular expression $\mathrm{reg}_{I,\mathrm{r}}(E)$ inductively defined as follows:
    \begin{align*}
      \mathrm{reg}_{I,\mathrm{r}}(\alpha)&=\mathrm{r}(\alpha),\\
      \mathrm{reg}_{I,\mathrm{r}}(\emptyset)&=\emptyset,\\
      \mathrm{reg}_{I,\mathrm{r}}(\mathrm{o}(E_1,\ldots,E_k))&=\mathrm{o}(\mathrm{reg}_{I,\mathrm{r}}(E_1),\ldots,\mathrm{reg}_{I,\mathrm{r}}(E_k)),\\
      \mathrm{reg}_{I,\mathrm{r}}(E_1\cdot E_2)&=\mathrm{reg}_{I,\mathrm{r}}(E_1)\cdot \mathrm{reg}_{I,\mathrm{r}}(E_2),\\
      \mathrm{reg}_{I,\mathrm{r}}(E_1^*)&=\mathrm{reg}_{I,\mathrm{r}}(E_1)^*,\\
      \mathrm{reg}_{I,\mathrm{r}}(E_1\mid\phi)&=
        \begin{cases}
             \mathrm{reg}_{I,\mathrm{r}}(E_1) & \text{ if } \mathrm{eval}_{I,\mathrm{r}}(\phi),\\
             \emptyset & \text{ otherwise,}
        \end{cases}\\
      \mathrm{reg}_{I,\mathrm{r}}(\alpha \dashv E_1)&= \mathrm{r}(\alpha) \cap \mathrm{reg}_{I,\mathrm{r}}(E_1),
    \end{align*}
    where $k$ is any integer, $\mathrm{o}$ is any $k$-ary boolean operator associated with a mapping from $\{0,1\}^k$ to $\{0,1\}$, $E_1,\ldots,E_k$ are any $k$ constrained expressions over $\mathcal{E}$, $\alpha$ is any word in $(\Sigma\cup \Gamma)^*$ and $\phi$ is a boolean formula in $\mathcal{P}(\mathcal{F}(\Gamma))$.  
  \end{definition}
  
  \begin{proposition}
    Let $\mathcal{E}=(\Sigma,\Gamma,\mathcal{P},\mathcal{F})$ be an expression environment and $E$ be a constrained expression over $\mathcal{E}$. Let $I$ be an expression interpretation over $\mathcal{E}$ and $\mathrm{r}$ be a $\Gamma$-realization over $I$. Then:
    \begin{align*}
      L_{I,\mathrm{r}}(E)=L(\mathrm{reg}_{I,\mathrm{r}}(E)).
    \end{align*}
  \end{proposition}
  \begin{proof}
    By induction over the structure of $E$. According to Definition~\ref{def i r reg exp} and to Definition~\ref{def i r lang}:
    \begin{align*}
      L(\mathrm{reg}_{I,\mathrm{r}}(\alpha))&=L(\mathrm{r}(\alpha))\\
      &=L_{I,\mathrm{r}}(\alpha)\\
      L(\mathrm{reg}_{I,\mathrm{r}}(\emptyset))&=\emptyset\\
      &=L_{I,\mathrm{r}}(\emptyset)\\
      L(\mathrm{reg}_{I,\mathrm{r}}(\mathrm{o}(E_1,\ldots,E_k))) & =L(\mathrm{o}(\mathrm{reg}_{I,\mathrm{r}}(E_1),\ldots,\mathrm{reg}_{I,\mathrm{r}}(E_k)))\\
      & =\mathrm{o'}(L(\mathrm{reg}_{I,\mathrm{r}}(E_1)),\ldots,L(\mathrm{reg}_{I,\mathrm{r}}(E_k)))\\
      & =\mathrm{o'}(L_{I,\mathrm{r}}(E_1),\ldots,L_{I,\mathrm{r}}(E_k)) & (\textbf{Induction hypothesis})\\
      & =L_{I,\mathrm{r}}(\mathrm{o}(E_1,\ldots,E_k))\\
      L(\mathrm{reg}_{I,\mathrm{r}}(E_1\cdot E_2)) & =L(\mathrm{reg}_{I,\mathrm{r}}(E_1)\cdot \mathrm{reg}_{I,\mathrm{r}}(E_2))\\
      & =L(\mathrm{reg}_{I,\mathrm{r}}(E_1)) \cdot L(\mathrm{reg}_{I,\mathrm{r}}(E_2))\\
      & =L_{I,\mathrm{r}}(E_1) \cdot L_{I,\mathrm{r}}(E_2)&(\textbf{Induction hypothesis})\\
      & =L_{I,\mathrm{r}}(E_1  \cdot E_2)\\
      L(\mathrm{reg}_{I,\mathrm{r}}(E_1^*)) & =L(\mathrm{reg}_{I,\mathrm{r}}(E_1)^*)\\
      & =L(\mathrm{reg}_{I,\mathrm{r}}(E_1))^*\\
      & =L_{I,\mathrm{r}}(E_1)^*&(\textbf{Induction hypothesis})\\
      & =L_{I,\mathrm{r}}(E_1^*)\\
      L(\mathrm{reg}_{I,\mathrm{r}}(E_1\mid\phi)) & =
         \begin{cases}
             L(\mathrm{reg}_{I,\mathrm{r}}(E_1)) & \text{ if } \mathrm{eval}_{I,\mathrm{r}}(\phi),\\
             L(\emptyset) & \text{ otherwise,}
         \end{cases}\\
         & =
         \begin{cases}
             L_{I,\mathrm{r}}(E_1) & \text{ if } \mathrm{eval}_{I,\mathrm{r}}(\phi),\\
             L(\emptyset) & \text{ otherwise,}
         \end{cases}&(\textbf{Induction hypothesis})\\
         & = L_{I,\mathrm{r}}(E_1\mid\phi)\\
         L(\mathrm{reg}_{I,\mathrm{r}}(\alpha \dashv E_1)) & = L(\mathrm{r}(\alpha) \cap \mathrm{reg}_{I,\mathrm{r}}(E_1))\\
      & = L(\mathrm{r}(\alpha)) \cap L(\mathrm{reg}_{I,\mathrm{r}}(E_1))\\
      & = L(\mathrm{r}(\alpha)) \cap L_{I,\mathrm{r}}(E_1)&(\textbf{Induction hypothesis})\\
      & = \{\mathrm{r}(\alpha)\} \cap L_{I,\mathrm{r}}(E_1)\\
      & = \{\mathrm{r}(\alpha)\mid \mathrm{r}(\alpha)\in  L_{I,\mathrm{r}}(E_1)\}\\
      & = L_{I,\mathrm{r}}(\alpha \dashv E_1)\\
    \end{align*}
    where $k$ is any integer, $\mathrm{o}$ is any $k$-ary boolean operator associated with a mapping from $\{0,1\}^k$ to $\{0,1\}$, $\mathrm{o}'$ is the language operator associated with $\mathrm{o}$, $E_1,\ldots,E_k$ are any $k$ constrained expressions over $\mathcal{E}$, $\alpha$ is any word in $(\Sigma\cup \Gamma)^*$ and $\phi$ is a boolean formula in $\mathcal{P}(\mathcal{F}(\Gamma))$.
    \qed
  \end{proof}
  
 Once this regular
expression is computed, any classical membership test can be performed; hence:
  
  \begin{corollary}\label{cor i r lang rat}
    Let $\mathcal{E}=(\Sigma,\Gamma,\mathcal{P},\mathcal{F})$ be an expression environment and $E$ be a constrained expression over $\mathcal{E}$. Let $I$ be an expression interpretation over $\mathcal{E}$ and $\mathrm{r}$ be a $\Gamma$-realization over $I$. Let $w$ be a word in $\Sigma^*$. Then:
    \begin{align*}
      \text{To determine whether or not w belongs to } L_{I,\mathrm{r}}(E) \text{ is polynomially decidable.}
    \end{align*}
  \end{corollary}
  
  \subsection{Derivatives for Constrained Expressions}   
   
    \begin{remark}
    From now on, the set of boolean operators is restricted to the sum.
  \end{remark}

  

We showed in the previous section that the language of any constrained expression with a fixed interpretation and realization is regular. However, whenever the realization or the interpretation is not given, the language denoted by a constrained expression is an infinite union of regular languages which is not necessarily regular. Thus, in order to perform the membership test, the notion of partial derivatives is extended to the case of constrained expressions.
The idea is the following: for any interpretation and realization, a syntactical test can be achieved by computing all the splits of the word for which the membership test is performed. Once these precomputations terminated, new constrained expressions are generated and the membership test has to be performed for the empty word. In fact, we show that it is equivalent to solving the logical part, that is, to determine the satisfiability of the new formulae.
   While deriving expressions, choices have to be made in order to fix a realization. As an example, deriving the expression $x\cdot x$, where $x$ is a variable symbol, with respect to the symbol $a$, implies that the variable symbol $x$ is associated with a word starting with a, otherwise, the derivative would be empty. Consequently, such a realization transforms $x$ in $a$ and then associates the expression $x\cdot x$ with the expression $a\cdot a$. Deriving this expression w.r.t. $a$ returns the expression $\varepsilon\cdot a$ which is equivalent to $a$.
  

  As a direct consequence, the partial derivation has to memorize the assumptions made during the computation. Therefore, a partial derivative needs to be a set of tuples composed of an expression and a set of assumptions, where an assumption is a tuple composed of a variable symbol $x$ and a word $\alpha$: the realization associates the variable $x$ with the word $\alpha$.    
  These assumptions are needed to transform subexpressions of the initial expression. As an example, let us consider the expression $E\cdot F$. If assumptions are needed to perform the membership test while deriving $E$, these assumptions have to be applied over $F$ too via a substitution. Let us then extend the notion of substitution to words, to boolean formulae and to constrained expressions.


  Let $\Sigma$ be an alphabet and let $\alpha$ and $w$ be two words in $\Sigma^*$. Let $x$ be a symbol in $\Sigma$. We denote by $(\alpha)_{(x,w)}$ the word obtained by substituting any occurrence of $x$ in $\alpha$ by $w$, that is:
  \begin{align*}
    (\alpha)_{(x,w)}&=
      \begin{cases}
          \varepsilon & \text{ if }\alpha=\varepsilon,\\
          a(\alpha')_{(x,w)} & \text{ if }\alpha=a\alpha'\ \wedge\ a\in\Sigma\setminus\{x\},\\
          w(\alpha')_{(x,w)} & \text{ if }\alpha=x\alpha'
      \end{cases}
  \end{align*}
    \begin{definition}[$\mathrm{Term}$ Function]
     Let $\mathrm{Term}$ be the function from $\Sigma^*$ to $(\Sigma\cup\{\cdot,\varepsilon\})(\emptyset)$ (that is, the set of functions over the empty set of variables) inductively defined for any word $w$ as follows:
        \begin{align*}
        \mathrm{Term}(w)=
        \begin{cases}
            w & \text{ if } w\in\Sigma\cup\{\varepsilon\},\\
            \cdot(a,\mathrm{Term}(w')) & \text{ if }w=aw'\ \wedge\ a\in\Sigma\ \wedge\ w'\in\Sigma^*.\\
          \end{cases}
        \end{align*}
  \end{definition}
  For a boolean formula $\phi$, we denote by $\phi_{(x,w)}$ the boolean formula defined by:  
  \begin{align*}
      \phi_{(x,w)}=\phi_{x\leftarrow \mathrm{Term}(w)}.
  \end{align*}
  Finally, for any constrained expression $E$, we denote by $E_{(x,w)}$ the expression:
  \begin{align*}
    \emptyset_{(x,w)}&=\emptyset,\\
    (\alpha\dashv E_1)_{(x,w)}&=(\alpha)_{(x,w)}\dashv (E_1)_{(x,w)},\\
    (E_1+E_2)_{(x,w)}&=(E_1)_{(x,w)}+(E_2)_{(x,w)},\\
    (E_1\cdot E_2)_{(x,w)}&=(E_1)_{(x,w)}\cdot(E_2)_{(x,w)},\\
    (E_1^*)_{(x,w)}&=(E_1)_{(x,w)}^*,\\
    (E_1\mid\phi)_{(x,w)}&=(E_1)_{(x,w)}\mid \phi_{(x,w)}.
  \end{align*}
 Let $X$ be a subset of $\Sigma\times\Sigma^*$ satisfying the following two conditions:
 \begin{enumerate}
   \item \emph{Functional:} for any two distinct couples $(x,w)$ and $(x',w')$ in $X$, $x\neq x'$;
   \item \emph{Non-crossing: }for any two distinct couples $(x,w)$ and $(x',w')$ in $X$, $x$ does not appear in $w'$.
 \end{enumerate}
 Since $\Sigma$ is finite and therefore can be considered as ordered, we consider that $\Sigma\times\Sigma^*$ is ordered by an arbitrary lexicographic order from $\Sigma$.
 We extend the substitution to couples in $X$ as follows:
  \begin{align*}
    (\alpha)_{X}&=
      \begin{cases}
          \alpha & \text{ if } X=\emptyset,\\
          (\alpha_{(x,w)})_{X'} & \text{ if } X'=X\setminus\{(x,w)\} \wedge (x,w)=\mathrm{min}(X)
      \end{cases}\\     
    (\phi)_{X}&=
      \begin{cases}
          \phi & \text{ if } X=\emptyset,\\
          (\phi_{(x,w)})_{X'} & \text{ if } X'=X\setminus\{(x,w)\} \wedge (x,w)=\mathrm{min}(X)
      \end{cases}\\
    (E)_{X}&=
      \begin{cases}
          E & \text{ if } X=\emptyset,\\
          (E_{(x,w)})_{X'} & \text{ if } X'=X\setminus\{(x,w)\} \wedge (x,w)=\mathrm{min}(X)
      \end{cases}
  \end{align*}
  
%
  
  Let us continue with the previous example with the expression $x\cdot x$. If we want to check that the word $aa$ belongs to the language denoted by this expression, $x$ can be replaced by $a$, and then the derivation of $xx$ w.r.t. $aa$ produces an expression that denotes $\varepsilon$. However, substituting $x$ by a symbol is not sufficient in the general case. If we want to perform the membership test of the word $abab$, the derivation w.r.t. $a$ has to memorize that the realization associates $x$ with a word that \emph{starts with} the symbol $a$. Then the variable $x$ can be replaced by the word $ax$: the expression $xx$ is transformed into $axax$ when the derivative w.r.t. $a$ is computed, producing the expression $xax$. Deriving w.r.t. $b$, the assumption that $x$ (the new $x$, not the old one) is associated with a word that starts with $b$ has to be made, replacing $xax$ by $bxabx$ and producing $xabx$. Deriving it w.r.t. $a$, a new assumption can be made: if the new $x$ is replaced by $\varepsilon$, then the expression $xabx$ is replaced by the word $ab$ and its derivation w.r.t. $ab$ will produce $\varepsilon$, proving  that the word $abab$ is denoted by $xx$.
  
  As a direct consequence, the partial derivation of a constrained expression will compute all the combinations of assumptions that can be made during the derivation.


Let $\mathcal{E}=(\Sigma,\Gamma,\mathcal{P},\mathcal{F})$ be an expression environment.
  We denote by $\mathrm{Exp}(\mathcal{E})$ the set of the constrained expressions over $\mathcal{E}$. 
  Let us consider a word $\alpha$ in $(\Sigma\cup\Gamma)^*$.
  Either $\alpha=\varepsilon$, and therefore $\frac{\partial}{\partial_a}(\alpha)=\emptyset$, or $\alpha=\beta\alpha'$ with $\beta$ in $(\Sigma\cup\Gamma)$ and $\alpha'$ in $(\Sigma\cup\Gamma)^*$.
  If $\beta\in\Sigma\setminus\{a\}$, then $\frac{\partial}{\partial_a}(\alpha)=\emptyset$. 
  If $\beta=a$, then the only derived term of $\alpha$ is $\alpha'$ with no assumption of substitution made; therefore $\frac{\partial}{\partial_a}(\alpha)=\{(\alpha',\emptyset\}$.
  If $\beta=x\in\Gamma$, then two assumptions have to be considered:
  \begin{itemize}
    \item if $x$ starts with $a$, then $x$ can be substituted by $ax$, and then $\alpha$ becomes $ax\alpha'$. In this case, the only derived term is $x\alpha'$ under the substitution $(x,ax)$; therefore $\{(x\alpha',\{(x,ax)\})\}\in\frac{\partial}{\partial_a}(\alpha)$.
    \item if $x$ equals $\varepsilon$, then all the occurrences of $x$ have to be substituted by $\varepsilon$. 
    Thus $\alpha$ becomes $\alpha''=\alpha'_{(x,\varepsilon)}$. 
    Once this substitution made, there are no more occurrences of $x$ in $\alpha''$ and $\alpha''$ has to be derived w.r.t. $a$.
    Consequently, $\bigcup_{(\alpha'',X)\in\frac{\partial}{\partial_a}((\alpha')_{(x,\varepsilon)})} \{(\alpha'',X\cup\{(x,\varepsilon)\})\} \subset \frac{\partial}{\partial_a}(E)$.
  \end{itemize} 
  More formally,
  \begin{definition}[Constrained Derivative of a word]\label{def deriv expr cont mot}
    Let $\mathcal{E}=(\Sigma,\Gamma,\mathcal{P},\mathcal{F})$ be an expression environment and let $\alpha$ be a word in $(\Sigma\cup\Gamma)^*$. 
    Let $a$ be a symbol in $\Sigma$. 
    The \emph{constrained derivative of} $\alpha$ \emph{w.r.t.} $a$ is the subset $\frac{\partial}{\partial_a}(\alpha)$ of $(\Sigma\cup\Gamma)^*\times 2^{(\Gamma\times (\Sigma\cup\Gamma)^*)}$ inductively computed as follows:
    \begin{align*}
      \frac{\partial}{\partial_a}(\varepsilon)&=\emptyset,\\
      \frac{\partial}{\partial_a}(\alpha)&=
        \begin{cases}
            \{(\alpha',\emptyset)\} & \text{ if } \alpha=a\alpha'\\
            \{(x\cdot(\alpha')_{(x,ax)},\{(x,ax)\})\}\cup \\
            \quad \bigcup_{(\alpha'',X)\in\frac{\partial}{\partial_a}((\alpha')_{(x,\varepsilon)})} \{(\alpha'',X\cup\{(x,\varepsilon)\})\} & \text{ if }\alpha=x\cdot \alpha' \wedge x\in\Gamma,\\
            \emptyset & \text{otherwise}.
        \end{cases}
    \end{align*}
  \end{definition}
  
Let us check that the sets that that appear in a derived term are functional and non-crossing: 
  \begin{lemma}\label{lem ens deriv part ok}
    Let $\mathcal{E}=(\Sigma,\Gamma,\mathcal{P},\mathcal{F})$ be an expression environment and let $\alpha$ be a word in $(\Sigma\cup\Gamma)^*$. 
    Let $a$ be a symbol in $\Sigma$. 
    Then for any couple $(\alpha',Y)$ in $\frac{\partial}{\partial_a}(\alpha)$, it holds:
    \begin{align*}
      \text{$Y$ is a functional non-crossing set.}
    \end{align*}
  \end{lemma}
  \begin{proof}
    The proof is done by induction over the length of the words.
    Obviously the condition holds for non inductive cases of Definition~\ref{def deriv expr cont mot}.
    Let $(\alpha'',Y)$ be a couple in $\mathcal{P}=\frac{\partial}{\partial_a}((\alpha')_{(x,\varepsilon)})$.
    By induction over the length of $\alpha'$, it can be shown that $(y,z)\in Y \Rightarrow y\neq x$, since there is no occurrence of $x$ in $(\alpha')_{(x,\varepsilon)}$.
    By induction hypothesis, $Y$ is functional and non-crossing, and therefore so is $Y\cup\{(x,\varepsilon)\}$.
    \qed
  \end{proof}
  
  Let us now extend the partial derivation to constrained expressions. 
  We first syntactically define the derivatives, and then we prove their existence.
     
  \begin{definition}[Constrained Derivative of a Constrained Expression]\label{def deriv expr cont exp}
    Let $\mathcal{E}=(\Sigma,\Gamma,\mathcal{P},\mathcal{F})$ be an expression environment and let $E$ be a constrained expression over $\mathcal{E}$. Let $a$ be a symbol in $\Sigma$. The \emph{constrained derivative of} $E$ \emph{w.r.t.} $a$ is the subset $\frac{\partial}{\partial_a}(E)$ of $\mathrm{Exp}(\mathcal{E})\times 2^{(\Gamma\times (\Sigma\cup\Gamma)^*)}$ inductively computed as follows:
    \begin{align*}
      \frac{\partial}{\partial_a}(\alpha\dashv E_1)&=      
            \bigcup_{ 
              \subalign{
                (\alpha',X_1)&\in \frac{\partial}{\partial_a}(\alpha)\\
                (E_2,X_2)&\in \frac{\partial}{\partial_a}({E_1}_{X_1})
              }
            } \{((\alpha')_{X_2} \dashv E_2,X_1\cup X_2)\}\\
      \frac{\partial}{\partial_a}(E_1+E_2)&=\frac{\partial}{\partial_a}(E_1)\cup \frac{\partial}{\partial_a}(E_2),\\
      \frac{\partial}{\partial_a}(E_1\cdot E_2)&=\frac{\partial}{\partial_a}(E_1) \odot E_2 \cup (\varepsilon\dashv E_1) \odot  \frac{\partial}{\partial_a}(E_2),\\
      \frac{\partial}{\partial_a}(E_1^*)&=\frac{\partial}{\partial_a}(E_1) \odot E_1^*,\\
      \frac{\partial}{\partial_a}(E_1\mid\phi)&=\frac{\partial}{\partial_a}(E_1)\mid\mid \phi
    \end{align*}
    where for any subset $\mathcal{E}$ of $\mathrm{Exp}(\mathcal{E})\times 2^\Gamma$, for any expression $F$ and for any formula $\phi$,
    \begin{align*}
      \mathcal{E}  \odot  F&= \bigcup_{(E,X)\in\mathcal{E}} \{(E\cdot (F)_{X},X) \},\\
      F  \odot  \mathcal{E}&= \bigcup_{(E,X)\in\mathcal{E}} \{((F)_{X}\cdot E,X) \},\\
      \mathcal{E}  \mid\mid \phi&= \bigcup_{(E,X)\in\mathcal{E}} \{(E \mid (\phi)_{X},X) \}.
    \end{align*}
  \end{definition}
\begin{lemma}\label{lem union X dis}
    Let $\mathcal{E}=(\Sigma,\Gamma,\mathcal{P},\mathcal{F})$ be an expression environment and let $E$ be a constrained expression over $\mathcal{E}$. Let $X$ be a subset of $\{(x,ax),(x,\varepsilon)\mid a\in\Sigma,x\in\Gamma\}$. Let $a$ be a symbol in $\Sigma$. Then:
        \begin{align*}
      (E',X')\in\frac{\partial}{\partial_a}(E_X) \Rightarrow \{x\mid \exists (x,u)\in X\} \cap \{x'\mid \exists (x',u)\in X'\}=\emptyset.
        \end{align*}
  
  \end{lemma}
  \begin{proof}
    Let us notice that according to Definition~\ref{def deriv expr cont mot}, there exists a symbol $x'$ in $\Gamma$ and a word $u$ in $(\Sigma\cup\Gamma)^*$ satisfying $(x',u)\in X'$ only if there exists a subexpression $\alpha$ of $E_X$ such that $\alpha=x'\alpha'$ for some $\alpha'$ in $(\Sigma\cup\Gamma)^*$.    
    Furthermore, since applying $(x,u)\in X$ for some $u\in\{ax,\varepsilon\}$ over $E$ (producing $E_{x\leftarrow u}$) replaces any occurrence of $x$ either by $\varepsilon$ or by $ax$, there exists no subexpression $\alpha$ of $E_X$ such that $\alpha=x'\alpha'$ for some $\alpha'$ in $(\Sigma\cup\Gamma)^*$.    
    Consequently, the same symbol $x'$ cannot be the first component of both a tuple in $X$ and of a tuple in $X'$.
    \qed
  \end{proof}
  
  \begin{lemma}
    Let $\mathcal{E}=(\Sigma,\Gamma,\mathcal{P},\mathcal{F})$ be an expression environment and let $E$ be an expression over $\mathcal{E}$. 
    Let $a$ be a symbol in $\Sigma$. 
    Then for any couple $(E',Y)$ in $\frac{\partial}{\partial_a}(E)$, it holds that:
    \begin{align*}
      \text{$Y$ is a functional non-crossing set.}
    \end{align*}
  \end{lemma}
  \begin{proof}
    The proof is done by induction over the structure of expressions.
    Basic cases are well defined from Lemma~\ref{lem ens deriv part ok}.
    The cases of the sum, catenation, star and "\emph{such that}"-operation leave the properties of the sets unchanged.
    Thus, let us consider the case of the "\emph{membership}"-operation.
    
    Let $E=\alpha\dashv E_1$ and $(E',Y)$ be a couple in $\bigcup_{ 
              \subalign{
                (\alpha',X_1)&\in \frac{\partial}{\partial_a}(\alpha)\\
                (E_2,X_2)&\in \frac{\partial}{\partial_a}({E_1}_{X_1})
              }
            } \{((\alpha')_{X_2} \dashv E_2,X_1\cup X_2)\}$.
    Consider a couple $(\alpha',X_1)$ in $\frac{\partial}{\partial_a}(\alpha)$.
    From Lemma~\ref{lem ens deriv part ok}, $X_1$ is a functional non-crossing set.
    Let $(E_2,X_2)$ be a couple in $\frac{\partial}{\partial_a}({E_1}_{X_1})$.
    By induction hypothesis, $X_2$  is a functional non-crossing set.
    Finally, from Lemma~\ref{lem union X dis}, $X_1\cup X_2$ is a functional non-crossing set.
    \qed
  \end{proof}
  
  \begin{corollary}
    The constrained derivation is well-defined.
  \end{corollary} 
  In the following, in the examples, we use the symbol $\equiv$ as a semantical equivalence between expressions or sets. As an example, $\alpha\dashv \emptyset \equiv \emptyset$ or $\emptyset\equiv \{(\emptyset,\emptyset)\}$.
  \begin{example}\label{ex deriv anbncn}
    Let us consider the expression $E$ defined in Example~\ref{ex anbncn}.
    Let us set 
    \begin{align*}
      E_a &= (x\dashv a^*) &
      E_b &= (y\dashv b^*)\\
      E_c &= (z\dashv c^*) &
      \phi&= \sim(x,y)\wedge \sim (y,z)
    \end{align*}
    Consequently
    \begin{align*}
      E&= (E_a\cdot E_b \cdot E_c)\mid \phi
    \end{align*}
    Then
    \begin{align*}      
        \frac{\partial}{\partial_a}(E) &= \frac{\partial}{\partial_a}(E_a\cdot E_b\cdot E_c) \mid\mid \phi\\
        \frac{\partial}{\partial_a}(E_a\cdot E_b\cdot E_c)&=
          \frac{\partial}{\partial_a}(E_a)\odot (E_b \cdot E_c)
          \cup
          (\varepsilon\dashv E_a)\odot \frac{\partial}{\partial_a}(E_b \cdot E_c)\\
        \frac{\partial}{\partial_a}(E_b \cdot E_c) &=        
          \frac{\partial}{\partial_a}(E_b)\odot E_c
          \cup
          (\varepsilon\dashv E_b)\odot \frac{\partial}{\partial_a}(E_c)
    \end{align*}
    Furthermore,
    \begin{align*}
        \frac{\partial}{\partial_a}(x) &=\{(x,\{(x,ax)\})\} &
        \frac{\partial}{\partial_a}(x\dashv a^*) &=\{(x\dashv a^*,\{(x,ax)\})\}=\{(E_a,\{(x,ax)\})\}\\
        \frac{\partial}{\partial_a}(y) &=\{(y,\{(y,ay)\})\} &
        \frac{\partial}{\partial_a}(y\dashv b^*) &=\{(y\dashv \emptyset,\{(y,ay)\})\}\equiv \emptyset\\
        \frac{\partial}{\partial_a}(z) &=\{(z,\{(z,az)\})\} &
        \frac{\partial}{\partial_a}(z\dashv c^*) &=\{(z\dashv \emptyset,\{(z,az)\})\}\equiv \emptyset
    \end{align*}
    Then
    \begin{align*}
      \frac{\partial}{\partial_a}(E_b \cdot E_c) &\equiv \emptyset\\
      \frac{\partial}{\partial_a}(E_a\cdot E_b\cdot E_c)&\equiv \{(E_a,\{(x,ax)\})\} \odot (E_b \cdot E_c)\\
      &\equiv \{(E_a\cdot (E_b \cdot E_c),\{(x,ax)\})\}\\     
        \frac{\partial}{\partial_a}(E) &\equiv \{(E_a\cdot (E_b \cdot E_c),\{(x,ax)\})\} \mid\mid \phi\\
        &\equiv \{((E_a\cdot (E_b \cdot E_c))\mid \sim(ax,y)\wedge \sim (y,z),\{(x,ax)\})\}
    \end{align*}
  \end{example}
  
  The following of this section is devoted to proving that the derivation can be used to perform the membership test. In fact, we show that to determine whether or not a word $w$ is denoted by a constrained expression $E$ is equivalent to determining whether or not $\varepsilon$ is denoted by one of the derived expressions from $E$.
  
  We first model the fact that an assumption made through the derivation can be performed through a substitution without modifying the membership test: the main idea is that if a realization associates a word $au$ with a symbol $x$, the result is the same as if any occurrence of $x$ is replaced by $ax$ and if another realization is considered, where $u$ is associated with $x$. Hence, we can transfer a symbol from the realization to the expression.
  
  \begin{definition}[Compatible Realization]
    Let $\mathcal{E}=(\Sigma,\Gamma,\mathcal{P},\mathcal{F})$ be an expression environment. Let $I$ be an expression interpretation over $\mathcal{E}$ and $\mathrm{r}$ be a $\Gamma$-realization over $I$. Let $X$ be a subset of $\{(x,ax),(x,\varepsilon)\mid a\in\Sigma,x\in\Gamma\}$. The realization $\mathrm{r}$ is said to be \emph{compatible with} $X$ if and only if the following  two conditions hold:
    \begin{itemize}
      \item $\forall (x,ax)\in X$, $\mathrm{r}(x)=au$ for some word $u$ in $\Sigma^*$,
      \item $\forall (x,\varepsilon)\in X$, $\mathrm{r}(x)=\varepsilon$.
    \end{itemize}
  \end{definition}
  
  \begin{definition}[Associated Realization]
    Let $\mathcal{E}=(\Sigma,\Gamma,\mathcal{P},\mathcal{F})$ be an expression environment. Let $X$ be a subset of $\{(x,ax),(x,\varepsilon)\mid a\in\Sigma,x\in\Gamma\}$. Let $I$ be an expression interpretation over $\mathcal{E}$ and $\mathrm{r}$ be a $\Gamma$-realization over $I$ compatible with $X$.  The \emph{realization} $X$\emph{-associated with} $\mathrm{r}$ is defined for any symbol $x$ in $\Gamma$ as follows:
    \begin{align*} 
      \mathrm{r}'(x)&=
        \begin{cases}
            w & \text{ if } \mathrm{r}(x)=aw \wedge (x,ax)\in X,\\
            \varepsilon & \text{ if } \mathrm{r}(x)=\varepsilon \wedge (x,\varepsilon)\in X,\\
            \mathrm{r}(x) & \text{ otherwise.}
         \end{cases}
     \end{align*}
  \end{definition}
  
    \begin{lemma}\label{lem r assoc egal r pour phi}
    Let $\mathcal{E}=(\Sigma,\Gamma,\mathcal{P},\mathcal{F})$ be an expression environment and let $\phi$ be a boolean formula in $\mathcal{P}(\mathcal{F}(\Gamma))$. Let $X$ be a subset of $\{(x,ax),(x,\varepsilon)\mid a\in\Sigma,x\in\Gamma\}$. Let $I$ be an expression interpretation over $\mathcal{E}$ and $\mathrm{r}$ be a $\Gamma$-realization over $I$ compatible with $X$. Let $\mathrm{r}'$ be the realization $X$-associated with $\mathrm{r}$. Then:
    \begin{align*}
      \mathrm{eval}_{I,\mathrm{r}}(\phi)=\mathrm{eval}_{I,\mathrm{r}'}(\phi_X).
    \end{align*}
  \end{lemma}
  \begin{proof}
    We proceed in two steps.
    \begin{enumerate}
    \item\label{p1} Let us first show that for any term $t$ in $\mathcal{F}(\Gamma)$, $\mathrm{eval}_{I,\mathrm{r}}(t)=\mathrm{eval}_{I,\mathrm{r}'}(t_X)$.
    By induction over $t$.
    \begin{enumerate}
      \item Suppose that $t=x\in\Gamma$.
        Then 
        \begin{align*}
          \mathrm{eval}_{I,\mathrm{r}}(x)=\mathrm{r}(x) & =
            \begin{cases}
                  aw & \text{ if }\mathrm{r}(x)=aw \wedge (x,ax)\in X\\
                  \varepsilon & \text{ if }\mathrm{r}(x)=\varepsilon \wedge (x,\varepsilon)\in X\\
                  \mathrm{r}(x) & \text{ otherwise,}
            \end{cases}\\
            & =
              \begin{cases}
                  a\mathrm{r}'(x) & \text{ if }\mathrm{r}(x)=aw \wedge (x,ax)\in X\\
                  \varepsilon & \text{ if }\mathrm{r}(x)=\varepsilon \wedge (x,\varepsilon)\in X\\
                  \mathrm{r}'(x) & \text{ otherwise.}
              \end{cases}
        \end{align*}
        
        Furthermore,
        \begin{align*}
          x_X&=
            \begin{cases}
                  ax & \text{ if }\mathrm{r}(x)=aw \wedge (x,ax)\in X\\
                  \varepsilon & \text{ if }\mathrm{r}(x)=\varepsilon \wedge (x,\varepsilon)\in X\\
                  x & \text{ otherwise.}
             \end{cases}
         \end{align*}        
        Consequently,
        \begin{align*}
          \mathrm{eval}_{I,\mathrm{r}'}(x_X)=\mathrm{r}(x)  & =
            \begin{cases}
                  a\mathrm{r}'(x) & \text{ if }\mathrm{r}(x)=aw \wedge (x,ax)\in X\\
                  \varepsilon & \text{ if }\mathrm{r}(x)=\varepsilon \wedge (x,\varepsilon)\in X\\
                  \mathrm{r}'(x) & \text{ otherwise.}
            \end{cases}
         \end{align*}
      \item Suppose that $t=f(t_1,\ldots,t_n)$, that $I=(\Sigma^*,\mathfrak{F})$ and that $(\mathrm{eval}_{I,r}(t_1),\ldots,$ $\mathrm{eval}_{I,r}(t_n),x_{k+1})\in \mathfrak{F}(f)$. 
      By induction hypothesis, it holds that:
        \begin{align*}
          (\mathrm{eval}_{I,r}(t_1),\ldots,\mathrm{eval}_{I,r}(t_n),x_{k+1})\in \mathfrak{F}(f) \Leftrightarrow (\mathrm{eval}_{I,r'}({t_1}_X),\ldots,\mathrm{eval}_{I,r'}({t_n}_X),x_{k+1})\in \mathfrak{F}(f)
         \end{align*}
        Then:
        \begin{align*}
        \mathrm{eval}_{I,\mathrm{r}}(f(t_1,\ldots,t_n)) =x_{k+1} =\mathrm{eval}_{I,\mathrm{r}'}(f({t_1}_X,\ldots,{t_n}_X)) =\mathrm{eval}_{I,\mathrm{r}'}({t}_X)
         \end{align*}
    \end{enumerate}    
    \item Let us show now by induction over $\phi$ that $\mathrm{eval}_{I,\mathrm{r}}(\phi)=\mathrm{eval}_{I,\mathrm{r}'}(\phi_X)$.    
    \begin{enumerate}
      \item If $\phi=P(t_1,\ldots,t_k)$, then $\phi_X=P({t_1}_X,\ldots,{t_n}_X)$.
        Then
        \begin{align*}
        \mathrm{eval}_{I,\mathrm{r}}(\phi)=1 & \Leftrightarrow \mathrm{eval}_{I,\mathrm{r}}(P(t_1,\ldots,t_k))=1\\
        & \Leftrightarrow (\mathrm{eval}_{I,r}(t_1),\ldots,\mathrm{eval}_{I,r}(t_n))\in \mathfrak{F}(P)\\
        &\Leftrightarrow (\mathrm{eval}_{I,r'}({t_1}_X),\ldots,\mathrm{eval}_{I,r'}({t_n}_X))\in \mathfrak{F}(P)&\textbf{(Previous item~\ref{p1})}\\
        & \Leftrightarrow \mathrm{eval}_{I,\mathrm{r}'}(P({t_1}_X,\ldots,{t_n}_X))=1\\
        & \Leftrightarrow \mathrm{eval}_{I,\mathrm{r}'}(\phi_X)=1
         \end{align*}
      \item Suppose that $\phi=o(\phi_1,\ldots,\phi_n)$.
        Then $\phi_X=o({\phi_1}_X,\ldots,{\phi_n}_X)$.
        Then   
        \begin{align*}
          \mathrm{eval}_{I,\mathrm{r}}(o(\phi_1,\ldots,\phi_n))=1 & \Leftrightarrow o(\mathrm{eval}_{I,\mathrm{r}}({\phi_1},)\ldots,\mathrm{eval}_{I,\mathrm{r}}({\phi_n}))=1\\
          & \Leftrightarrow o(\mathrm{eval}_{I,\mathrm{r}'}({\phi_1}_X,)\ldots,\mathrm{eval}_{I,\mathrm{r}'}({\phi_n}_X))=1&\textbf{(Induction hypothesis)}\\
          & \Leftrightarrow \mathrm{eval}_{I,\mathrm{r}'}(o({\phi_1}_X,\ldots,{\phi_n}_X))=1\\
          & \Leftrightarrow \mathrm{eval}_{I,\mathrm{r}'}(\phi_X)=1
         \end{align*}
    \end{enumerate}
    \end{enumerate}    
    \qed
  \end{proof}

  \begin{lemma}\label{lem r assoc egal r}
    Let $\mathcal{E}=(\Sigma,\Gamma,\mathcal{P},\mathcal{F})$ be an expression environment and let $E$ be a constrained expression over $\mathcal{E}$. Let $X$ be a subset of $\{(x,ax),(x,\varepsilon)\mid a\in\Sigma,x\in\Gamma\}$. Let $I$ be an expression interpretation over $\mathcal{E}$ and $\mathrm{r}$ be a $\Gamma$-realization over $I$ compatible with $X$. Let $\mathrm{r}'$ be the realization $X$-associated with $\mathrm{r}$. Then:
        \begin{align*}
      L_{I,\mathrm{r}}(E)=L_{I,\mathrm{r}'}(E_X).
        \end{align*}
  \end{lemma}
  \begin{proof}
    By induction over the structure of $E$.
    \begin{enumerate}
      \item Let us suppose that $E=\alpha$. 
        By recurrence over the length of $\alpha$.
        \begin{enumerate}
          \item If $\alpha=\varepsilon$ or $\alpha=a\in\Sigma$, $E=E_X$ and then $L_{I,\mathrm{r}}(E)=\{\alpha\}=L_{I,\mathrm{r}'}(E_X)$.
          \item If $\alpha=x\in\Gamma$, then           
        \begin{align*}
          L_{I,\mathrm{r}}(x)=\{\mathrm{r}(x)\} & =
            \begin{cases}
                      \{aw\} & \text{ if } \mathrm{r}(x)=aw\wedge (x,ax)\in X,\\
                      \{\varepsilon\} & \text{ if } \mathrm{r}(x)=\varepsilon \wedge (x,\varepsilon)\in X,\\
                      \{\mathrm{r}(x)\} & \text{ otherwise.}
            \end{cases}\\              
                  & =
            \begin{cases}
                      \{a\mathrm{r}'(x)\}& \text{ if } \mathrm{r}(x)=aw\wedge (x,ax)\in X,\\
                      \{\varepsilon\} & \text{ if } \mathrm{r}(x)=\varepsilon \wedge (x,\varepsilon)\in X,\\
                          \{\mathrm{r}'(x)\} & \text{ otherwise.}\\
                    \end{cases}
        \end{align*}
              
              Furthermore, 
              $E_X=
            \begin{cases}
                    ax & \text{ if } \mathrm{r}(x)=aw\wedge (x,ax)\in X,\\
                    \varepsilon & \text{ if } \mathrm{r}(x)=\varepsilon\wedge (x,\varepsilon)\in X,\\
                    x& \text{ otherwise,}
                  \end{cases}
               $
               
               and then 
               
               $L_{I,\mathrm{r}'}(E_X)=               
            \begin{cases}
                  \{a\mathrm{r}'(x)\} & \text{ if } \mathrm{r}(x)=aw\wedge (x,ax)\in X,\\
                  \{\varepsilon\} & \text{ if } \mathrm{r}(x)=\varepsilon \wedge (x,\varepsilon)\in X,\\
                  \{\mathrm{r}'(x)\} & \text{ otherwise.}
                \end{cases}$
              
          \item Suppose that $\alpha=\alpha'\beta$ with $\alpha'\in\Sigma\cup\Gamma$.
                Then        
        \begin{align*}
                    L_{I,\mathrm{r}}(\alpha) & =L_{I,\mathrm{r}}(\alpha')L_{I,\mathrm{r}}(\beta)\\
                    & =L_{I,\mathrm{r}'}(\alpha'_X) L_{I,\mathrm{r}'}(\beta_X) &\textbf{(Recurrence hypothesis)}\\
                    & =L_{I,\mathrm{r}'}(\alpha'_X\beta_X)\\
                    & =L_{I,\mathrm{r}'}(\alpha_X)       
        \end{align*}
        \end{enumerate}
        
      \item Let us suppose that $E=\alpha\vdash F$.
        Then    
        \begin{align*}
            L_{I,\mathrm{r}}(\alpha\vdash E) & =\{\mathrm{r}(\alpha)\mid \mathrm{r}(\alpha)\in L_{I,\mathrm{r}}(F)\}\\
            & =\{\mathrm{r}'(\alpha_X)\mid \mathrm{r}'(\alpha_X)\in L_{I,\mathrm{r}'}(F_X)\}&\textbf{(Induction hypothesis)}\\
            & =L_{I,\mathrm{r}'}(\alpha_X\vdash F_X)\\
            & =L_{I,\mathrm{r}'}(E_X)   
        \end{align*}
        
      \item Let us suppose that $E=F\mid\phi$.
        According to Lemma~\ref{lem r assoc egal r pour phi}, $\mathrm{eval}_{I,r}(\phi)=\mathrm{eval}_{I,r'}(\phi_X)$.
        Hence if $\mathrm{eval}_{I,r}(\phi)=\mathrm{eval}_{I,r'}(\phi')=0$, $L_{I,\mathrm{r}}(E)=L_{I,\mathrm{r}'}(E_X)=\emptyset$.
        Otherwise,
        \begin{align*}
            L_{I,\mathrm{r}}(F\mid\phi) & =L_{I,\mathrm{r}}(F)\\
            & =L_{I,\mathrm{r}'}(F_X)&\textbf{(Induction hypothesis)}\\
            & =L_{I,\mathrm{r}'}(F_X\mid\phi_X)\\
            & =L_{I,\mathrm{r}'}(E_X)
        \end{align*}
      \item Let us suppose that $E=F+G$.
        Then:
        \begin{align*}
            L_{I,\mathrm{r}}(F+G) & =L_{I,\mathrm{r}}(F)\cup L_{I,\mathrm{r}}(G)\\
            & =L_{I,\mathrm{r}'}(F_X)\cup L_{I,\mathrm{r}'}(G_X)&\textbf{(Induction hypothesis)}\\
            & =L_{I,\mathrm{r}'}(F_X+G_X)\\
            & =L_{I,\mathrm{r}'}(E_X)
        \end{align*}
        
      \item Let us suppose that $E=F\cdot G$.
        Then:
        \begin{align*}
            L_{I,\mathrm{r}}(F\cdot G) & =L_{I,\mathrm{r}}(F)\cdot L_{I,\mathrm{r}}(G)\\
            & =L_{I,\mathrm{r}'}(F_X)\cdot L_{I,\mathrm{r}'}(G_X)&\textbf{(Induction hypothesis)}\\
            & =L_{I,\mathrm{r}'}(F_X\cdot G_X)\\
            & =L_{I,\mathrm{r}'}(E_X)
        \end{align*}
      \item Let us suppose that $E=F^*$.
        Then:
        \begin{align*}
            L_{I,\mathrm{r}}(F^*) & =L_{I,\mathrm{r}}(F)^*\\
            & =L_{I,\mathrm{r}'}(F_X)^*&\textbf{(Induction hypothesis)}\\
            & =L_{I,\mathrm{r}'}(F_X^*)\\
            & =L_{I,\mathrm{r}'}(E_X)
        \end{align*}
    \end{enumerate}
    \qed
  \end{proof}
  
  Let us now show that the partial derivation can be used to perform the membership test over constrained expressions whenever the realization is not fixed: if a word $aw$ belongs to the language of a constrained expression $E$ whenever a realization $\mathrm{r}$ is considered, the partial derivation w.r.t. $a$ always produces at least a tuple $(E',X)$ where $E'$ denotes $w'$ when another realization is considered (the realization $X$-associated with $\mathrm{r}'$). 
  
  \begin{proposition}\label{prop eq quot deriv part}
    Let $\mathcal{E}=(\Sigma,\Gamma,\mathcal{P},\mathcal{F})$ be an expression environment and let $E$ be a constrained expression over $\mathcal{E}$. Let $I$ be an expression interpretation over $\mathcal{E}$ and $\mathrm{r}$ be a $\Gamma$-realization over $I$. Let $w$ be a word in $\Sigma^*$ and $a$ be a symbol in $\Sigma$. Then the  following  two conditions are equivalent:
    \begin{itemize}
      \item $w \in a^{-1}(L_{I,\mathrm{r}}(E))$
      \item there exists a tuple $(E',X)\in\frac{\partial}{\partial_a}(E)$ such that $w \in L_{I,\mathrm{r}'}(E')$, where $\mathrm{r}'$ is the realization $X$-associated with $\mathrm{r}$.
    \end{itemize}
  \end{proposition}
  \begin{proof}
    By induction over the structure of $E$. By definition, $w\in a^{-1}(L_{I,\mathrm{r}}(E))$ $\Leftrightarrow$ $aw \in L_{I,\mathrm{r}}(E)$.     
    \begin{enumerate}
      \item Whenever $E\in\{\emptyset,\varepsilon\}$, $a^{-1}(L_{I,\mathrm{r}}(E))$ and $\frac{\partial}{\partial_a}(E)$ are both empty.
      \item Let us suppose that $E=\alpha$. Three cases can occur.
      \begin{enumerate}
        \item Let us suppose that $\alpha=b\alpha'$. If $b\neq a$, $a^{-1}(L_{I,\mathrm{r}}(E))=\emptyset$.        
        Otherwise (if $b=a$), $(\alpha',\emptyset)\in \frac{\partial}{\partial_a}(E)$. Furthermore,
        \begin{align*}
          aw \in L_{I,\mathrm{r}}(a\alpha') \Leftrightarrow w\in L_{I,\mathrm{r}}(\alpha')
        \end{align*}        
        Finally, since $\mathrm{r}'=\mathrm{r}$ is the realization $\emptyset$-associated with $\mathrm{r}$,
        \begin{align*}
          w\in a^{-1}(L_{I,\mathrm{r}}(\alpha)) \Leftrightarrow w\in L_{I,\mathrm{r}'}(\alpha')
        \end{align*}         
        \item Let us suppose that $\alpha=x\alpha'$ and that $\mathrm{r}(x)=aw'$ for some $w'\in\Sigma^*$. Let us denote by $\mathrm{r}'$ the realization $\{(x,ax)\}$-associated with $\mathrm{r}$. Then
        \begin{align*}
          \mathrm{r}(\alpha)=\mathrm{r'}(\alpha)_{x\leftarrow ax}=\mathrm{r}'(ax(\alpha')_{x\leftarrow ax})
        \end{align*}        
        As a direct consequence, 
        \begin{align*}
          aw \in L_{I,\mathrm{r}}(\alpha) \Leftrightarrow aw \in L_{I,\mathrm{r}'}(ax(\alpha')_{x\leftarrow ax}) \Leftrightarrow w \in L_{I,\mathrm{r}'}(x(\alpha')_{x\leftarrow ax})
        \end{align*}        
        Finally, it holds by definition that $(x(\alpha')_{x\leftarrow ax},\{(x,ax)\})\in \frac{\partial}{\partial_a}(E)$.        
        \item Let us suppose that $\alpha=x\alpha'$ and that $\mathrm{r}(x)=\varepsilon$. Then:
        \begin{align*}
          \mathrm{r}(\alpha) =\mathrm{r}(\alpha')=\mathrm{r}((\alpha')_{x\leftarrow\varepsilon})
        \end{align*}        
        As a direct consequence,
        \begin{align*}
          w\in a^{-1}(L_{I,\mathrm{r}}((\alpha')_{x\leftarrow\varepsilon}))
        \end{align*}        
        According to induction hypothesis, there exists a tuple $(\alpha'',X)$ belonging to $\frac{\partial}{\partial_a}((\alpha')_{x\leftarrow\varepsilon})$ such that $w \in L_{I,\mathrm{r}''}(\alpha'')$ with $\mathrm{r}''$ the realization $X$-associated with $\mathrm{r}$. Let $\mathrm{r}'$ be the realization $X\cup \{(x,\varepsilon)\}$-associated with $\mathrm{r}$. Since there is no occurrence of $x$ in $\alpha''$ (since $\alpha''$ is a derivated term of $\alpha_{x\leftarrow\varepsilon}$), it holds that $\mathrm{r}''(\alpha'')=\mathrm{r}'(\alpha'')$. Furthermore, it holds by definition that $(\alpha'',X\cup\{(x,\varepsilon)\})\in \frac{\partial}{\partial_a}(E)$. Finally, $w \in L_{I,\mathrm{r}'}(\alpha'')$
      \end{enumerate}
      \item Let us suppose that $E=\alpha\dashv E_1$. 
      Consider that there exists a tuple $(E',X)\in\frac{\partial}{\partial_a}(E)$ such that $w \in L_{I,\mathrm{r}'}(E')$, where $\mathrm{r}'$ is the realization $X$-associated with $\mathrm{r}$. 
      
      Equivalently, there exist $(\alpha',X_1)\in \frac{\partial}{\partial_a}(\alpha)$ and  $(E_2,X_2)\in \frac{\partial}{\partial_a}({E_1}_{X_1})$ 
      such that $(E'=(\alpha')_{X_2} \dashv E_2,X_1\cup X_2)\in \frac{\partial}{\partial_a}(E)$ and $w \in L_{I,\mathrm{r}'}(E')$, where $\mathrm{r}'$ is the realization $X_1\cup X_2$-associated with $\mathrm{r}$.
      
      By definition, $L_{I,\mathrm{r}'}(E')=L_{I,\mathrm{r}'}((\alpha')_{X_2})\cap L_{I,\mathrm{r}'}(E_2)$. Consequently $w\in L_{I,\mathrm{r}'}((\alpha')_{X_2})$ and $w\in L_{I,\mathrm{r}'}(E_2)$. Let us denote by $\mathrm{r}_1$ (resp. $\mathrm{r}_2$) the realization $X_1$-associated (resp. $X_2$-associated) with $\mathrm{r}$.
      
      Since $\mathrm{r}'$ is the realization $(X_1\cup X_2)$-associated with $\mathrm{r}$, and since according to Lemma~\ref{lem union X dis}, $\{x_1\mid \exists (x_1,u)\in X_1\} \cap \{x_2\mid \exists (x_2,u)\in X_2\}=\emptyset$, the by definition, for any symbol $x$ in $\Gamma$, the following equality is satisfied:
        \begin{align*}
        \mathrm{r}'(x)=
        \begin{cases}
            w & \text{ if } \mathrm{r}(x)=aw \wedge (x,ax)\in X_1,\\
            w & \text{ if } \mathrm{r}(x)=aw \wedge (x,ax)\in X_2,\\
            \varepsilon & \text{ if } \mathrm{r}(x)=\varepsilon \wedge (x,\varepsilon)\in X_1,\\
            \varepsilon & \text{ if } \mathrm{r}(x)=\varepsilon \wedge (x,\varepsilon)\in X_2,\\
            \mathrm{r}(x) & \text{ otherwise.}
          \end{cases}
        \end{align*}    
    By definitions of $\mathrm{r}_1$ and $\mathrm{r}_2$,  for any symbol $x$ in $\Gamma$:
        \begin{align*}
        \mathrm{r}_1(x)=
        \begin{cases}
            w & \text{ if } \mathrm{r}(x)=aw \wedge (x,ax)\in X_1,\\
            \varepsilon & \text{ if } \mathrm{r}(x)=\varepsilon \wedge (x,\varepsilon)\in X_1,\\
            \mathrm{r}(x) & \text{ otherwise,}
          \end{cases}\\
        \mathrm{r}_2(x)=
        \begin{cases}
            w & \text{ if } \mathrm{r}(x)=aw \wedge (x,ax)\in X_2,\\
            \varepsilon & \text{ if } \mathrm{r}(x)=\varepsilon \wedge (x,\varepsilon)\in X_2,\\
            \mathrm{r}(x) & \text{ otherwise.}
          \end{cases}
        \end{align*}    
    Hence, since $X_1\cap X_2=\emptyset$,
        \begin{align*}
        \mathrm{r}'(x)=
        \begin{cases}
            w & \text{ if } \mathrm{r}_1(x)=aw \wedge (x,ax)\in X_1,\\
            w & \text{ if } \mathrm{r}(x)=aw=\mathrm{r}_1(x) \wedge (x,ax)\in X_2\setminus X_1,\\
            \varepsilon & \text{ if } \mathrm{r}_1(x)=\varepsilon \wedge (x,\varepsilon)\in X_1,\\
            \varepsilon & \text{ if } \mathrm{r}(x)=\varepsilon=\mathrm{r}_1(x) \wedge (x,\varepsilon)\in X_2\setminus X_1,\\
            \mathrm{r}(x) & \text{ otherwise.}
          \end{cases}
        \end{align*}
    Consequently, $\mathrm{r}'$ is $X_2$-associated with $\mathrm{r}_1$. Symmetrically, $\mathrm{r}'$ is $X_1$-associated with $\mathrm{r}_2$.
    
      Since $w$ belongs to $L_{I,\mathrm{r}'}(E_2)$, there exists a tuple $(E_2,X_2)\in\frac{\partial}{\partial_a}((E_1)_{X_1})$ such that $w \in L_{I,\mathrm{r}'}(E_2)$, where $\mathrm{r}'$ is the realization $X_2$-associated with $\mathrm{r}_1$. By induction hypothesis $aw\in L_{I,\mathrm{r}_1}({E_1}_{X_1})$. According to Lemma~\ref{lem r assoc egal r}, $aw\in L_{I,\mathrm{r}_1}({E_1}_{X_1})$ $\Leftrightarrow$ $aw\in L_{I,\mathrm{r}}({E_1})$. Since $\mathrm{r}'$ is $X_2$-associated with $\mathrm{r}_1$, according to Lemma~\ref{lem r assoc egal r}, $L_{I,\mathrm{r}'}((\alpha')_{X_2})=L_{I,\mathrm{r}_1}(\alpha')$. Hence $w\in L_{I,\mathrm{r}_1}(\alpha')$ and by induction hypothesis $aw\in L_{I,\mathrm{r}}(\alpha)$.      
      Finally, it holds that $aw\in L_{I,\mathrm{r}}(E)$.
      \item Let us suppose that $E=E_1 + E_2$. Then       
      $w\in a^{-1}(L_{i,r}(E_1+ E_2)$      
      $\Leftrightarrow$ $w\in a^{-1}(L_{i,r}(E_1))$ $\vee$ $w\in a^{-1}(L_{i,r}(E_2))$.      
      By induction, it is equivalent to $\exists k\in\{1,2\} \mid \exists (E',X)\in\frac{\partial}{\partial_a}(E_k)$, $w\in L_{I,r'}(E')$ where $r'$ is the realization $X$ associated with $r$.       
      Since $\frac{\partial}{\partial_a}(E_1)\cup \frac{\partial}{\partial_a}(E_2)\subset \frac{\partial}{\partial_a}(E_1+E_2)$, it is equivalent to $\exists (E',X)\in\frac{\partial}{\partial_a}(E_1+E_2)$, $w\in L_{I,r'}(E')$ where $r'$ is the realization $X$ associated with $r$.
      \item Let us suppose that $E=E_1\cdot E_2$. Then 
        \begin{align*}
        w\in a^{-1}(L_{i,r}(E_1\cdot E_2)
        \Leftrightarrow
        \begin{cases}
          & w\in a^{-1}(L_{I,r}(E_1))\cdot L_{I,r}(E_2)\\
          \vee &  (\varepsilon\in L_{I,r}(E_1) \wedge w\in a^{-1}(L_{I,r}(E_2)))\\
        \end{cases}
        \end{align*}            
      Moreover 
        \begin{align*}
        w\in a^{-1}(L_{I,r}(E_1))\cdot L_{I,r}(E_2)
        \Leftrightarrow
        \begin{cases}
          & w=w_1\cdot w_2\\
          \wedge & \exists (E',X)\in\frac{\partial}{\partial_a}(E_1), w_1\in L_{I,r'}(E') \\
          & \text{ where $r'$ is the realization $X$-associated with $r$}\\
          \wedge& w_2\in L_{I,r}(E_2)
        \end{cases}
        \end{align*}      
      According to Lemma~\ref{lem r assoc egal r}, $w_2\in L_{I,r}(E_2)$ $\Leftrightarrow$ $w_2\in L_{I,r'}((E_2)_X)$.      
      Hence
        \begin{align*}
        w\in a^{-1}(L_{I,r}(E_1))\cdot L_{I,r}(E_2)
        \Leftrightarrow
        \begin{cases}
          & w\in L_{I,r'}(E'\cdot ({E_2}_X))\\
          & \text{ where $r'$ is the realization $X$-associated with $r$}\\
          \wedge & (E'\cdot (E_2)_X,X) \in \frac{\partial}{\partial_a}(E_1)\odot E_2\subset \frac{\partial}{\partial_a}(E_1\cdot E_2)
        \end{cases}
        \end{align*}      
      Finally consider that $ \varepsilon\in L_{I,r}(E_1)$ $\wedge$ $w\in a^{-1}(L_{I,r}(E_2))$.      
      By induction $w\in a^{-1}(L_{I,r}(E_2))$      
      $\Leftrightarrow$ $\exists (E',X)\in\frac{\partial}{\partial_a}(E_2),$ $w\in L_{I,r'}(E')$ where $r'$ is the realization $X$-associated with $r$.      
      Moreover $ \varepsilon\in L_{I,r}(E_1)$ $\Leftrightarrow$ $ \varepsilon\in L_{I,r}(\varepsilon \dashv E_1)$ is equivalent to $ \varepsilon\in L_{I,r'}(\varepsilon \dashv (E_1)_X)$ according to Lemma~\ref{lem r assoc egal r}.
      
      Hence 
        \begin{align*}
        w\in a^{-1}(L_{I,r}(E_2))
        \Leftrightarrow
        \begin{cases}
          & w\in L_{I,r'}( (\varepsilon\dashv (E_1)_X) \cdot E_2)\\
          & \text{ where $r'$ is the realization $X$-associated with $r$}\\
          \wedge & ((\varepsilon\dashv (E_1)_X) \cdot E',X) \in (\varepsilon\dashv (E_1)_X)\odot \frac{\partial}{\partial_a}(E_2) \subset \frac{\partial}{\partial_a}(E_1\cdot E_2)
        \end{cases}
        \end{align*}
      \item Let us suppose that $E=E_1^*$. Then       
      $w\in a^{-1}(L_{i,r}(E_1^*)$      
      $\Leftrightarrow$ $w=w_1\cdot w_2 $ $\wedge$ $w_1\in a^{-1}(L_{i,r}(E_1))$ $\wedge$ $w_2\in L_{i,r}(E_1^*)$.      
      By induction, $w_1\in a^{-1}(L_{i,r}(E_1))$      
      $\Leftrightarrow$ $\exists (E',X)\in\frac{\partial}{\partial_a}(E_1)$ $w_1\in L_{I,r'}(E')$ where $r'$ is the realization $X$-associated with $r$.      
      According to Lemma~\ref{lem r assoc egal r}, $w_2\in L_{i,r}(E_1^*)$      
      $\Leftrightarrow$ $w_2\in L_{i,r'}((E_1^*)_X)$.      
      Hence 
        \begin{align*}
        w\in a^{-1}(L_{i,r}(E_1^*)
        \Leftrightarrow
        \begin{cases}
          & w=w_1\cdot w_2 \\
          \wedge & w\in L_{I,r'}(E'\cdot (E_1^*)_X)\\
          & \text{ $r'$ is the realization $X$-associated with $r$}\\
          \wedge & (E'\cdot (E_1^*)_X,X) \in \frac{\partial}{\partial_a}(E_1)\odot E_1^* \subset \frac{\partial}{\partial_a}(E_1^*)\\
        \end{cases}
        \end{align*}
      \item Let us suppose that $E=E_1\mid \phi$. Then      
      $w\in a^{-1}(L_{i,r}(E_1\mid \phi)$      
      $\Leftrightarrow$ $w\in a^{-1}(L_{i,r}(E_1))$ $\wedge$ $\mathrm{eval}_{I,r}(\phi)$.      
      By induction, $w_1\in a^{-1}(L_{i,r}(E_1))$      
      $\Leftrightarrow$ $\exists (E',X)\in\frac{\partial}{\partial_a}(E_1)$, $w_1\in L_{I,r'}(E')$ where $r'$ is the realization $X$ associated with $r$.      
      According to Lemma\ref{lem r assoc egal r pour phi}, $\mathrm{eval}_{I,r}(\phi)=\mathrm{eval}_{I,r'}(\phi_X)$.      
      Consequently, 
        \begin{align*}
        w\in a^{-1}(L_{i,r}(E_1\mid \phi)
        \Leftrightarrow
        \begin{cases}
          & w\in L_{i,r'}(E' \mid \phi_X)\\
          & \text{ where $r'$ is the realization $X$-associated with $r$}\\
          \wedge & (E'\mid \phi_X,X)\in \frac{\partial}{\partial_a}(E_1) \mid\mid \phi \subset \frac{\partial}{\partial_a}(E_1\mid\phi)\\
        \end{cases}
        \end{align*}       
    \end{enumerate}    
    \qed
  \end{proof}
  
  As a direct consequence of Proposition~\ref{prop eq quot deriv part}, the partial derivation of a constrained expression w.r.t. a symbol is valid.
  
  \begin{theorem}\label{thm egal quot deriv part}
    Let $\mathcal{E}=(\Sigma,\Gamma,\mathcal{P},\mathcal{F})$ be an expression environment and let $E$ be a constrained expression over $\mathcal{E}$. Let $I$ be an expression interpretation over $\mathcal{E}$. Let $a$ be a symbol in $\Sigma$. Then the  following two conditions hold:
    \begin{enumerate}
      \item $a^{-1}(L_{I}(E))= \bigcup_{(E',X)\in\frac{\partial}{\partial_a}(E)} L_I(E')$,
      \item $a^{-1}(L(E))= \bigcup_{(E',X)\in\frac{\partial}{\partial_a}(E)} L(E')$.
    \end{enumerate}
  \end{theorem}
  \begin{proof}
    Let $w$ be a word in $\Sigma^*$. 
    \begin{enumerate}
      \item By definition of $L_{I}(E)$, $w\in a^{-1}(L_{I}(E))$ $\Leftrightarrow$ there exists a realization $\mathrm{r}$ such that $w\in a^{-1}(L_{I,\mathrm{r}}(E))$.
      
      According to Proposition~\ref{prop eq quot deriv part},      
      $w\in a^{-1}(L_{I,\mathrm{r}}(E))$ $\Leftrightarrow$ there exists a tuple $(E',X)\in\frac{\partial}{\partial_a}(E)$ such that $w \in L_{I,\mathrm{r}'}(E')$, where $\mathrm{r}'$ is the realization $X$-associated with $\mathrm{r}$. As a direct conclusion, $w\in \bigcup_{(E',X)\in\frac{\partial}{\partial_a}(E)} L_{I,\mathrm{r}'}(E')\subset \bigcup_{(E',X)\in\frac{\partial}{\partial_a}(E)} L_I(E')$.
      
      Suppose that $w\in \bigcup_{(E',X)\in\frac{\partial}{\partial_a}(E)} L_I(E')$. Hence there exists a tuple $(E',X)\in\frac{\partial}{\partial_a}(E)$ such that $w\in L_I(E')$. By definition of $ L_I(E')$, there exists a realization $\mathrm{r}'$ such that $w\in L_{I,\mathrm{r}'}(E')$. Let $\mathrm{r}$ be the realization defined for any symbol $x$ in $\Gamma$ as follows:
        \begin{align*}
        \mathrm{r}(x)=
        \begin{cases}
              aw & \text{ if } \mathrm{r}'(x)=w \wedge (x,ax)\in X,\\
              \varepsilon & \text{ if } \mathrm{r}'(x)=\varepsilon \wedge (x,\varepsilon)\in X,\\
              \mathrm{r}'(x) & \text{ otherwise.}
            \end{cases}
        \end{align*}
      
      As a direct consequence, $\mathrm{r}'$ is the realization $X$-associated with $\mathrm{r}$, and according to Proposition~\ref{prop eq quot deriv part}, since there exists a tuple $(E',X)\in\frac{\partial}{\partial_a}(E)$ such that $w \in L_{I,\mathrm{r}'}(E')$, where $\mathrm{r}'$ is the realization $X$-associated with $\mathrm{r}$, it holds that $w\in a^{-1}(L_{I,\mathrm{r}}(E))$. By definition of $L_{I}(E)$, $w\in a^{-1}(L_{I}(E))$.
      \item By definition of $L(E)$, $w\in a^{-1}(L(E))$ $\Leftrightarrow$ there exists an interpretation $I$ such that $w\in a^{-1}(L_I(E))$. We have shown that there exists an interpretation $I$ such that $w\in a^{-1}(L_I(E))$ $\Leftrightarrow$ there exists an interpretation $I$ such that $w\in \bigcup_{(E',X)\in\frac{\partial}{\partial_a}(E)} L_I(E')$, 
      which is by definition of $L(E')$ equivalent to the fact that $w\in$ $ \bigcup_{(E',X)\in\frac{\partial}{\partial_a}(E)} L(E')$. 
    \end{enumerate}
    \qed
  \end{proof}
  
  The partial derivation can be extended from symbols to words as follows: 
  
  \begin{definition}[Word Derivative]
    Let $\mathcal{E}=(\Sigma,\Gamma,\mathcal{P},\mathcal{F})$ be an expression environment and let $E$ be a constrained expression over $\mathcal{E}$. Let $a$ be a symbol in $\Sigma$ and $w$ be a word in $\Sigma^*$. Then:
        \begin{align*}
        \frac{\partial}{\partial_{aw}}(E)=
        \begin{cases}
            \frac{\partial}{\partial_{a}}(E) & \text{ if }w=\varepsilon,\\
            \bigcup_{(E',X)\in\frac{\partial}{\partial_{a}}(E)} \frac{\partial}{\partial_{w}}(E') & \text{ otherwise.}
          \end{cases}
        \end{align*}
  \end{definition}  
  
  \begin{example}\label{ex deriv anbncn abc}
    Let us continue Example~\ref{ex deriv anbncn}.
    Let us set $E'=((E_a\cdot (E_b \cdot E_c))\mid \sim(ax,y)\wedge \sim (y,z)$ and let us compute 
    \begin{align*}
      \frac{\partial}{\partial_{ab}}(E)&\equiv \frac{\partial}{\partial_b}(E')\\
      &=\frac{\partial}{\partial_b}(E_a\cdot (E_b \cdot E_c)) \mid\mid \sim(ax,y)\wedge \sim (y,z)\\
      \frac{\partial}{\partial_b}(E_a\cdot (E_b \cdot E_c))&=
        \frac{\partial}{\partial_b}(E_a)\odot (E_b \cdot E_c)
        \cup
        (\varepsilon\dashv E_a)\odot \frac{\partial}{\partial_b}(E_b\cdot E_c)\\
        \frac{\partial}{\partial_b}(E_b\cdot E_c) &= 
        (\varepsilon\dashv E_b)\odot \frac{\partial}{\partial_b}(E_c)
    \end{align*}
    Furthermore
    \begin{align*}
      \frac{\partial}{\partial_b}(E_a) &\equiv \emptyset\\
      \frac{\partial}{\partial_b}(E_b) &=\{(E_b,\{(y,by)\})\}\\
      \frac{\partial}{\partial_b}(E_c) &\equiv \emptyset
    \end{align*}
    Then
    \begin{align*}
      \frac{\partial}{\partial_b}(E_b\cdot E_c)&\equiv \{(E_b\cdot E_c,\{(y,by)\})\}\\
      \frac{\partial}{\partial_b}(E_a\cdot (E_b \cdot E_c))&\equiv(\varepsilon\dashv E_a)\odot \frac{\partial}{\partial_b}(E_b\cdot E_c)\\
      &\equiv \{ ((\varepsilon\dashv E_a)\cdot E_b\cdot E_c,\{(y,by)\})
    \end{align*}
    And consequently
    \begin{align*}
      \frac{\partial}{\partial_{ab}}(E')&\equiv \{((\varepsilon\dashv E_a)\cdot E_b \cdot E_c)\mid \sim(ax,by)\wedge \sim (by,z),\{(y,by)\})\}
    \end{align*}
    Finally, setting
    \begin{align*}
      E''&=(\varepsilon\dashv E_a)\cdot E_b \cdot E_c)\mid \sim(ax,by)\wedge \sim (by,z)
    \end{align*}
    one can compute
    \begin{align*}
      \frac{\partial}{\partial_{abc}}(E)& \equiv \frac{\partial}{\partial_{c}}(E'')\\
      & \equiv \{(((\varepsilon\dashv E_a)\cdot (\varepsilon\dashv E_b) \cdot E_c)\mid \sim(ax,by)\wedge \sim (by,cz),\{(z,cz)\})\}
    \end{align*}
  \end{example}
  
  \begin{theorem}
    Let $\mathcal{E}=(\Sigma,\Gamma,\mathcal{P},\mathcal{F})$ be an expression environment and let $E$ be a constrained expression over $\mathcal{E}$. Let $w$ be a word in $\Sigma^+$. Then the following conditions hold:
    \begin{enumerate}
      \item $w^{-1}(L_{I}(E))= \bigcup_{(E',X)\in\frac{\partial}{\partial_w}(E)} L_I(E')$,
      \item $w^{-1}(L(E))= \bigcup_{(E',X)\in\frac{\partial}{\partial_w}(E)} L(E')$.
    \end{enumerate}
  \end{theorem}
  \begin{proof}
    \ 
    \begin{enumerate}
      \item By recurrence over the length of $w$. If $w\in\Sigma$, the condition is satisfied according to Theorem~\ref{thm egal quot deriv part}. Let $w=aw'$ with $a\in\Sigma$ and $w'\in\Sigma^+$. Then $w^{-1}(L_{I}(E))=w'^{-1}(a^{-1}(L_{I}(E)))$. According to Theorem~\ref{thm egal quot deriv part}, it holds that $w'^{-1}(a^{-1}(L_{I}(E)))$ $=$ $w'^{-1}(\bigcup_{(E',X)\in\frac{\partial}{\partial_a}(E)} L_I(E'))$ that equals $\bigcup_{(E',X)\in\frac{\partial}{\partial_a}(E)} w'^{-1}(L_I(E'))$. By recurrence hypothesis, $w'^{-1}(L_I(E'))=\bigcup_{(E'',X')\in\frac{\partial}{\partial_{w'}}(E')} L_I(E'')$.       
      Consequently $\bigcup_{(E',X)\in\frac{\partial}{\partial_a}(E)} w'^{-1}(L_I(E'))= \bigcup_{(E',X)\in\frac{\partial}{\partial_a}(E)} \bigcup_{(E'',X')\in\frac{\partial}{\partial_{w'}}(E')} L_I(E'')$. Since by definition, $\{(E'',X')\in\frac{\partial}{\partial_{w'}}(E')\mid (E',X)\in\frac{\partial}{\partial_a}(E)\}$ is equal to $\{(E'',X')\in\frac{\partial}{\partial_{aw'}}(E)\}$, then 
      \begin{align*}
        \bigcup_{
         \subalign{
          (E',X)&\in\frac{\partial}{\partial_a}(E)\\
          (E'',X')&\in\frac{\partial}{\partial_{w'}}(E')
         }
        } L_I(E'')=\bigcup_{(E',X)\in\frac{\partial}{\partial_w}(E)} L_I(E')
      \end{align*}
      \item Let $u$ be a word in $w^{-1}(L(E))$. By definition, it is equivalent to the fact that there exists an interpretation $I$ such that $u\in w^{-1}(L_I(E))$. We have shown that there exists an interpretation $I$ such that $u\in w^{-1}(L_I(E))$ if and only if there exists an interpretation $I$ such that $u\in \bigcup_{(E',X)\in\frac{\partial}{\partial_w}(E)} L_I(E')$, which is equivalent by definition of $L(E')$ to $u\in \bigcup_{(E',X)\in\frac{\partial}{\partial_w}(E)} L(E')$.
    \end{enumerate}
    \qed
  \end{proof}
  
  \begin{corollary}
    Let $\mathcal{E}=(\Sigma,\Gamma,\mathcal{P},\mathcal{F})$ be an expression environment and let $E$ be a constrained expression over $\mathcal{E}$. Let $w$ be a word in $\Sigma^+$. Then:
        \begin{align*}
          w \in L_{I}(E) &\Leftrightarrow \exists (E',X)\in\frac{\partial}{\partial_w}(E),\ \varepsilon\in L_I(E')\\
          w \in L(E) &\Leftrightarrow \exists (E',X)\in\frac{\partial}{\partial_w}(E),\ \varepsilon\in L(E')
        \end{align*}
  \end{corollary} 
  
  \begin{example}\label{ex derivativ}
    Let us consider the expression $E_1=xb^*y\mid \sim(\mathrm{f}(x),\mathrm{f}(y))$ from Example~\ref{ex exp cons lang}. Let us compute the constrained derivative of $E_1$ w.r.t. the symbol $a$. Since the expression starts with a variable, an assumption has to be made and the process can be expressed as follows: 
    \begin{enumerate}
      \item Maybe the variable $x$ starts with an $a$. In this case, we replace all the occurrences of $x$ by $ax$ except the first one, where the symbol $a$ is erased by the derivation. Hence we get the tuple $(F_1,\{(x,ax)\})$ with $F_1=xb^*y\mid \sim(\mathrm{f}(ax),\mathrm{f}(y))$.
      \item Otherwise the variable $x$ can be replaced by $\varepsilon$, and we try to derive the obtained expression $\varepsilon b^*y\mid \sim(\mathrm{f}(\varepsilon),\mathrm{f}(y))$. The catenation $\varepsilon b^*y$ implies that we have to make the assumption that $y$ starts with the symbol $a$. In this case, we replace all the occurrences of $y$ by $ay$ except the first one, where the symbol $a$ is erased by the derivation. Hence we get the tuple $(F_2,\{(x,\varepsilon),(y,ay)\})$ with $F_2=y\mid \sim(\mathrm{f}(\varepsilon),\mathrm{f}(ay))$.
    \end{enumerate}    
    Hence:
        \begin{align*}
        \frac{\partial}{\partial_a}(E_1)=
      \begin{cases}
          (xb^*y\mid \sim(\mathrm{f}(ax),\mathrm{f}(y)),\{(x,ax)\}),\\
          (y\mid \sim(\mathrm{f}(\varepsilon),\mathrm{f}(ay)),\{(x,\varepsilon),(y,ay)\})\\
        \end{cases}
        \end{align*}
        The constrained derivative of $E_1$ w.r.t. the word $ab$ is obtained by computing the constrained derivative of $F_1$ and $F_2$ w.r.t. $b$:   
        \begin{align*}
        \frac{\partial}{\partial_b}(F_1)=
      \begin{cases}
          (xb^*y\mid \sim(\mathrm{f}(abx),\mathrm{f}(y)),\{(x,bx)\}),\\
          (\varepsilon b^*y\mid \sim(\mathrm{f}(a),\mathrm{f}(y)),\{(x,\varepsilon)\}),\\
          (y\mid \sim(\mathrm{f}(a),\mathrm{f}(by)),\{(x,\varepsilon),(y,by)\})
        \end{cases}\\
        \frac{\partial}{\partial_b}(F_2)=
      \begin{cases}
          (y\mid \sim(\mathrm{f}(\varepsilon),\mathrm{f}(aby)),\{(y,by)\})\\
        \end{cases}
        \end{align*}
        Hence
        \begin{align*}
        \frac{\partial}{\partial_{ab}}(E_1)=
      \begin{cases}
          (xb^*y\mid \sim(\mathrm{f}(abx),\mathrm{f}(y)),\{(x,bx)\}),\\
          (\varepsilon b^*y\mid \sim(\mathrm{f}(a),\mathrm{f}(y)),\{(x,\varepsilon)\}),\\
          (y\mid \sim(\mathrm{f}(a),\mathrm{f}(by)),\{(x,\varepsilon),(y,by)\}),\\
          (y\mid \sim(\mathrm{f}(\varepsilon),\mathrm{f}(aby)),\{(y,by)\})
        \end{cases}
        \end{align*}
    \qed
  \end{example}
  
  \begin{example}
    Let us continue Example~\ref{ex deriv anbncn abc}.
    Let us consider the expression $E$ of Example~\ref{ex anbncn} and its derived term 
    \begin{align*}
    E'''&=(((\varepsilon\dashv E_a)\cdot (\varepsilon\dashv E_b) \cdot E_c)\mid \sim(ax,by)\wedge \sim (by,cz)
    \end{align*}
    w.r.t. $abc$.    
      Let us consider an expression interpretation $J=(\Sigma^*,\mathfrak{G})$ of Example~\ref{ex anbncn} that satisfies
      \begin{align*}
        \mathfrak{G}(\sim)&=\{(w_1,w_2)\in\Sigma^*\mid |w_1|=|w_2|\}
      \end{align*}
      By considering a realization $\mathrm{r}$ that associates $x$, $y$ and $z$ with $\varepsilon$, one can check that
      \begin{align*}
        L_{J,\mathrm{r}}(E_a) &= L_{J,\mathrm{r}}(\varepsilon \dashv a^*)\\
        = L_{J,\mathrm{r}}(E_b) &= L_{J,\mathrm{r}}(\varepsilon \dashv b^*)\\
        = L_{J,\mathrm{r}}(E_c) &= L_{J,\mathrm{r}}(\varepsilon \dashv c^*)\\
        &=\{\varepsilon\}
      \end{align*}
      Then
      \begin{align*}
        L_{J,\mathrm{r}}((\varepsilon\dashv E_a)\cdot (\varepsilon\dashv E_b) \cdot E_c)) &= \{\varepsilon\}
      \end{align*}  
       Furthermore
      \begin{align*}
        \mathrm{eval}_{J,\mathrm{r}}(\sim(ax,by)\wedge \sim (by,cz))&= \mathrm{eval}_{J,\mathrm{r}}(\sim(a,b)\wedge \sim (b,c))\\
        &= (|a|==|b| \ \mathrm{ And }\ |b|==|c|)\\
        &= 1
      \end{align*} 
      Therefore $\varepsilon\in L_{J,\mathrm{r}}(E''')$ and consequently $abc\in L_{J,\mathrm{r}}(E)\subset L_{J}(E)\subset L(E)$.       
  \end{example}
  
  \section{Membership Test of $\varepsilon$ for Constrained Expressions}\label{sec:eps membership test}
  
  In this section, we consider the membership test for the empty word $\varepsilon$. We first consider the case where both the interpretation and the realization are fixed, which is a case where this test is decidable. Then we consider the two other cases and show that they are equivalent to a satisfiability problem.
  
  \subsection{$(I,r)$-Language and $\varepsilon$}
  
  Corollary~\ref{cor i r lang rat} asserts that the $(I,r)$-language denoted by a constrained expression is regular. Consequently, any membership test can be performed \emph{via} the regularization. However, this transformation may be avoided by directly and inductively computing the classical predicate function $\mathrm{Null}$ and embedding the regularization in its computation.
  
  \begin{definition}[$\mathrm{Null}_{I,\mathrm{r}}$ Predicate]
    Let $\mathcal{E}=(\Sigma,\Gamma,\mathcal{P},\mathcal{F})$ be an expression environment and let $E$ be a constrained expression over $\mathcal{E}$. Let $I$ be  an expression interpretation over $\mathcal{E}$ and $\mathrm{r}$ be a $\Gamma$-realization over $I$. The boolean $\mathrm{Null}_{I,\mathrm{r}}(E)$ is defined by:
        \begin{align*}
        \mathrm{Null}_{I,\mathrm{r}}(E)=(\varepsilon\in L_{I,\mathrm{r}}(E))
        \end{align*}
  \end{definition}
  
  \begin{proposition}
    Let $\mathcal{E}=(\Sigma,\Gamma,\mathcal{P},\mathcal{F})$ be an expression environment and let $E$ be a constrained expression over $\mathcal{E}$. Let $I$ be  an expression interpretation over $\mathcal{E}$ and $\mathrm{r}$ be a $\Gamma$-realization over $I$. Then the boolean $\mathrm{Null}_{I,\mathrm{r}}(E)$ is inductively computed as follows:
        \begin{align*}
          \mathrm{Null}_{I,\mathrm{r}}(\alpha)&=(\mathrm{r}(\alpha)==\varepsilon),\\ 
      \mathrm{Null}_{I,\mathrm{r}}(\emptyset)&=0,\\
      \mathrm{Null}_{I,\mathrm{r}}(\alpha\dashv E_1)& =(\mathrm{r}(\alpha)==\varepsilon)\wedge \mathrm{Null}_{I,\mathrm{r}}(E_1),\\
      \mathrm{Null}_{I,\mathrm{r}}(\mathrm{o}(E_1,\ldots,E_k))&=\mathrm{o}'(\mathrm{Null}_{I,\mathrm{r}}(E_1),\ldots,\mathrm{Null}_{I,\mathrm{r}}(E_k)),\\
      \mathrm{Null}_{I,\mathrm{r}}(F\cdot G)&=\mathrm{Null}_{I,\mathrm{r}}(F)\wedge \mathrm{Null}_{I,\mathrm{r}}(G),\\
      \mathrm{Null}_{I,\mathrm{r}}(F^*)&=1,\\
      \mathrm{Null}_{I,\mathrm{r}}(E_1\mid \phi)&=\mathrm{Null}_{I,\mathrm{r}}(E_1)\wedge \mathrm{eval}_{I,\mathrm{r}}(\phi),
        \end{align*}
        where $k$ is any integer, $\mathrm{o}$ is any $k$-ary boolean operator, $\mathrm{o}'$ is the boolean operator associated with $\mathrm{o}$, $E_1,\ldots,E_k$ are any $k$ constrained expression over $\mathcal{E}$, $\alpha$ is any word in $(\Sigma\cup \Gamma)^*$ and $\phi$ is any boolean formula in $\mathcal{P}(\mathcal{F}(\Gamma))$.
  \end{proposition}
  \begin{proof}
    By induction over the structure of constrained expressions.
      \begin{multicols}{2}
        \begin{align*}
          \mathrm{Null}_{I,\mathrm{r}}(\alpha)&=(\mathrm{r}(\alpha)==\varepsilon)\\
          & =(\varepsilon\in \{\mathrm{r}(\alpha)\})\\
          &=(\varepsilon\in L_{I,\mathrm{r}}(\alpha))\\
          \mathrm{Null}_{I,\mathrm{r}}(\emptyset)&=0\\
          &=(\varepsilon\in \emptyset)\\
          &=(\varepsilon\in L_{I,\mathrm{r}}(\emptyset))\\
          \mathrm{Null}_{I,\mathrm{r}}(\alpha\dashv E_1) & =(\mathrm{r}(\alpha)==\varepsilon)\wedge \mathrm{Null}_{I,\mathrm{r}}(E_1)\\
          &  =(\mathrm{r}(\alpha)==\varepsilon)\wedge (\varepsilon\in L_{I,\mathrm{r}}(E_1))\\
          &  =(\varepsilon\in\{\mathrm{r}(\alpha)\})\wedge (\varepsilon\in L_{I,\mathrm{r}}(E_1))\\
          & =(\varepsilon\in L_{I,\mathrm{r}}(\alpha\dashv E_1))\\
          \mathrm{Null}_{I,\mathrm{r}}(\mathrm{o}(E_1,\ldots,E_k)) & =\mathrm{o}'(\mathrm{Null}_{I,\mathrm{r}}(E_1),\ldots,\mathrm{Null}_{I,\mathrm{r}}(E_k))\\
      & =\mathrm{o}'(\varepsilon\in L_{I,\mathrm{r}}(E_1),\ldots,\varepsilon\in L_{I,\mathrm{r}}(E_k))\\
      & =(\varepsilon \in L_{I,\mathrm{r}}(\mathrm{o}(E_1,\ldots,E_k)) )\\
      \mathrm{Null}_{I,\mathrm{r}}(F\cdot G) & =\mathrm{Null}_{I,\mathrm{r}}(F)\wedge \mathrm{Null}_{I,\mathrm{r}}(G)\\
      & =(\varepsilon\in L_{I,\mathrm{r}}(F))\wedge (\varepsilon\in L_{I,\mathrm{r}}(G))\\
      & =(\varepsilon\in L_{I,\mathrm{r}}(F\cdot G))\\
      \mathrm{Null}_{I,\mathrm{r}}(F^*)&=1\\
      &=(\varepsilon\in L_{I,\mathrm{r}}(F^*))\\
      \mathrm{Null}_{I,\mathrm{r}}(E_1\mid \phi) & =\mathrm{Null}_{I,\mathrm{r}}(E_1)\wedge \mathrm{eval}_{I,\mathrm{r}}(\phi)\\ 
      & =\varepsilon\in L_{I,\mathrm{r}}(E_1) \wedge \mathrm{eval}_{I,\mathrm{r}}(\phi)\\
      & =\varepsilon\in L_{I,\mathrm{r}}(E_1\mid\phi)
        \end{align*}
      \end{multicols}
    \qed
  \end{proof}
  
  \subsection{General Cases}
  
  As was the case for derivating, the computation of the $\mathrm{Null}$ predicate  needs assumptions to be made. However, we only need here to determine which variable symbols have to be transformed into the empty word. Since we need to "erase" several symbols at the same time, we define several notations to perform the corresponding substitutions.
  
  Let $\mathcal{E}=(\Sigma,\Gamma,\mathcal{P},\mathcal{F})$ be an expression environment.
  We denote by $\mathrm{Sub}(\Gamma,\Sigma)$ the set of the functions from $\Gamma$ to $(\Sigma\cup\Gamma)^*$. 
  Let $X\subset \Gamma$. We denote by $\mathrm{S}_{X\leftarrow\varepsilon}$ the substitution defined by 
  $\mathrm{S}_{X\leftarrow\varepsilon}(y)=
    \begin{cases}
        \varepsilon & \text{ if } y\in X,\\
        y & \text{ otherwise.}
      \end{cases}$

  Let $\alpha$ be a word in $\Gamma^*$. We denote by $\Gamma_\alpha$ the subset $\{y\in\Gamma\mid \exists u,v\in\Gamma^*, \alpha=uyv\}$ of $\Gamma$.  
  We denote by $\top$ (resp. $\bot$) the $0$-ary boolean operator True (resp. False).  
      Given a formula $\phi$, we denote by $\phi_{X\leftarrow\varepsilon}$ the formula inductively computed by:
        \begin{align*}
          \phi_{X\leftarrow \varepsilon}=
        \begin{cases}
            \phi & \text{ if }X=\emptyset,\\
            (\phi_{x\leftarrow \varepsilon})_{X'\leftarrow \varepsilon} & \text{ if } X=X'\cup\{x\}.\\
          \end{cases}
        \end{align*}
    
  The general computation of the $\mathrm{Null}$ predicate  takes into
account two aspects of an expression: first, it needs to consider the expression itself, in order to determine if $\varepsilon$ may appear; but secondly, it has to consider the fact that several formulae that appear in the expression have to be satisfied, otherwise the language may be empty. Hence we first compute a particular indicator set, made of 
tuples composed of a set of variable symbols that need to be erased and a formula that needs to be satisfied.
  
  \begin{definition}[$\mathcal{S}^{\varepsilon}(E)$]
    Let $\mathcal{E}=(\Sigma,\Gamma,\mathcal{P},\mathcal{F})$ be an expression environment and let $E$ be a constrained expression over $\mathcal{E}$. We denote by $\mathcal{S}^{\varepsilon}(E)$ the subset of $2^\Gamma\times \mathcal{P}(\mathcal{F}(\Gamma))$ inductively defined by:
    \begin{multicols}{2}
        \begin{align*}
          \mathcal{S}^{\varepsilon}(\alpha)& =
        \begin{cases}
            \emptyset & \text{ if } \alpha\notin \Gamma^*,\\
             \{(\Gamma_\alpha,\top)\} & \text{otherwise},
          \end{cases}\\
        \mathcal{S}^{\varepsilon}(\emptyset)&=\emptyset,\\
        \mathcal{S}^{\varepsilon}(\alpha\dashv E_1)&=
        \begin{cases}
            \emptyset & \text{ if } \alpha\notin \Gamma^*,\\
            \{(\Gamma_\alpha,\top)\} \otimes \mathcal{S}^{\varepsilon}(E_1) & \text{otherwise},\\
          \end{cases}\\
        \mathcal{S}^{\varepsilon}(E_1+E_2)&=\mathcal{S}^{\varepsilon}(E_1) \cup \mathcal{S}^{\varepsilon}(E_2)\\
        \mathcal{S}^{\varepsilon}(E_1\cdot E_2)&=
       \mathcal{S}^{\varepsilon}(E_1)\otimes \mathcal{S}^{\varepsilon}(E_2)\\
       \mathcal{S}^{\varepsilon}(E_1^*)&=\{(\emptyset,\top)\}\\
       \mathcal{S}^{\varepsilon}(E_1\mid \phi)&=
      \bigcup_{(X,\psi)\in \mathcal{S}^{\varepsilon}(E_1)} \{(X,(\phi\wedge\psi)_{X\leftarrow\varepsilon})\}
        \end{align*}
      \end{multicols}
        where for any two subsets $\mathcal{S}_1$, $\mathcal{S}_2$ of $2^\Gamma\times \mathcal{P}(\mathcal{F}(\Gamma))$, $\mathcal{S}_1\otimes\mathcal{S}_2$ $=
          \bigcup_{
           \subalign{
            (X_1,\phi_1)&\in\mathcal{S}_1\\
            (X_2,\phi_2)&\in\mathcal{S}_2
           }
          }
          \{(X_1\cup X_2,(\phi_1\wedge \phi_2)_{(X_1\cup X_2)\leftarrow\varepsilon})\}$.
  \end{definition}
  
  Using this previous indicator set, it can be shown that the computation of the different $\mathrm{Null}$ predicates is equivalent to different satisfiability problems. 
  
  \begin{theorem}\label{thm caract null i r}
    Let $\mathcal{E}=(\Sigma,\Gamma,\mathcal{P},\mathcal{F})$ be an expression environment and let $E$ be a constrained expression over $\mathcal{E}$. Let $I$ be  an expression interpretation over $\mathcal{E}$. Let $\mathrm{r}$ be a realisation in $\mathrm{Real}_\Gamma(I)$. Then the  following two conditions are equivalent:
    \begin{itemize}
        \item $\mathrm{Null}_{I,\mathrm{r}}(E)$
        \item there exists $(X,\phi)$ in $\mathcal{S}^{\varepsilon}(E)$ such that the  following two conditions are satisfied:
          \begin{itemize}
            \item $\forall x\in\Gamma$, $x\in X$ $\Rightarrow$ $\mathrm{r}(x)=\varepsilon$,
            \item $\mathrm{eval}_{(I,\mathrm{r})}(\phi)=1$.
          \end{itemize}
     \end{itemize}
  \end{theorem}
  \begin{proof}
    By induction over the structure of $E$.
    
    Let us say that a tuple $(X,\phi)$ in $2^\Gamma\times \mathcal{P}(\mathcal{F}(\Gamma))$ satisfies the condition $\mathbb{C}$ if the  following two conditions are satisfied:
          \begin{itemize}
            \item $\forall x\in\Gamma$, $x\in X$ $\Rightarrow$ $\mathrm{r}(x)=\varepsilon$,
            \item $\mathrm{eval}_{(I,\mathrm{r})}(\phi)=1$.
          \end{itemize}
          
    Hence the second condition of the equivalence can be rephrased as "there exists a tuple $(X,\phi)$ in $\mathcal{S}^{\varepsilon}(E)$ that satisfies $\mathbb{C}$", formally denoted by $\exists(X,\phi)\in \mathcal{S}^{\varepsilon}(E)\mid\mathbb{C}$.
        \begin{align*}
        \mathrm{Null}_{I,\mathrm{r}}(\alpha) & \Longleftrightarrow \mathrm{r}(\alpha)=\varepsilon \\
        & \Longleftrightarrow \alpha\in\Gamma^*\ \wedge\ x\in\Gamma_\alpha \Rightarrow \mathrm{r}(x)=\varepsilon\\
        & \Longleftrightarrow(\Gamma_\alpha,\top)\in \mathcal{S}^{\varepsilon}(\alpha) \mid \mathbb{C}\\
        & \Longleftrightarrow \exists (X,\phi)\in \mathcal{S}^{\varepsilon}(\alpha) \mid \mathbb{C}\\
        \mathrm{Null}_{I,\mathrm{r}}(\emptyset)&=0\text{ and }\mathcal{S}^{\varepsilon}(\emptyset)=\emptyset\\
        \mathrm{Null}_{I,\mathrm{r}}(\alpha\dashv E_1) & \Longleftrightarrow \mathrm{r}(\alpha)=\varepsilon\ \wedge\ \varepsilon\in L_{I,\mathrm{r}}(E_1)\\
        & \Longleftrightarrow \mathrm{r}(\alpha)=\varepsilon\ \wedge\ \exists(X,\phi)\in\mathcal{S}^{\varepsilon}(E_1)\mid\mathbb{C}\\
        & \Longleftrightarrow \alpha\in\Gamma^*\ \wedge\ x\in\Gamma_\alpha\Rightarrow\mathrm{r}(x)=\varepsilon\ \wedge\ \exists(X,\phi)\in\mathcal{S}^{\varepsilon}(E_1)\mid\mathbb{C}\\
        & \Longleftrightarrow (\Gamma_\alpha,\top)\mid\mathbb{C}\ \wedge\ \exists(X,\phi)\in\mathcal{S}^{\varepsilon}(E_1)\mid\mathbb{C}\\
        & \Longleftrightarrow (\Gamma_\alpha\cup X,(\top\ \wedge\ \phi)_{X\leftarrow\varepsilon})\in \{(\Gamma_\alpha,\top)\}\otimes \mathcal{S}^{\varepsilon}(E_1)\mid\mathbb{C}\\
        & \Longleftrightarrow \exists(X,\phi)\in\mathcal{S}^{\varepsilon}(\alpha\dashv E_1)\mid\mathbb{C}\\
        \mathrm{Null}_{I,\mathrm{r}}(E_1+E_2) & \Longleftrightarrow \varepsilon\in L_{I,\mathrm{r}}(E_1)\ \vee\ \varepsilon\in L_{I,\mathrm{r}}(E_2)\\
        & \Longleftrightarrow \exists(X,\phi)\in\mathcal{S}^{\varepsilon}(E_1)\mid\mathbb{C} \vee \exists(X,\phi)\in\mathcal{S}^{\varepsilon}(E_2)\mid\mathbb{C}\\
        & \Longleftrightarrow \exists(X,\phi)\in\mathcal{S}^{\varepsilon}(E_1)\cup\mathcal{S}^{\varepsilon}(E_2)\mid\mathbb{C}\\
        & \Longleftrightarrow \exists(X,\phi)\in\mathcal{S}^{\varepsilon}(E_1+E_2)\mid\mathbb{C}\\
        \mathrm{Null}_{I,\mathrm{r}}(E_1\cdot E_2) & \Longleftrightarrow \varepsilon\in L_{I,\mathrm{r}}(E_1)\ \wedge\ \varepsilon\in L_{I,\mathrm{r}}(E_2)\\
        & \Longleftrightarrow \exists(X_1,\phi_1)\in\mathcal{S}^{\varepsilon}(E_1)\mid\mathbb{C} \wedge \exists(X_2,\phi_2)\in\mathcal{S}^{\varepsilon}(E_2)\mid\mathbb{C}\\
        & \Longleftrightarrow \exists(X_1,\phi_1)\in\mathcal{S}^{\varepsilon}(E_1), \exists(X_2,\phi_2)\in\mathcal{S}^{\varepsilon}(E_2)\mid\\
        & \ \ \ \ \ \ \ \ (X_1\cup X_2,(\phi_1\wedge\phi_2)_{X_1\cup X_2\leftarrow\varepsilon})\mid\mathbb{C}\\
        & \Longleftrightarrow\exists(X,\phi)\in\mathcal{S}^{\varepsilon}(E_1)\otimes \mathcal{S}^{\varepsilon}(E_2)\mid\mathbb{C}\\
      & \Longleftrightarrow \exists(X,\phi)\in\mathcal{S}^{\varepsilon}(E_1\cdot E_2)\mid\mathbb{C}\\
      \mathrm{Null}_{I,\mathrm{r}}(E_1^*)&=1\text{ and }(\emptyset,\top)\in\mathcal{S}^{\varepsilon}(E_1^*)\mid \mathbb{C}\\
      \mathrm{Null}_{I,\mathrm{r}}(E_1\mid\phi) & \Longleftrightarrow\varepsilon\in L_{I,\mathrm{r}}(E_1)\ \wedge\ \mathrm{eval}_{I,\mathrm{r}}(\phi)\\
      & \Longleftrightarrow \exists(X',\phi')\in\mathcal{S}^{\varepsilon}(E_1)\mid\mathbb{C}\ \wedge\ \mathrm{eval}_{I,\mathrm{r}}(\phi)\\
      & \Longleftrightarrow \exists(X',\phi')\in\mathcal{S}^{\varepsilon}(E_1)\mid\mathbb{C}\ \wedge\ \mathrm{eval}_{I,\mathrm{r}}((\phi'\wedge\phi)_{X'\leftarrow\varepsilon})\\
      & \Longleftrightarrow \exists(X',\phi')\in\mathcal{S}^{\varepsilon}(E_1\mid\phi)\mid\mathbb{C}\\
        \end{align*}    
    \qed
  \end{proof}

  \begin{definition}[$\mathrm{Null}_{I}$ Predicate]
    Let $\mathcal{E}=(\Sigma,\Gamma,\mathcal{P},\mathcal{F})$ be an expression environment and let $E$ be a constrained expression over $\mathcal{E}$. Let $I$ be  an expression interpretation over $\mathcal{E}$. The boolean $\mathrm{Null}_{I}(E)$ is defined by:
        \begin{align*}
       \mathrm{Null}_{I}(E)&=(\varepsilon\in L_{I}(E)).
        \end{align*}
  \end{definition}
  
  \begin{corollary}\label{cor nullIE lien sat}[of Theorem~\ref{thm caract null i r}]
    Let $\mathcal{E}=(\Sigma,\Gamma,\mathcal{P},\mathcal{F})$ be an expression environment and let $E$ be a constrained expression over $\mathcal{E}$. Let $I$ be  an expression interpretation over $\mathcal{E}$. Then the  following two conditions are equivalent :
    \begin{itemize}
        \item $\mathrm{Null}_{I}(E)=1$,
        \item there exists $(X,\phi)$ in $\mathcal{S}^{\varepsilon}(E)$ and $\mathrm{r}$ in $\mathrm{Real}_\Gamma(I)$ such that $\mathrm{eval}_{(I,\mathrm{r})}(\phi)=1$.
     \end{itemize}
  \end{corollary}

  \begin{definition}[$\mathrm{Null}$ Predicate]
    Let $\mathcal{E}=(\Sigma,\Gamma,\mathcal{P},\mathcal{F})$ be an expression environment and let $E$ be a constrained expression over $\mathcal{E}$. Let $I$ be  an expression interpretation over $\mathcal{E}$. The boolean $\mathrm{Null}(E)$ is defined by:
        \begin{align*}
       \mathrm{Null}(E)=(\varepsilon\in L(E)).
        \end{align*}
  \end{definition}
  
  \begin{corollary}\label{cor nullE lien sat}[of Theorem~\ref{thm caract null i r}]
    Let $\mathcal{E}=(\Sigma,\Gamma,\mathcal{P},\mathcal{F})$ be an expression environment and let $E$ be a constrained expression over $\mathcal{E}$. Then the  following two conditions are equivalent :
    \begin{itemize}
        \item $\mathrm{Null}(E)=1$
        \item there exists $(X,\phi)$ in $\mathcal{S}^{\varepsilon}(E)$, $I$ in $\mathrm{Int}(\mathcal{E})$ and $\mathrm{r}$ in $\mathrm{Real}_\Gamma(I)$ such that $\mathrm{eval}_{(I,\mathrm{r})}(\phi)=1$.
     \end{itemize}
  \end{corollary}
  
  \begin{example}
    Let us consider the expression $E_1=xb^*y\mid \sim(\mathrm{f}(x),\mathrm{f}(y))$ and its constrained derivative w.r.t. $ab$ (Example~\ref{ex derivativ}):
        \begin{align*}
    \frac{\partial}{\partial_{ab}}(E_1)=
    \begin{cases}
          (xb^*y\mid \sim(\mathrm{f}(abx),\mathrm{f}(y)),\{(x,bx)\}),\\
          (\varepsilon b^*y\mid \sim(\mathrm{f}(a),\mathrm{f}(y)),\{(x,\varepsilon)\}),\\
          (y\mid \sim(\mathrm{f}(a),\mathrm{f}(by)),\{(x,\varepsilon),(y,by)\}),\\
          (y\mid \sim(\mathrm{f}(\varepsilon),\mathrm{f}(aby)),\{(y,by)\})\\
        \end{cases}
        \end{align*}      
      In order to decide whether $ab$ belongs to $L(E_1)$, let us test whether there is an expression $E'$ in $\frac{\partial}{\partial_{ab}}(E_1)$ such that $\mathrm{Null}(E')=1$.
      Let us consider the expression $E'_1=xb^*y\mid \sim(\mathrm{f}(abx),\mathrm{f}(y))$. 
      Since epsilon should be matched, $xb*y$ has to be nullable and both $x$ and $y$ have to be realized as epsilon. 
      Consequently, the boolean formula, which has to be satisfied, is transformed into $\sim(\mathrm{f}(ab),\mathrm{f}(\varepsilon))$. This step is exactly what the set $\mathcal{S}^{\varepsilon}(E'_1)$ computes:
        \begin{align*}
        \mathcal{S}^{\varepsilon}(xb^*y) & =\mathcal{S}^{\varepsilon}(x)\otimes \mathcal{S}^{\varepsilon}(b^*) \otimes \mathcal{S}^{\varepsilon}(y)\\
          & =\{(\{x\},\top)\}\otimes \{(\emptyset,\top)\} \otimes \{(\{y\},\top)\}\\
          & =\{(\{x,y,\top)\}\\
          \mathcal{S}^{\varepsilon}(E'_1) & =\{(\{x,y\},\sim(\mathrm{f}(ab),\mathrm{f}(\varepsilon)))\}
        \end{align*}
        If there exists an interpretation $\mathrm{I}$ and a realization $\mathrm{r}$ associating $x$ and $y$ with $\varepsilon$ such that $\mathrm{eval}_{(I,\mathrm{r})}(\sim(\mathrm{f}(ab),\mathrm{f}(\varepsilon)))=1$, then $\varepsilon$ belongs to $L(E)$. Finally, considering the same interpretation and a realization $\mathrm{r}'$ associating $ab$ with $x$ and $\varepsilon$ with $y$, $\mathrm{eval}_{(I,\mathrm{r}')}(\sim(\mathrm{f}(x),\mathrm{f}(y)))=\mathrm{eval}_{(I,\mathrm{r})}(\sim(\mathrm{f}(ab),\mathrm{f}(\varepsilon)))=1$ and then $ab$ belongs to the $(\mathrm{I},\mathrm{r}')$-language denoted by $E_1$, \emph{i.e.} $[E_1]_{\mathrm{I},\mathrm{r}'}=abb^*\varepsilon\mid \sim(\mathrm{f}(ab),\mathrm{f}(\varepsilon))$.
      \qed
  \end{example}
  
  \section{Decidability Considerations}\label{sec:decidab}
  
  The previous section (Corollary~\ref{cor nullIE lien sat} and Corollary~\ref{cor nullE lien sat}) shows that the membership test is equivalent to a classical satisfiability problem. Moreover, it is well-known that such a problem can be undecidable when the interpretation is fixed but there is no realization.
  
   \begin{theorem}
     Let $\mathcal{E}=(\Sigma,\Gamma,\mathcal{P},\mathcal{F})$ be an expression environment and $\phi$ be a boolean formula in $\mathcal{P}(\mathcal{F}(\Gamma))$. Then there exists an interpretation $I$ in $\mathrm{Int}(\mathcal{E})$ such that :
     \begin{align*}
      \text{to determine whether or not there exists a realization $\mathrm{r}$ in $\mathrm{Real}_\Gamma(I)$ satisfying $\mathrm{eval}_{(I,\mathrm{r})}(\phi)=1$ is undecidable.}
     \end{align*}
   \end{theorem}   
   \begin{proof}
      Let $k$ be an integer. Let $\mathcal{E}=(\Sigma,\{x_1,\ldots,x_k\},\mathcal{P},\mathcal{F})$. Let $\mathbb{S}$ be a system of diophantine equations with $k$ variables and $P$ be a symbol in $\mathcal{P}_k$.      
     Let us consider the expression interpretation $I=(\Sigma^*,\mathfrak{F})$ such that $\mathfrak{F}(P)=\{(w_1,\ldots,w_k)\mid (|w_1|,\ldots,|w_k|)$ $\text{ is a solution of } \mathbb{S}\}$.
     Let $\phi= P(x_1,\ldots,x_k)$. 
     Then there exists a solution $(n_1,\ldots,n_k)$ for $\mathbb{S}$ if and only if there exists a realization $\mathrm{r}$ that associates for any integer $j$ in $\{1,\ldots,k\}$ the variable $x_j$ with a word $w_j$ of length $n_j$ such that $\mathrm{eval}_{I,\mathrm{r}}(\phi)=1$.
     The solvability of diophantine systems (a.k.a. the tenth Hilbert problem) has been proved to be undecidable by Matiyasevich~\cite{My93}. 
     Hence to determine whether or not "there exists a realization $\mathrm{r}$ in $\mathrm{Real}_\Gamma(I)$ such that $\mathrm{eval}_{(I,\mathrm{r})}(\phi)=1$" is undecidable.
     \qed
   \end{proof}
  
  However, given a boolean formula $\phi$ in $\mathcal{P}(\mathcal{F}(\Gamma))$, to determine whether or not "there exists an interpretation $I$ in $\mathrm{Int}(\mathcal{E})$ and a realization $\mathrm{r}$ in $\mathrm{Real}_\Gamma(I)$ such that $\mathrm{eval}_{(I,\mathrm{r})}(\phi)=1$" is decidable, only using propositional logic.
  
  \begin{theorem}\label{thm sat int exp avc rien est deci}
     Let $\mathcal{E}=(\Sigma,\Gamma,\mathcal{P},\mathcal{F})$ be an expression environment. Given a boolean formula $\phi$ in $\mathcal{P}(\mathcal{F}(\Gamma))$, to determine whether or not there exists an interpretation $I$ in $\mathrm{Int}(\mathcal{E})$ and a realization $\mathrm{r}$ in $\mathrm{Real}_\Gamma(I)$ such that $\mathrm{eval}_{(I,\mathrm{r})}(\phi)=1$ is decidable.
   \end{theorem} 
   
   The next subsections are devoted to proving Theorem~\ref{thm sat int exp avc rien est deci}. We first show that any boolean formula can be transformed into a propositional formula (which is a boolean formula with only $0$-ary predicate symbols). Then we show that any formula is equisatisfiable to its propositionnal form whenever there exists an evaluation which is an injection. We finally show that any formula admits an equivalent formula such that an injection exists.
   
  \subsection{Propositionalisation}
  
  The propositionalisation of a boolean formula is performed by replacing any predicate by a unique symbol; in fact, any predicate appearing in the formula is considered as a new symbol.
  
  \begin{definition}[Propositionalisation]\label{def propositionalisation}
    Let $\mathcal{E}=(\Sigma,\Gamma,\mathcal{P},\mathcal{F})$ be an expression environment. The \emph{propositionalisation} of a boolean formula $\phi$ in $\mathcal{P}(\mathcal{F}(\Gamma))$ is the transformation $T$ inductively defined as follows:
        \begin{align*}
      T(o(\phi_1,\ldots,\phi_k))&=o(T(\phi_1),\ldots,T(\phi_k))\\
      T(P(t_1,\ldots,t_k))&=P_{(t_1,\ldots,t_k)},
        \end{align*}
        where $k$ is any integer, $P$ is any predicate symbol in $\mathcal{P}_k$, $t_1,\ldots,t_k$ are any $k$ elements in $\mathcal{F}(\Gamma)$, $\mathrm{o}$ is any $k$-ary boolean operator associated with a mapping $\mathrm{o}'$ from $\{0,1\}^k$ to $\{0,1\}$ and $\phi_1,\ldots,\phi_k$ are any $k$ boolean formulae in $\mathcal{P}(\mathcal{F}(\Gamma))$. The symbol $P_{(t_1,\ldots,t_k)}$ is the \emph{propositional predicate symbol associated with} the term $P(t_1,\ldots,t_k)$.
  \end{definition}

  \begin{definition}[Propositional Alphabet]\label{def propositional alphabet}
    Let $\mathcal{E}=(\Sigma,\Gamma,\mathcal{P},\mathcal{F})$ be an expression environment. The \emph{propositional alphabet} of a boolean formula $\phi$ in $\mathcal{P}(\mathcal{F}(\Gamma))$ is the set $\mathcal{P}'(\phi)$ inductively defined as follows:
        \begin{align*}
        \mathcal{P}'(o(\phi_1,\ldots,\phi_k))&=\mathcal{P}'(\phi_1)\cup\cdots\cup \mathcal{P}'(\phi_k),\\
        \mathcal{P}'(P(t_1,\ldots,t_k))&=\{P_{(t_1,\ldots,t_k)}\},
        \end{align*}
        where $k$ is any integer, $P$ is any predicate symbol in $\mathcal{P}_k$, $t_1,\ldots,t_k$ are any $k$ elements in $\mathcal{F}(\Gamma)$, $\mathrm{o}$ is any $k$-ary boolean operator associated with a mapping $\mathrm{o}'$ from $\{0,1\}^k$ to $\{0,1\}$ and $\phi_1,\ldots,\phi_k$ are any $k$ boolean formulae in $\mathcal{P}(\mathcal{F}(\Gamma))$.
  \end{definition}
  
  \begin{proposition}
    Let $\mathcal{E}=(\Sigma,\Gamma,\mathcal{P},\mathcal{F})$ be an expression environment. Let $\phi$ be a boolean formula in $\mathcal{P}(\mathcal{F}(\Gamma))$. Then:
        \begin{align*}
        T(\phi)\text{ is a boolean formula in }(\mathcal{P}'(\phi))(\emptyset).
        \end{align*}    
    Furthermore, $\mathcal{P}'(\phi)$ is a finite set.
  \end{proposition}
  \begin{proof}
    Inductively deduced from Definition~\ref{def propositionalisation} and from Definition~\ref{def propositional alphabet}.
    \qed
  \end{proof}
  
  Two of the main interests of these propositional formulae are that \textbf{(I)} they do not need realization to be evaluated (since there is no variable symbols nor terms) and \textbf{(II)} their satisfiability is decidable, using truth tables for example.
  
  Let us now show that the Propositionalisation may produce an equisatisfiable formula.
  
  \subsection{Equisatisfiability of the Propositionalisation}
  
  Once the propositionalisation has been applied over a formula, it can be determined if the obtained formula is satisfiable. This leads to two cases corresponding to the following propositions.
  
   \begin{proposition}\label{prop si tphi contra phi aussi}
    Let $\mathcal{E}=(\Sigma,\Gamma,\mathcal{P},\mathcal{F})$ be an expression environment. Let $\phi$ be a boolean formula in $\mathcal{P}(\mathcal{F}(\Gamma))$. Let $I$ be an interpretation in $\mathrm{Int}(\mathcal{E})$ and $\mathrm{r}$ be a realization in $\mathrm{Real}_\Gamma(I)$. Let $I'=(\Sigma^*,\mathfrak{F}')$ be the expression interpretation over $\mathcal{E}'$ such that for any symbol $P_{(t_1,\ldots,t_k)}$ in $\mathcal{P}'(\phi)$, $\mathrm{eval}_{I'}(P_{(t_1,\ldots,t_k)})=\mathrm{eval}_{I,r}(P(t_1,\ldots,t_k))$.
    Then:
        \begin{align*}
        \mathrm{eval}_{I,r}(\phi)=\mathrm{eval}_{I'}(T(\phi)).
        \end{align*}
  \end{proposition}
  \begin{proof}
    By induction over the structure of $\phi$.
    
    If $\phi=P(t_1,\ldots,t_k)$, then $\mathrm{eval}_{I,r}(P(t_1,\ldots,t_k))=\mathrm{eval}_{I'}(P_{(t_1,\ldots,t_k)})$.
    
    If $\phi=o(\phi_1,\ldots,\phi_k)$, then 
        \begin{align*}
        \mathrm{eval}_{I,r}(o(\phi_1,\ldots,\phi_k)) & = o'(\mathrm{eval}_{I,r}(\phi_1),\ldots,\mathrm{eval}_{I,r}(\phi_k))\\
      & = o'(\mathrm{eval}_{I'}(T(\phi_1)),\ldots,\mathrm{eval}_{I'}(T(\phi_k)))\\
      & = \mathrm{eval}_{I'}(o(T(\phi_1),\ldots,T(\phi_k))\\
      & = \mathrm{eval}_{I'}(T(\phi)).
        \end{align*}
    \qed
  \end{proof}
  
  \begin{corollary}\label{cor tphi contra phi aussi}
    Let $\mathcal{E}=(\Sigma,\Gamma,\mathcal{P},\mathcal{F})$ be an expression environment. Let $\phi$ be a boolean formula in $\mathcal{P}(\mathcal{F}(\Gamma))$. Then:
    \begin{itemize}
      \item If $T(\phi)$ is a contradiction, so is $\phi$.
      \item If $T(\phi)$ is a tautology, so is $\phi$.
    \end{itemize}
  \end{corollary}
  
  However, the satisfiability of $T(\phi)$, when it is not a tautology, is not sufficient to conclude over the satisfiability of $\phi$. 
  Indeed, 
  it can happen that two distinct predicates in $T(\phi)$ have to be evaluated differently while the associated predicates cannot be in $\phi$. 
  As an example, consider the formulae $\phi=P((x\cdot y)\cdot z)\wedge \neg P(x\cdot(y\cdot z))$ and $T(\phi)=P_{((x\cdot y)\cdot z)}\wedge \neg P_{(x\cdot(y\cdot z))}$. The formula $T(\phi)$ is satisfiable when $P_{((x\cdot y)\cdot z)}$ is true and  $P_{(x\cdot(y\cdot z))}$ is not. However, whatever the realization considered, $P((x\cdot y)\cdot z)$ and $P(x\cdot(y\cdot z))$ will always be equi-evaluated. Let us formally define the notion of injection that separates two distinct terms while evaluating.  
    
  Let $\mathcal{E}=(\Sigma,\Gamma,\mathcal{P},\mathcal{F})$ be an expression environment. Let $\phi$ be a boolean formula in $\mathcal{P}(\mathcal{F}(\Gamma))$. The set of the terms of $\phi$ is the set $\mathrm{Term}(\phi)$ inductively defined by:
        \begin{align*}
          \mathrm{Term}(o(\phi_1,\ldots,\phi_k))&=\mathrm{Term}(\phi_1)\cup\cdots\cup \mathrm{Term}(\phi_k),\\
          \mathrm{Term}(P(t_1,\ldots,t_k))&=\{t_1,\ldots,t_k\},
        \end{align*}
        where $k$ is any integer, $P$ is any predicate symbol in $\mathcal{P}_k$, $t_1,\ldots,t_k$ are any $k$ elements in $\mathcal{F}(\Gamma)$, $\mathrm{o}$ is any $k$-ary boolean operator associated with a mapping $\mathrm{o}'$ from $\{0,1\}^k$ to $\{0,1\}$ and $\phi_1,\ldots,\phi_k$ are any $k$ boolean formulae in $\mathcal{P}(\mathcal{F}(\Gamma))$.
    
  \begin{definition}[Injection]
    Let $\mathcal{E}=(\Sigma,\Gamma,\mathcal{P},\mathcal{F})$ be an expression environment. Let $T$ be a subset of $\mathcal{F}(\Gamma)$. Let $I$ be an interpretation in $\mathrm{Int}(\mathcal{E})$ and $\mathrm{r}$ be a realization in $\mathrm{Real}_\Gamma(I)$. The function $\mathrm{eval}_{I,r}$ is said to be an injection of $T$ in $\Sigma^*$ if:
        \begin{align*}
          \text{for any two terms $t_1$ and $t_2$ in  $T$, $\mathrm{eval}_{I,r}(t_1)\neq \mathrm{eval}_{I,r}(t_2)$}.
        \end{align*}
  \end{definition}
  
  Given that such an evaluation exists, let us show that the propositionalisation preserves the satisfiability.
  
  \begin{proposition}\label{prop si tphi sat phi aussi}
    Let $\mathcal{E}=(\Sigma,\Gamma,\mathcal{P},\mathcal{F})$ be an expression environment. 
    Let $\phi$ be a boolean formula in $\mathcal{P}(\mathcal{F}(\Gamma))$. 
    Let $I=(\Sigma^*,\mathfrak{F})$ be an interpretation in $\mathrm{Int}(\mathcal{E})$ and $\mathrm{r}$ be a realization in $\mathrm{Real}_\Gamma(I)$ such that $\mathrm{eval}_{I,r}$ is an injection of $\mathrm{Term}(\phi)$ in $\Sigma^*$. 
    Let $\mathcal{E}'=(\Sigma,\Gamma,\mathcal{P}'(\phi),\emptyset)$. 
    Let $I'=(\Sigma^*,\mathfrak{F}')$ be an expression interpretation over $\mathcal{E}'$.    
    Let $I''=(\Sigma^*,\mathfrak{F}'')$ be an expression interpretation over $\mathcal{E}$ satisfying the  following two conditions:
      \begin{itemize}
        \item for any function symbol $f$ in $F_k$, $\mathfrak{F}''(f)= \mathfrak{F}(f)$,
        \item for any predicate symbol $P_{(t_1,\ldots,t_k)}$ in $\mathcal{P}'(\phi)$, $\mathrm{eval}_{I'}(P_{(t_1,\ldots,t_k)})=1 \Leftrightarrow (\mathrm{eval}_{\mathrm{I}'',r}(t_1),$ $\ldots,$ $\mathrm{eval}_{\mathrm{I}'',r}(t_k))\in \mathfrak{F}''(P)$.
      \end{itemize}
    
    Then:
        \begin{align*}
        \text{$\mathrm{eval}_{I'',r}$ is an injection of $\mathrm{Term}(\phi)$ in $\Sigma^*$ such that $\mathrm{eval}_{I'',r}(\phi)=\mathrm{eval}_{I'}(T(\phi))$.} 
        \end{align*}   
  \end{proposition}
  \begin{proof} 
  \begin{enumerate}
    \item Let us show that $\mathrm{eval}_{I'',r}$ is an injection of $\mathrm{Term}(\phi)$ in $\Sigma^*$. Let $t$ be a term in $\mathcal{F}(\Gamma)$. 
    \begin{enumerate}
      \item\label{it a prop si tphi} Let us show by induction over the structure of $t$ that $\mathrm{eval}_{I'',r}(t)=\mathrm{eval}_{I,r}(t)$.
      \begin{enumerate}
        \item If $t=x$ in $\Gamma$, then $\mathrm{eval}_{I'',r}(x)=r(x)=\mathrm{eval}_{I,r}(x)$.
        \item Let us suppose that $t=f(t_1,\ldots,t_k)$ with $f$ any $k$-ary function symbol in $\mathcal{F}_k$ and $t_1,\ldots,t_k$ any $k$ terms in $\mathcal{F}(\Gamma)$. Then:
        \begin{align*}
        \mathrm{eval}_{I'',r}(f(t_1,\ldots,t_k))= x_{k+1} & \Leftrightarrow (\mathrm{eval}_{I'',r}(t_1),\ldots,\mathrm{eval}_{I'',r}(t_k),x_{k+1})\in \mathfrak{F}''(f)\\
        & \Leftrightarrow (\mathrm{eval}_{I'',r}(t_1),\ldots,\mathrm{eval}_{I'',r}(t_k),x_{k+1})\in \mathfrak{F}(f) & \text{(by definition of $\mathfrak{F}''$)}\\ 
        & \Leftrightarrow (\mathrm{eval}_{I,r}(t_1),\ldots,\mathrm{eval}_{I,r}(t_k),x_{k+1})\in \mathfrak{F}(f) &\text{(induction hypothesis)}\\
        & \Leftrightarrow \mathrm{eval}_{I,r}(f(t_1,\ldots,t_k))= x_{k+1}
        \end{align*}  
      \end{enumerate}  
      \item As a direct consequence of Item~\ref{it a prop si tphi}, since $\mathrm{eval}_{I,r}$ is an injection of $\mathrm{Term}(\phi)$ in $\Sigma^*$, so is $\mathrm{eval}_{I'',r}$.
    \end{enumerate}
    \item  Let us show by induction over $\phi$ that $\mathrm{eval}_{I'',r}(\phi)=\mathrm{eval}_{I'}(T(\phi))$.
    \begin{enumerate}
      \item If $\phi=P(t_1,\ldots,t_k)$ with $P$ a $k$-ary predicate symbol in $\mathcal{P}_k$, then 
      \begin{align*}
        \mathrm{eval}_{I'',r}(P(t_1,\ldots,t_k))&=(\mathrm{eval}_{\mathrm{I}'',r}(t_1),\ldots,\mathrm{eval}_{\mathrm{I}'',r}(t_k))\in \mathfrak{F}''(P)\\
        &=\mathrm{eval}_{I'}(P_{(t_1,\ldots,t_k)})
      \end{align*}
      \item Let us consider that $\phi=o(\phi_1,\ldots,\phi_k)$. Then:
        \begin{align*}
          \mathrm{eval}_{I'',r}(o(\phi_1,\ldots,\phi_k))& =o'(\mathrm{eval}_{I'',r}(\phi_1),\ldots,\mathrm{eval}_{I'',r}(\phi_k))\\
          & =o'(\mathrm{eval}_{I'}(T(\phi_1)),\ldots,\mathrm{eval}_{I'}(T(\phi_k))) & \text{(induction hypothesis)}\\
          & =\mathrm{eval}_{I'}(o(T(\phi_1),\ldots,T(\phi_k)))\\
          & =\mathrm{eval}_{I'}(T(\phi)).
        \end{align*}
    \end{enumerate} 
    \end{enumerate}
    \qed
  \end{proof}
  
  \begin{corollary}\label{cor tphi sat phi aussi}
    Let $\mathcal{E}=(\Sigma,\Gamma,\mathcal{P},\mathcal{F})$ be an expression environment. Let $\phi$ be a boolean formula in $\mathcal{P}(\mathcal{F}(\Gamma))$ such that there exists an injection of $\mathrm{Term}(\phi)$ in $\Sigma^*$. Then:
        \begin{align*}
      \text{$T(\phi)$ is satisfiable if and only if $\phi$ is.}
        \end{align*}
  \end{corollary}
  
  \begin{example}\label{ex inj form}
    Let $\mathcal{E}=(\Sigma,\Gamma,\mathcal{P},\mathcal{F})$ be the expression environment defined by:
    \begin{itemize}
      \item $\Sigma=\{a,b,c\}$,
      \item $\Gamma=\{x,y,z\}$,
      \item $\mathcal{P}=\mathcal{P}_2=\{\lessdot,\sim\}$,
      \item $\mathcal{F}_0=\Sigma\cup\{\varepsilon\}$, $\mathcal{F}_2=\{\mathrm{g},\cdot\}$.
    \end{itemize}
    Let us consider the two boolean formulae defined by:
        \begin{align*}
          \phi_1&=\lessdot(\mathrm{g}(ab,x),abx)\wedge \neg (\sim(abx,\mathrm{g}(a,bx))),\\
          \phi_2&=\lessdot(\cdot(ab,x),abx)\wedge \neg (\lessdot(abx,\cdot(a,bx))).
        \end{align*}
        The terms that appear in these two formulae are:
        \begin{align*}
          \mathrm{Term}(\phi_1)&=\{\mathrm{g}(ab,x),abc,\mathrm{g}(a,bx)\},\\
          \mathrm{Term}(\phi_2)&=\{\cdot(ab,c),abc,\cdot(a,bc)\},
        \end{align*}
        Let $\mathrm{I}=(\Sigma^*,\mathfrak{F})$ be the expression interpretation defined by:
    \begin{itemize}
      \item $\mathfrak{F}(\lessdot)=\mathfrak{F}(\sim)=\{(u,v)\mid |u|\leq |v| \}$,
      \item $\mathfrak{F}(\alpha)=\{\alpha\}$, for any $\alpha$ in $\mathcal{F}_0$, 
      \item $\mathfrak{F}(\mathrm{g})=\{(u,v,vu)\}$
      \item $\mathfrak{F}(\cdot)=\{(u,v,u\cdot v)\}$.
    \end{itemize} 
    Finally, let us consider a realization $\mathrm{r}$ that associates $c$ with $x$.  Then
        \begin{align*}
        \mathrm{eval}(I,\mathrm{r})(\mathrm{g}(ab,x))&=cab,\\
        \mathrm{eval}(I,\mathrm{r})(abx)&=abc,\\
        \mathrm{eval}(I,\mathrm{r})(\mathrm{g}(a,bx))&=bca,\\
        \mathrm{eval}(I,\mathrm{r})(\cdot(ab,x))&=\mathrm{eval}(I,\mathrm{r})(abx)\\
        &=\mathrm{eval}(I,\mathrm{r})(\cdot(a,bx)=abc,
        \end{align*}
        Consequently, $\mathrm{eval}(I,\mathrm{r})$ is an injection of $\mathrm{Term}(\phi_1)$ in $\Sigma^*$ but it is not an injection of $\mathrm{Term}(\phi_2)$ in $\Sigma^*$. Furthermore, it holds that
        \begin{align*}
        \mathrm{eval}(I,\mathrm{r})(\phi_1) & =\mathrm{eval}(I,\mathrm{r})(\lessdot(cab,abc)\wedge \neg (\sim(abc,bca)))\\
      & =\mathrm{eval}(I,\mathrm{r})(1\wedge \neg (1))\\
      & =0\\
      \mathrm{eval}(I,\mathrm{r})(\phi_2) & =\mathrm{eval}(I,\mathrm{r})(\lessdot(abc,abc)\wedge \neg (\lessdot(abc,abc)))\\
      & =0\\
        \end{align*}
        Notice that $\phi_2$ is what can we call an \emph{expression contradiction}, since it is a contradiction whenever the function $\cdot$ is interpreted as the catenation function, because of its associativity property. Consequently, for any expression interpretation $\mathrm{I}'$ and any realization $\mathrm{r}'$, $\mathrm{eval}_{\mathrm{I}',\mathrm{r}'}(\phi_2)=0$.
  
  \noindent It is not the case for $\phi_1$, since there exists an injection of its terms in $\Sigma^*$. Let us show that $\phi_1$ is satisfiable.
  
  \noindent First, we need to compute the formula $T(\phi_1)=\lessdot_{\mathrm{g}(ab,x),abx}$ $\wedge \neg (\sim_{abx,\mathrm{g}(a,bx)})$ associated with $\phi_1$. It contains two predicate symbols, $\lessdot_{\mathrm{g}(ab,x),abx}$ and $\sim_{abx,\mathrm{g}(a,bx)}$. Then, let us consider an interpretation $\mathrm{I}'=(\Sigma^*,\mathfrak{F}')$ such that $\mathrm{eval}_{\mathrm{I}'}(\lessdot_{\mathrm{g}(ab,x),abx})=1$ and $\mathrm{eval}_{\mathrm{I}'}(\sim_{abx,\mathrm{g}(a,bx)})=0$. Consequently $\mathrm{eval}_{\mathrm{I}'}(T(\phi_1))=1$. From this interpretation, we can construct the interpretation $\mathrm{I}''=(\Sigma^*,\mathfrak{F}'')$ defined by:
    \begin{itemize}
      \item $\mathfrak{F}''(\alpha)=\{\alpha\}$, for any $\alpha$ in $\mathcal{F}_0$, 
      \item $\mathfrak{F}''(\mathrm{g})=\{(u,v,vu)\}$,
      \item $\mathfrak{F}''(\cdot)=\{(u,v,u\cdot v)\}$,      
      \item $\mathfrak{F}''(\lessdot)=\{(cab,abc)\}$,
      \item $\mathfrak{F}''(\sim)=\emptyset$.
    \end{itemize} 
    Then:
        \begin{align*}
        \mathrm{eval}_{\mathrm{I}'',\mathrm{r}}(\phi_1) & =\mathrm{eval}_{\mathrm{I}'',\mathrm{r}}(\lessdot(\mathrm{g}(ab,x),abx)\wedge \neg (\sim(abx,\mathrm{g}(a,bx))))\\
        & =\mathrm{eval}_{\mathrm{I}'',\mathrm{r}}(\lessdot(cab,abc)\wedge \neg (\sim(abc,cab)))\\
        & =\mathrm{eval}_{\mathrm{I}'',\mathrm{r}}(1\wedge \neg (0))\\
        & =1\\
        \end{align*}
        The existence of the injection allowed us to show that $\phi_1$ was satisfiable \emph{via} the satisfiability of its propositionalised form. Notice that $T(\phi_2)=\lessdot_{\cdot(ab,x),abx}\wedge \neg (\lessdot_{abx,\cdot(a,bx)})$ is satisfiable too, since for any interpretation $\mathrm{I}$ satisfying $\mathrm{eval}_{\mathrm{I}}(\lessdot_{\cdot(ab,x),abx})=1$ and $\mathrm{eval}_{\mathrm{I}}(\lessdot_{abx,\cdot(a,bx)})=0$, $\mathrm{eval}_{\mathrm{I}}(T(\phi_2))=1$. However, since there is no injection due to the associativity of $\cdot$ in any expression interpretation, the satisfiability of $T(\phi_2)$ does not allow us to conclude about the satisfiability of the formula $\phi_2$ (see the notion of normalization in the next subsection).  
  \qed
  \end{example}

  \subsection{Injections for non-Unary Alphabets \emph{via} the Normalization}\label{ssec normalization alph non unaire}
  
  In this subsection, we show that any formula can be transformed into an equivalent one where the set of terms can be evaluated by an injection. In fact, we compute a normal form that takes into account the associativity of the catenation and the identity element $\varepsilon$. Notice that we do not consider unary alphabets where the catenation is also commutative.
  
  \begin{definition}[Normalized Term]
    Let $\mathcal{E}=(\Sigma,\Gamma,P,F)$ be an expression environment. Let $t$ be a term in $F(\Gamma)$. The term $t$ is said to be \emph{normalized} if the  following two conditions are satisfied:
    \begin{itemize}
      \item any child of a concatenation node is not equal to $\varepsilon$;
      \item the root of the left child of any concatenation node in $t$ is not a concatenation node. 
    \end{itemize}
  \end{definition}   
  
  \begin{definition}[Normalization]
    Let $\mathcal{E}=(\Sigma,\Gamma,P,F)$ be an expression environment. The \emph{normalization} of a term $t$ in $F(\Gamma)$ is the transformation $'$ inductively defined as follows:
        \begin{align*}
          x'&=x\\
          (f(t_1,\ldots,t_k))'&=f(t'_1,\ldots,t'_k)\\
          (t_1\cdot t_2)'&=
        \begin{cases}
            (t_2)' & \text{ if } t_1=\varepsilon,\\
            (t_1)' & \text{ if } t_2=\varepsilon,\\
            (t_1)'\cdot(t_2)' & \text{ if } (t_1=f(r_1,\ldots,r_k)\ \vee t_1=x)\ \wedge t_2\neq\varepsilon,\\
            (t_3\cdot(t_4\cdot t_2))' & \text{ if } t_1=(t_3\cdot t_4)\ \wedge t_2\neq\varepsilon,\\
          \end{cases}
        \end{align*}
        where $x$ is any symbol in $\Gamma$, $f$ is any symbol in $F_k\setminus\{\cdot\}$ and $t_1$, $\ldots$, $t_k$ are any $k$ terms in $F(\Gamma)$.   
  \end{definition}
  
  \begin{definition}[Left-Dot Level]
    Let $\mathcal{E}=(\Sigma,\Gamma,P,F)$ be an expression environment. Let $t$ be a term in $F(\Gamma)$. The \emph{left-dot level} $\mathrm{ldl}(t)$ is the integer inductively computed as follows: 
        \begin{align*}
        \mathrm{ldl}(t)&=
        \begin{cases}
            0 & \text{ if }t=x\in\Gamma,\\
            0 & \text{ if } t=f(t_1,\ldots,t_k)\ \wedge f\in F_k\setminus\{\cdot\},\\
            1+\mathrm{ldl}(t_1) & \text{ if } t=t_1\cdot t_2,\\
          \end{cases}
        \end{align*}
  \end{definition}

  \begin{proposition}\label{prop t prim norm}
    Let $\mathcal{E}=(\Sigma,\Gamma,P,F)$ be an expression environment. Let $t$ be a term in $F(\Gamma)$. Then:
        \begin{align*}
        \text{$t'$ is a normalized term.} 
        \end{align*}   
    
    Furthermore, whenever $t$ is a normalized term, then $t=t'$.
  \end{proposition}
  \begin{proof}
    By induction over the structure of $t'$.
    \begin{enumerate}
      \item If $t=x\in\Gamma$, then $t$ is normalized, $x=x'$ and then $t=t'$.
    
      \item If $t=f(t_1,\ldots,t_k)$ with $f$ any symbol in $F_k\setminus\{\cdot\}$, by induction hypothesis it holds that for any integer $j$ in $\{1,\ldots,k\}$, $t'_j$ is normalized and if $t_j$ is normalized, then $t_j=t'_j$. As a direct consequence, $t'$ is normalized and if $t$ is normalized, since it implies that for any integer $j$ in $\{1,\ldots,k\}$, $t_j$ is normalized, then $t=t'$.
    
      \item Suppose that $t=t_1\cdot t_2$. 
      \begin{enumerate}
        \item If $t_1=\varepsilon$ (resp. $t_2=\varepsilon$), then $t'=t'_2$ (resp. $t'=t'_1$). By induction hypothesis it holds that $t'_2$ (resp. $t'_1$) is normalized. As a consequence, $t'$ is normalized. Notice that in this case, $t$ is not normalized.
        \item Suppose that $t_1=x$ with $x\in\Gamma$. Hence, $t'=x\cdot t'_2$. According to induction hypothesis, $t'_2$ is normalized and if $t_2$ is normalized, then $t_2=t'_2$. Since $x'=x$, then $t'=x\cdot t'_2$ is normalized and if $t=x\cdot t_2$ is normalized, then $t'=t$.
        \item Suppose that $t_1=f(r_1,\ldots,r_k)$ with $f$ any symbol in $F_k$ and that $t_2\neq\varepsilon$. By recurrence over $\mathrm{ldl}(t_1)$. 
        \begin{enumerate}
          \item If $\mathrm{ldl}(t_1)=0$, then $t_1=f(r_1,\ldots,r_k)$ with $f\neq\{\cdot\}$. Hence $t'=(t_1)'\cdot(t_2)'$. According to induction hypothesis, for any integer $j$ in $\{1,2\}$, $t'_j$ is normalized and if $t_j$ is normalized, then $t_j=t'_j$. Since $t'_1=f(r'_1,\ldots,r'_k)$, then $t'$ is normalized (since the left child of its concatenation root is not a concatenation node). Furthermore, if $t$ is normalized, since it implies that both $t_1$ and $t_2$ are normalized and that $t'_1=t_1$ and $t'_2=t_2$, it holds that $t=t'$. 
          \item Suppose that $\mathrm{ldl}(t_1)=m$ with $m>0$. Then $t_1=(t_3\cdot t_4)$. As a consequence, $t'=(t_3\cdot(t_4\cdot t_2))'$. Let us notice that $\mathrm{ldl}(t')=\mathrm{ldl}(t)-1$. According to recurrence hypothesis, $(t_3\cdot(t_4\cdot t_2))'$ is normalized. Notice that in this case, $t$ is not normalized.
        \end{enumerate}
      \end{enumerate}
    \end{enumerate}
    \qed
  \end{proof}
  
  Let us show now that the normalization preserves the evaluation.
  
  \begin{proposition}\label{prop t tprim mem eval}
    Let $\mathcal{E}=(\Sigma,\Gamma,P,F)$ be an expression environment. Let $t$ be a term in $F(\Gamma)$. Let $I$ be an interpretation in $\mathrm{Int}(\mathcal{E})$ and $\mathrm{r}$ be a realization in $\mathrm{Real}_\Gamma(I)$. Then:
        \begin{align*}
      \mathrm{eval}_{(I,\mathrm{r})}(t)=\mathrm{eval}_{(I,\mathrm{r})}(t').
        \end{align*}
  \end{proposition}
  \begin{proof}
    By induction over the structure of $t'$.
    \begin{enumerate}
      \item If $t=x\in\Gamma$, then $t'=x=t$. Hence $\mathrm{eval}_{(I,\mathrm{r})}(t)=\mathrm{eval}_{(I,\mathrm{r})}(t')$.    
      \item If $t=f(t_1,\ldots,t_k)$ with $f$ any symbol in $F_k\setminus\{\cdot\}$, then $t'=f(t'_1,\ldots,t'_k)$. By induction hypothesis, it holds that for any integer $j$ in $\{1,\ldots,k\}$, $\mathrm{eval}_{(I,\mathrm{r})}(t_j)=\mathrm{eval}_{(I,\mathrm{r})}(t'_j)$. Hence:
        \begin{align*}
          \mathrm{eval}_{(I,\mathrm{r})}(t)& =x_{k+1} & \text{ with }(\mathrm{eval}_{(I,\mathrm{r})}(t_1),\ldots,\mathrm{eval}_{(I,\mathrm{r})}(t_k),x_{k+1})\in\mathfrak{F}(f)\\
          \mathrm{eval}_{(I,\mathrm{r})}(t')&=x'_{k+1} & \text{ with }(\mathrm{eval}_{(I,\mathrm{r})}(t'_1),\ldots,\mathrm{eval}_{(I,\mathrm{r})}(t'_k),x_{k+1})\in\mathfrak{F}(f)
        \end{align*}    
    Finally, using Definition~\ref{def interpretation}, it holds that $x_{k+1}=x'_{k+1}$.
    
    \item Suppose that $t=t_1\cdot t_2$.     
    \begin{enumerate}
      \item If $t_1=\varepsilon$ (resp. $t_2=\varepsilon$), then $t'=t'_2$ (resp. $t'=t'_1$). By induction hypothesis, $\mathrm{eval}_{(I,\mathrm{r})}(t_2)=\mathrm{eval}_{(I,\mathrm{r})}(t'_2)$ (resp. $\mathrm{eval}_{(I,\mathrm{r})}(t_1)=\mathrm{eval}_{(I,\mathrm{r})}(t'_1)$). Then:
        \begin{align*}
          \mathrm{eval}_{(I,\mathrm{r})}(t) &=\mathrm{eval}_{(I,\mathrm{r})}(t_1)\cdot \mathrm{eval}_{(I,\mathrm{r})}(t_2)\\
          & =\varepsilon \cdot \mathrm{eval}_{(I,\mathrm{r})}(t'_2)\\
          \text{ (resp.} & =\mathrm{eval}_{(I,\mathrm{r})}(t'_1) \cdot \varepsilon)\\
          &= \mathrm{eval}_{(I,\mathrm{r})}(t')
        \end{align*}     
    \item  Suppose that $t_1=x$ with $x$ in $\Gamma$. Then $t'=x\cdot t'_2$. By induction hypothesis, $\mathrm{eval}_{(I,\mathrm{r})}(t_2)=\mathrm{eval}_{(I,\mathrm{r})}(t'_2)$. Then:
        \begin{align*}
          \mathrm{eval}_{(I,\mathrm{r})}(t)&=\mathrm{eval}_{(I,\mathrm{r})}(x)\cdot \mathrm{eval}_{(I,\mathrm{r})}(t_2)\\
          &=\mathrm{r}(x) \cdot \mathrm{eval}_{(I,\mathrm{r})}(t'_2)\\
          & =\mathrm{eval}_{(I,\mathrm{r})}(x'\cdot t'_2)\\
          &=\mathrm{eval}_{(I,\mathrm{r})}(t')
        \end{align*}    
    \item  Suppose that $t_1=f(r_1,\ldots,r_k)$ with $f$ any symbol in $F_k$ and that $t_2\neq\varepsilon$. By recurrence over $\mathrm{ldl}(t_1)$. 
    \begin{enumerate}
      \item  If $\mathrm{ldl}(t_1)=0$, then $t_1=f(r_1,\ldots,r_k)$ with $f\neq\{\cdot\}$. Hence $t'=(t_1)'\cdot(t_2)'$.
     According to induction hypothesis, for any integer $j$ in $\{1,2\}$, $\mathrm{eval}_{(I,\mathrm{r})}(t_j)=\mathrm{eval}_{(I,\mathrm{r})}(t'_j)$.
     As a consequence, 
        \begin{align*}
          \mathrm{eval}_{(I,\mathrm{r})}(t) &=\mathrm{eval}_{(I,\mathrm{r})}(t_1)\cdot \mathrm{eval}_{(I,\mathrm{r})}(t_2)\\
          & =\mathrm{eval}_{(I,\mathrm{r})}(t'_1) \cdot \mathrm{eval}_{(I,\mathrm{r})}(t'_2)\\
          & = \mathrm{eval}_{(I,\mathrm{r})}(t')
        \end{align*}      
    \item Suppose that $\mathrm{ldl}(t_1)=m$ with $m>0$. Then $t_1=(t_3\cdot t_4)$. As a consequence, $t'=(t_3\cdot(t_4\cdot t_2))'$. Notice that $\mathrm{ldl}(t')=\mathrm{ldl}(t)-1$. According to recurrence hypothesis,
    $\mathrm{eval}_{(I,\mathrm{r})}(t_3\cdot(t_4\cdot t_2))=\mathrm{eval}_{(I,\mathrm{r})}(t_1)\cdot \mathrm{eval}_{(I,\mathrm{r})}((t_3\cdot(t_4\cdot t_2))')$. Hence
    $\mathrm{eval}_{(I,\mathrm{r})}(t)=\mathrm{eval}_{(I,\mathrm{r})}(t_1)\cdot \mathrm{eval}_{(I,\mathrm{r})}(t')$.
    \end{enumerate}
    \end{enumerate}
    \end{enumerate}    
    \qed
  \end{proof}
  
  Let $\mathcal{E}=(\Sigma,\Gamma,P,F)$ be an expression environment. We denote by $F(\Gamma)'$ the set of normalized terms in $F(\Gamma)$.   
  
  \begin{definition}[Normalized Formula]
    Let $\mathcal{E}=(\Sigma,\Gamma,P,F)$ be an expression environment. Let $\phi$ be a boolean formula in $P(F(\Gamma))$. The formula $\phi$ is said to be \emph{normalized} if any of the terms appearing in it are normalized (\emph{i.e.} if it belongs to $P(F(\Gamma)')$).
  \end{definition}

  \begin{definition}[Formula Normalization]
    Let $\mathcal{E}=(\Sigma,\Gamma,P,F)$ be an expression environment. The \emph{normalization} of a boolean formula $\phi$ in $P(F(\Gamma))$ is the transformation $'$ inductively defined as follows:
        \begin{align*}
        (o(\phi_1,\ldots,\phi_k))'&=o(\phi'_1,\ldots,\phi'_k)\\
        (P(t_1,\ldots,t_k))'&=P(t'_1,\ldots,t'_k)
        \end{align*}
        where $k$ is any integer, $P$ is any predicate symbol in $\mathcal{P}_k$, $t_1,\ldots,t_k$ are any $k$ elements in $\mathcal{F}(X)$, $\mathrm{o}$ is any $k$-ary boolean operator associated with a mapping $\mathrm{o}'$ from $\{0,1\}^k$ to $\{0,1\}$ and $\phi_1,\ldots,\phi_k$ are any $k$ boolean formulae over $(\mathcal{P},X)$.
  \end{definition}
  
  \begin{proposition}
    Let $\mathcal{E}=(\Sigma,\Gamma,P,F)$ be an expression environment. Let $\phi$ be a boolean formula in $P(F(\Gamma))$. Then:
        \begin{align*}
          \text{$\phi'$ is a normalized formula.}   
        \end{align*} 
  \end{proposition}
  \begin{proof}
    By induction over the structure of $\phi$, this is a direct corollary of Proposition~\ref{prop t prim norm} as an inductive extension of the normalization.
    \qed
  \end{proof}
  
  \begin{proposition}\label{prop normalization preserv eval}
    Let $\mathcal{E}=(\Sigma,\Gamma,P,F)$ be an expression environment. Let $\phi$ be a boolean formula in $P(F(\Gamma))$. Let $I$ be an interpretation in $\mathrm{Int}(\mathcal{E})$ and $\mathrm{r}$ be a realization in $\mathrm{Real}_\Gamma(I)$. Then:
        \begin{align*}
          \mathrm{eval}_{(I,\mathrm{r})}(\phi)=\mathrm{eval}_{(I,\mathrm{r})}(\phi')
        \end{align*}
  \end{proposition}
  \begin{proof}
    By induction over the structure of $\phi$, this is a direct corollary of Proposition~\ref{prop t tprim mem eval} as an inductive extension of the normalization.
    \qed
  \end{proof}
  
  \begin{example}
    Let us consider the formula $\phi_2=\lessdot(\cdot(ab,x),abx)\wedge \neg (\lessdot(abx,\cdot(a,bx)))$ of Example~\ref{ex inj form}. Considering the catenation as right-associative, its normalized form is the formula $\phi_2'=\lessdot(abx,abx)\wedge \neg (\lessdot(abx,abx))$, that is, a classical contradiction.
    \qed
  \end{example}
  
  Let us now show how to compute an injection from a set of normalized terms.
  
  \begin{definition}[Left, Right and Middle Word]
        A word is a \emph{left word} (resp. \emph{right word}, \emph{middle word}) of a term $t$ in $F(\Gamma)$ if it belongs to the set $\mathrm{LeftWord}(t)$ (resp. $\mathrm{RightWord}(t)$, $\mathrm{MiddleWords}(t)$) computed as follows:
        \begin{align*}
          \mathrm{LeftWord}(x)& =\{x\mid x\in\Sigma\cup\{\varepsilon\}\}\\
          \mathrm{LeftWord}(f(t_1,\ldots,t_k))&=\emptyset\\
          \mathrm{LeftWord}(\cdot(t_1,t_2))&=
        \begin{cases}
            \mathrm{LeftWord}(t_1) & \text{ if } t_1\notin (\{\cdot,\varepsilon\}\cup\Sigma)(\emptyset)\\
            & \quad \vee \mathrm{LeftWord}(t_2)=\emptyset,\\
            \mathrm{LeftWord}(t_1)\cdot \mathrm{LeftWord}(t_2) & \text{ otherwise,}
          \end{cases}\\
          \mathrm{RightWord}(x)&=\{x\mid x\in\Sigma\cup\{\varepsilon\}\}\\
          \mathrm{RightWord}(f(t_1,\ldots,t_k))&=\emptyset\\
          \mathrm{RightWord}(\cdot(t_1,t_2))&=
        \begin{cases}
            \mathrm{RightWord}(t_2) & \text{ if } t_2\notin (\{\cdot,\varepsilon\}\cup\Sigma)(\emptyset)\\
            &\quad \vee \mathrm{RightWord}(t_1)=\emptyset,\\
            \mathrm{RightWord}(t_1)\cdot \mathrm{RightWord}(t_2) & \text{ otherwise,}
          \end{cases}\\
          \mathrm{MiddleWords}(x)&=\{x\mid x\in\Sigma\cup\{\varepsilon\}\}\\
          \mathrm{MiddleWords}(f(t_1,\ldots,t_k))&=\bigcup_{j\in\{1,\ldots,k\}} \mathrm{MiddleWords}(t_j)\\
          \mathrm{MiddleWords}(\cdot(t_1,t_2))&=
        \begin{cases}
            \mathrm{RightWord}(t_1)\cdot \mathrm{LeftWord}(t_2) \\
            \quad\text{ if } t_1,t_2\in (\{\cdot,\varepsilon\}\cup\Sigma)(\emptyset)\\
            \mathrm{RightWord}(t_1)\cdot \mathrm{LeftWord}(t_2) \cup \mathrm{MiddleWords}(t_1)\\
            \quad\text{ if } t_1\notin (\{\cdot,\varepsilon\}\cup\Sigma)(\emptyset)\ \wedge\ t_2\in (\{\cdot,\varepsilon\}\cup\Sigma)(\emptyset)\\
            \mathrm{RightWord}(t_1)\cdot \mathrm{LeftWord}(t_2) \cup \mathrm{MiddleWords}(t_2)\\
            \quad\text{ if } t_1\in (\{\cdot,\varepsilon\}\cup\Sigma)(\emptyset)\ \wedge\ t_2\notin (\{\cdot,\varepsilon\}\cup\Sigma)(\emptyset)\\
            \mathrm{RightWord}(t_1)\cdot \mathrm{LeftWord}(t_2) \cup \mathrm{MiddleWords}(t_1)\cup \mathrm{MiddleWords}(t_2)\\
            \quad\text{ if } t_1,t_2\notin (\{\cdot,\varepsilon\}\cup\Sigma)(\emptyset)\\
          \end{cases}
        \end{align*}
        where $x$ is an element in $F_0\cup \Gamma$, $f$ is a function symbol in $F_k\setminus\{\cdot\}$ and $t_1,\ldots,t_k$ are any $k$ terms in $F(\Gamma)$.
    
    A word $u$ is a \emph{factor} of the term $t$ if it is a factor of $v$ where $v\in\mathrm{MiddleWords}(t)$. 
  \end{definition}
  
  \begin{example}\label{ex factor term}
    Let us illustrate the notion of factor:
      \begin{itemize}
         \item The factors of $f(a,g(a,baxc))$ are $\{a,b,ba,c\}$.    
         \item The factors of $\cdot(a,g(a,baxc))$ are $\{a,b,ba,c\}$.
         \item The factors of $f(a,\cdot(a,baxc))$ are $\{a,ab,aba,b,ba,c\}$. 
     \end{itemize}
    \qed
  \end{example}

  \begin{definition}[$\mathrm{root}$ Function]
     Let $\mathrm{root}$ be the function from $F(\Gamma)$ to $F\cup\Gamma$ defined for any term $t$ as follows: 
        \begin{align*}
        \mathrm{root}(t)=
        \begin{cases}
            t & \text{ if } t\in \Gamma\cup F_0,\\
            f & \text{ if }t=f(t_1,\ldots,t_k)\text{ with }f\in F_k.\\
          \end{cases}
        \end{align*}
  \end{definition}
  
  \begin{lemma}\label{lem exist mot separ}
    Let $\mathcal{E}=(\Sigma,\Gamma,P,F)$ be an expression environment such that $\mathrm{Card}(\Sigma)\geq 2$. Let $T$ be a finite subset of $F(\Gamma)'$. Then there exists a word $w$ in $\Sigma^*$ such that for any term $t$ in $F(\Gamma)'$, for any two distinct terms $t_1$ and $t_2$ in $F(\Gamma)'$, it holds:
        \begin{align*}
          ({t_1}_{t\leftarrow \mathrm{Term}(w)})'\neq({t_2}_{t\leftarrow \mathrm{Term}(w)})'
        \end{align*}
  \end{lemma}
  \begin{proof} 
  Let $w=ab^pa$ be such that $p$ is the smallest integer such that any factor $b^q$ of a term of $T$ satisfies $q<p$.
  Let $s_1=({t_1}_{t\leftarrow \mathrm{Term}(w)})$ and $s_2=({t_1}_{t\leftarrow \mathrm{Term}(w)})$.     
  
  If $t$ is neither a subterm of $t_1$ nor of $t_2$, then $t_1=s_1=s'_1$, $t_2=s_2=s'_2$ and thus $s'_1\neq s'_2$.
    
  Let $t$ be a subterm of $t_1$ but not of $t_2$. There exists a factor $ab^{p}a$ of $s'_1$ and any factor $b^q$ of $s'_2$ (that is, a factor of $t_2$) satisfies $q<p$. Then $s'_1\neq s'_2$.
    
  Suppose that $t$ is a subterm of $t_1$ and of $t_2$. 
  
  \begin{enumerate}
    \item Suppose that $t_1=t$. Then $\mathrm{root}(s'_1)=\cdot$ since $s'_1=\mathrm{Term}(ab^pa)$.
    \begin{enumerate}
    \item If $t_2=y\neq t$, then $\mathrm{root}(s'_2)=s'_2=y$. Hence $s'_1\neq s'_2$.
    
    \item\label{1b} Suppose that $t_2=f({t_2}_1,\ldots,{t_2}_k)$ with $F_k\setminus\{\cdot\}$. Then $\mathrm{root}(s'_2)=f$. Hence $s'_1\neq s'_2$.
    
    \item\label{1c} Suppose that $t_2=\cdot({t_2}_1,{t_2}_2)$. Since $t_2$ is normalized, then $\mathrm{root}(t_2)~\neq~\cdot$.
    \begin{enumerate}
    \item If ${t_2}_1=t$, then $({t_2}_{t\leftarrow w})'=\cdot(a,\cdot(b,(\ldots,b,\cdot(a,({{t_2}_2}_{x\leftarrow w})')\ldots)))$. Since ${t_2}_2\neq\varepsilon$, then ${{t_2}_2}_{t\leftarrow w}\neq\varepsilon$ and then $s'_1\neq s'_2$.
    
    \item If ${t_2}_1\neq t$, then $s'_1=\cdot(({t_2}_1)',({t_2}_2)')$.
    \begin{enumerate}
    \item If ${t_2}_1\neq a$ then $s'_1\neq s'_2$.
    
    \item Suppose that ${t_2}_1= a$.  Either $({t_2}_2)'\neq\mathrm{Term}(b^pa)$ and then $s'_1\neq s'_2$ or $({t_2}_2)'=\mathrm{Term}(b^pa)$, and $t_2=\mathrm{Term}(ab^pa)$  (Contradiction with the definition of $p$).
    \end{enumerate}
    \end{enumerate}
    \end{enumerate}
    
    \item  If $f=({t_1}_1,\ldots,{t_1}_k)$ with $F_k\setminus\{\cdot\}$. Then $({t_1}_{t\setminus w})'=f(({{t_1}_1}_{t\leftarrow w})',\ldots,({{t_1}_k}_{t\leftarrow w})')$ and $\mathrm{root}(s'_1)=f$.
    \begin{enumerate}
    \item  Suppose that $t_2=t$. See case~\ref{1b}.
    
    \item Suppose that $t_2=y\in\Gamma\neq t$. Then $\mathrm{root}(s'_2)=s'_2=y$. Hence $s'_1\neq s'_2$.
    
    \item\label{2c} Suppose that $t_2=g({t_2}_1,\ldots,{t_2}_l)$ with $g\in F_l$ and $g\neq f$. Then $\mathrm{root}(s'_2)=g$ and then $s'_1\neq s'_2$.
    
    \item Suppose that $t_2=f({t_2}_1,\ldots,{t_2}_k)$. Then $({t_2}_{t\leftarrow w})'=f(({{t_2}_1}_{t\leftarrow w})',\ldots,$ $({{t_2}_k}_{t\leftarrow w})')$. Since $t_1\neq t_2$, there exists $j$ in $\{1,\ldots,k\}$ such that ${t_1}_j\neq{t_2}_j$. According to induction hypothesis, $({{t_1}_j}_{t\leftarrow w})'\neq ({{t_2}_j}_{t\leftarrow w})'$, it holds that $s'_1\neq s'_2$.
    \end{enumerate}
    \item  Suppose that $t_1=\cdot({t_1}_1,{t_1}_2)$. Then $\mathrm{root}({t_1}_1)\neq \cdot$.
    \begin{enumerate}
    \item If $t_2=t$, see case case~\ref{1c}.
    
    \item If $t_2=y\in\Gamma\neq t$, then $\mathrm{root}(s'_2)=s'_2=y$. Hence $s'_1\neq s'_2$. 
    
    \item If $t_2=f({t_2}_1,\ldots,{t_2}_k)$ with $f\in F_k\setminus\{\cdot\}$, see case~\ref{2c}.
    
    \item Suppose that $t_2=\cdot({t_2}_1,{t_2}_2)$. Consequently $\mathrm{root}({t_2}_1)\neq\cdot$.
    
    \begin{enumerate}
    \item If ${t_1}_1=t$ then $s'_1=\cdot(a,\cdot(b,\cdot(\ldots,\cdot(b,\cdot(a,({{t_1}_2}_{t\leftarrow w})',))\ldots)))$.
    \begin{enumerate}
    \item If ${t_2}_1=t$, then ${t_1}_2\neq {t_2}_2$. According to induction hypothesis, $({{t_1}_2}_{t\leftarrow w})'$ $\neq ({{t_2}_2}_{t\leftarrow w})'$. Since $s'_2=\cdot(a,\cdot(b,\cdot(\ldots,\cdot(b\cdot(a,({{t_2}_2}_{t\leftarrow w})')))\ldots)))$, $s'_1\neq s'_2$.
    
    \item\label{3d1b} If ${t_2}_1=a$, then $s'_2=\cdot (a,({{t_2}_2}_{t\leftarrow w})')$. Either $({{t_2}_2}_{t\leftarrow w})'$ does not admit $b^p$ as a prefix of a left word and then $s'_1\neq s'_2$ or it does and then $t_2$ admits $b^p$ as a factor (contradiction with the definition of $p$).
    
    \item\label{3d1c} If ${t_2}_1=f({{t_2}_1}_1,\ldots, {{t_2}_1}_k)$, then the respective roots of the leftmost subterm of  $s'_1$ and $s'_2$ are distinct. Hence $s'_1\neq s'_2$.
    
    \end{enumerate}
    \item Suppose that ${t_1}_1=a$.
    
    \begin{enumerate}
    \item If ${t_2}_1=t$, see case~\ref{3d1b}.
    
    \item If ${t_2}_1=a$, then since $t_1\neq t_2$, it holds that ${t_{1}}_2\neq {t_2}_2$. By induction hypothesis, $({{t_1}_2}_{t\leftarrow w})'\neq ({{t_2}_2}_{t\leftarrow w})'$. Finally, since $s'_2=\cdot(a,({{t_2}_2}_{t\leftarrow w})')$, $s'_1\neq s'_2$.
    
    \item If ${t_2}_1=f({{t_2}_1}_1,\ldots, {{t_2}_1}_k)$, then $\mathrm{root}(s'_2)=f\neq \cdot=\mathrm{root}(s'_1)$. Hence $s'_1\neq s'_2$.
    \end{enumerate}
    \item Suppose that ${t_1}_1=f({{t_1}_1}_1,\ldots,{{t_1}_1}_k)$. Then $s'_1=\cdot(u_1,u_2)$ with $\mathrm{root}(u_1)=f$.
    \begin{enumerate}
    \item If ${t_2}_1=t$, see case~\ref{3d1c}.
    
    \item If ${t_2}_1=y\in\Gamma\neq t$ or if ${t_2}_1=g({{t_2}_1}_1,\ldots, {{t_2}_1}_l)$ with $g\in F_l$, the respective roots of the leftmost subterm of  $s'_1$ and $s'_2$ are distinct. Hence $s'_1\neq s'_2$.
    
    \item If ${t_2}_1=f({{t_2}_1}_1,\ldots, {{t_2}_1}_k)$, then $s'_2=\cdot(v_1,v_2)$ with $\mathrm{root}(v_1)=f$. 
    Two cases can occur: either ${t_1}_2\neq {t_2}_2$ or there exists $j$ in $\{1,\ldots,k\}$ such that ${{t_1}_1}_j\neq {{t_2}_1}_j$. 
    In the first (resp. second) case, it holds by induction that $u_2\neq v_2$ (resp. $u_1\neq v_1$). Consequently, $s'_1\neq s'_2$. 
    \end{enumerate}
    \end{enumerate}
    \end{enumerate}
    \end{enumerate}  
    \qed
  \end{proof}
  
  \begin{definition}[Term Index]
        Let us define the index $\mathrm{Ind}(T)$ as the integer computed as follows:
        \begin{align*}
          \mathrm{Ind}(T)=\bigcup_{t\in T} \mathrm{Ind}(T)
        \end{align*}
        where for any term $t$,
        \begin{align*}
        \mathrm{Ind}(t)=
        \begin{cases}
            0 & \text{ if } t\in \Sigma \cup\{\varepsilon\},\\
            1 & \text{ if } t\in \Gamma\cup F_0\setminus(\Sigma\cup\{\varepsilon\}),\\
            1 + \mathrm{Ind}(t_1)+\cdots +\mathrm{Ind}(t_k) & \text{ if } t=f(t_1,\ldots,t_k)\wedge f\in F_k\setminus\{\cdot\},\\
            \mathrm{Ind}(t_1)+\mathrm{Ind}(t_2) & \text{ if } t=t_1\cdot t_2.\\
          \end{cases}
        \end{align*}
    Let us define for any two terms $t_1$ and $t_2$ the set $T'_{t_1\leftarrow t_2}$ computed as follows:
        \begin{align*}
          T'_{t_1\leftarrow t_2}=\bigcup_{t\in T} \{t'_{t_1\leftarrow t_2}\}
        \end{align*}
        where for any term $t$,
        \begin{align*}
        t_{t_1\leftarrow t_2}=
        \begin{cases}
            t_2 & \text{ if }t=t_1,\\
            f({u_1}_{t_1\leftarrow t_2},\ldots,{u_k}_{t_1\leftarrow t_2}) & \text{ if } t=f(u_1,\ldots,u_k)\wedge f\in F_k,\\
            t & \text{ otherwise.}
          \end{cases}
        \end{align*}
    Let us define for any term $t$ the depth $D(t)$ inductively computed as follows:
        \begin{align*}
        d(t)=
        \begin{cases}
            1 & \text{ if }t \in F_0\cup\Gamma,\\
            1+\mathrm{max}(d(t_1,\ldots,t_k)) & \text{ if }t=f(t_1,\ldots,t_k)\ \wedge\ f\in F_k.\\
          \end{cases}
        \end{align*}
  \end{definition}
  
    \begin{proposition}\label{prop exist inj cas binaire}
    Let $\mathcal{E}=(\Sigma,\Gamma,P,F)$ be an expression environment with $\mathrm{Card}(\Sigma)\geq 2$. Let $T$ be a finite subset of $F(\Gamma)'$.
    There exist $I$ an interpretation in $\mathrm{Int}(\mathcal{E})$ and $\mathrm{r}$ a realization in $\mathrm{Real}_\Gamma(I)$ such that the function $\mathrm{eval}_{I,r}$ is an injection of $T$ in $\Sigma^*$.
  \end{proposition}
  \begin{proof}
    By recurrence over $\mathrm{Ind}(T)$.
    
    \begin{enumerate}
      \item  If $\mathrm{Ind}(T)=0$ then $T\subset(\Sigma\cup\{\cdot,\varepsilon\})(\emptyset)$. Let us show that for any two distinct terms $t_1$ and $t_2$ in $T$, for any interpretation $I$ and for any realization $r$, it holds that $\mathrm{eval}_{I,r}(t_1) \neq \mathrm{eval}_{I,r}(t_2)$.     
    By recurrence over $d(t_1)$. Let $I$ be any interpretation and $r$ be any realization.
    
    \begin{enumerate}
      \item\label{it 1a prop exist inj} Suppose that $d(t_1)=1$. Then $t_1\in\Sigma\cup\{\varepsilon\}$. Consequently, $|\mathrm{eval}_{I,r}(t_1)|\leq 1$.
    
      \begin{enumerate}
        \item If $d(t_2)=1$, since $t_1\neq t_2$, then $t_2\in\Sigma\cup\{\varepsilon\}\setminus\{t_1\}$. Hence, $\mathrm{eval}_{I,r}(t_1)\neq \mathrm{eval}_{I,r}(t_2)$.
    
        \item If $d(t_2)\neq 1$, then $t_2=\cdot(a,s_2)$ with $a\in\Sigma$ and $s_2\neq\varepsilon$. Hence $|\mathrm{eval}_{I,r}(t_2)|>1$. As a consequence,  $|\mathrm{eval}_{I,r}(t_1)|\neq |\mathrm{eval}_{I,r}(t_2)|$ and then $\mathrm{eval}_{I,r}(t_1)\neq \mathrm{eval}_{I,r}(t_2)$.
      \end{enumerate}
    
    \item  Suppose that $d(t_1)>1$. Then $t_1=\cdot(a,s_1)$ with $a\in\Sigma$.
    
      \begin{enumerate}
        \item If $d(t_2)=1$, then it is symmetrically equivalent to item~\ref{it 1a prop exist inj}.
    
        \item Suppose that $d(t_2)\neq 1$. Then $t_2=\cdot(b,s_2)$ with $b\in\Sigma$. 
    If $a\neq b$, then $a\cdot \mathrm{eval}_{I,r}(s_1)\neq b\cdot \mathrm{eval}_{I,r}(s_2)$ and then $\mathrm{eval}_{I,r}(t_1)\neq \mathrm{eval}_{I,r}(t_2)$.
    Otherwise, it holds $s_1\neq s_2$. According to recurrence hypothesis, $\mathrm{eval}_{I,r}(s_1)\neq\mathrm{eval}_{I,r}(s_2)$ and consequently $a\cdot \mathrm{eval}_{I,r}(s_1)\neq a\cdot \mathrm{eval}_{I,r}(s_2)$. Consequently, $\mathrm{eval}_{I,r}(t_1)\neq \mathrm{eval}_{I,r}(t_2)$.
      \end{enumerate}
    \end{enumerate}
    
    \item Suppose that $\mathrm{Ind}(T)\neq 0$.
    
    \begin{enumerate}
      \item Let $x$ be a symbol in $\Gamma$ such that $x$ is a subterm of a term in $T$. 
    
      \begin{enumerate}
        \item\label{it 2ai prop exist inj} According to Lemma~\ref{lem exist mot separ}, there exists $w$ in $\Sigma^*$ such that for any two terms $t_1$ and $t_2$ in $F(\Gamma)'$, it holds: $({t_1}_{x\leftarrow w})'\neq({t_2}_{x\leftarrow w})'$.

        \item It holds by recurrence hypothesis that there exists $I$ an interpretation in $\mathrm{Int}(\mathcal{E})$ and $\mathrm{r}$ a realization in $\mathrm{Real}_\Gamma(I)$ such that the function $\mathrm{eval}_{I,r}$ is an injection of $T_{x\leftarrow w}$ in $\Sigma^*$. Let us consider the realization $r'$ defined for any symbol $y$ in $\Gamma$ as follows:
        \begin{align*}
        r'(y)=
        \begin{cases}
            w & \text{ if }x=y,\\
            r(y) & \text{ otherwise}.
          \end{cases}
        \end{align*}
    Let us show that $\mathrm{eval}_{I,r}$ is an injection of $T$ in $\Sigma^*$. 
    
    Let $t_1$ and $t_2$ be two terms in $T$. According to Item~\ref{it 2ai prop exist inj}, $s'_1=({t_1}_{x\leftarrow w})'\neq ({t_2}_{x\leftarrow w})'=s'_2$. By definition of $r'$, $\mathrm{eval}_{I,r'}(s'_1)\neq \mathrm{eval}_{I,r'}(s'_2)$. Since by construction of $r'$, $\mathrm{eval}_{I,r'}(t_1)= \mathrm{eval}_{I,r'}(s'_1)$ and since $\mathrm{eval}_{I,r'}(t_2)= \mathrm{eval}_{I,r'}(s'_2)$, it holds that $\mathrm{eval}_{I,r'}(t_1)\neq \mathrm{eval}_{I,r'}(t_2)$.
    \end{enumerate}
    \item Suppose that there is no subterm of a term in $T$ that belongs to $\Gamma$. Let $t=f(t_1,\ldots,t_k)$ be a subterm in a term in $T$ such that $t_1,\ldots,t_k$ are $k$ terms in $(\Sigma\cup\{\cdot,\varepsilon\})(\emptyset)$.  
    
    \begin{enumerate}
      \item\label{it 2bi prop exist inj} According to Lemma~\ref{lem exist mot separ}, there exists $w$ in $\Sigma^*$ such that for any two terms $t_1$ and $t_2$ in $F(\Gamma)'$, it holds that: $({t_1}_{t\leftarrow w})'\neq({t_2}_{t\leftarrow w})'$. 
        
      \item It holds by recurrence hypothesis that there exists $I=(\Sigma^*,\mathfrak{F})$ an interpretation in $\mathrm{Int}(\mathcal{E})$ and $\mathrm{r}$ a realization in $\mathrm{Real}_\Gamma(I)$ such that the function $\mathrm{eval}_{I,r}$ is an injection of $T_{f(t_1,\ldots,t_k)\leftarrow w}$ in $\Sigma^*$. 
    Let us denote by $w_j$ the word $\mathrm{eval}_{I,r}(t_j)$ for any integer $j$ in $\{1,\ldots,k\}$.
    Let us consider the interpretation $I'=(\Sigma^*,\mathfrak{F}')$ defined as follows:
    \begin{enumerate}
      \item for any predicate symbol $P$ in $\mathcal{P}$, $\mathfrak{F}(P)=\mathfrak{F}'(P)$,
      \item for any function symbol $g$ in $\mathcal{F}\setminus\{f\}$, $\mathfrak{F}(g)=\mathfrak{F}'(g)$,
      \item for any $k$ word $u_1,\ldots,u_k$ in $\Sigma^*$:
        \begin{align*}
        u_1,\ldots,u_k,u_{k+1}\in \mathfrak{F}'(f) \Leftrightarrow u_1,\ldots,u_k,u_{k+1}\in \mathfrak{F}(f)\wedge (u_1,\ldots,u_k)\neq(w_1,\ldots,w_k)
        \end{align*}    
    \item $w_1,\ldots,w_k,w\in \mathfrak{F}'(f)$
  \end{enumerate}
    
    Let us show that $\mathrm{eval}_{I',r}$ is an injection of $T$ in $\Sigma^*$. 
    
        Let $t_1$ and $t_2$ be two terms in $T$. According to Item~\ref{it 2bi prop exist inj}, $s'_1=({t_1}_{t\leftarrow w})'\neq ({t_2}_{t\leftarrow w})'=s'_2$. By definition of $I'$, $\mathrm{eval}_{I',r}(s'_1)\neq \mathrm{eval}_{I',r}(s'_2)$. Since by construction of $I'$, $\mathrm{eval}_{I',r}(t_1)= \mathrm{eval}_{I',r}(s'_1)$ and since $\mathrm{eval}_{I',r}(t_2)= \mathrm{eval}_{I',r}(s'_2)$, it holds that $\mathrm{eval}_{I',r}(t_1)\neq \mathrm{eval}_{I',r}(t_2)$.
        \end{enumerate}
        \end{enumerate}
    \end{enumerate}
    \qed
  \end{proof}

  \begin{proposition}
  There exists an expression environment $\mathcal{E}=(\Sigma,\Gamma,P,F)$ with $\mathrm{Card}(\Sigma)=1$ and a finite subset $T$ of $F(\Gamma)'$ such that for any interpretation $I$ in $\mathrm{Int}(\mathcal{E})$, for any realization $\mathrm{r}$ in $\mathrm{Real}_\Gamma(I)$, $\mathrm{eval}_{I,r}$ is not an injection of $T$ in $\Sigma^*$.
  \end{proposition}
  \begin{proof}
    Let $\Gamma=\{x,y\}$ and $T=\{t_1=x\cdot y,t_2=y\cdot x\}$.
    By definition, both $t_1$ and $t_2$ are normalized.
    Since the catenation product is commutative for unary alphabets, for any interpretation $I$ in $\mathrm{Int}(\mathcal{E})$, for any realization $\mathrm{r}$ in $\mathrm{Real}_\Gamma(I)$, $\mathrm{eval}_{I,r}(t_1)=\mathrm{eval}_{I,r}(t_2)$, and therefore $\mathrm{eval}_{I,r}(t_1)$ is not an injection of $T$ in $\Sigma^*$.
  \qed
  \end{proof}
  
  \begin{example}
    Let us consider the terms $t_1=\cdot(a,g(a,baxc))$ and $v_1=f(a,\cdot(a,baxc))$ and their factors (Example~\ref{ex factor term}).
      \begin{itemize}
         \item The factors of $t_1$ are $F_{v_1}=\{a,b,ba,c\}$.
         \item The factors of $v_1$ are $F_{v_2}=\{a,ab,aba,b,ba,c\}$. 
     \end{itemize}
     Consider the word $w=ab^2a$ which is not in $F_{v_1}\cup F_{v_2}$. Consider a realization $\mathrm{r}$ associating $w_1$ with $x$. Let us substitute $x$ with $w_1$ in $t_1$ and $v_1$:
        \begin{align*}
          t_2&={t_1}_{x\leftarrow w_1}'=\cdot(a,g(a,baab^2ac)),\\
           v_2&={v_1}_{x\leftarrow w_1}'=f(a,abaab^2ac).
        \end{align*}
        The word $w_2=ab^3a$ is neither a factor of $t_2$ nor of $v_2$. Consider an interpretation $I=(\Sigma^*,\mathfrak{F})$ where $(a,abaab^2ac,ab^3a)\in \mathfrak{F}(f)$. Let us substitute $f(a,abaab^2ac)$ with $w_2$ in $t_2$ and $v_2$:
        \begin{align*}
          t_3&={t_2}_{f(a,abaab^2ac) \leftarrow w_2}'\\
          &=\cdot(a,g(a,baab^2ac))\\
          v_3&={v_2}_{f(a,abaab^2ac)\leftarrow w_2}'\\
          &=ab^3a
        \end{align*}
        The word $w_3=ab^4a$ is neither a factor of $t_3$ nor of $v_3$. Consider that the interpretation $I=(\Sigma^*,\mathfrak{F})$ satisfies $(a,baab^2ac,ab^4a)\in \mathfrak{F}(g)$. Let us substitute $g(a,baab^2ac)$ with $w_3$ in $t_3$ and $v_3$:
        \begin{align*}
          t_4&={t_3}_{g(a,baab^2ac) \leftarrow w_3}'\\
          &=aab^4a\\
          v_4&={v_3}_{g(a,baab^2ac)\leftarrow w_3}'\\
          &=ab^3a
        \end{align*}
        Hence, since $\mathrm{eval}_{I,\mathrm{r}}(t_1)=aab^4a$ and $\mathrm{eval}_{I,\mathrm{r}}(v_1)=ab^3a$ are distinct, the function $\mathrm{eval}_{I,\mathrm{r}}$ is an injection of  $\{t_1,v_1\}$ in $\Sigma^*$.     
    \qed
  \end{example}
  
  In conclusion, the following corollary holds from Proposition~\ref{prop exist inj cas binaire}, Proposition~\ref{prop normalization preserv eval} and Corollary~\ref{cor tphi sat phi aussi}:
  
  \begin{corollary}
    Let $\mathcal{E}=(\Sigma,\Gamma,P,F)$ be an expression environment such that $\mathrm{Card}(\Sigma)\geq 2$. Let $\phi$ be a boolean formula in $F(\Gamma)$. Then the  following two conditions are equivalent:
    \begin{itemize}
      \item $\phi$ is satisfiable,
      \item $T(\phi')$ is satisfiable.
    \end{itemize}
  \end{corollary}
  
  \begin{corollary}
    Let $\mathcal{E}=(\Sigma,\Gamma,\mathcal{P},\mathcal{F})$ be an expression environment and let $E$ be a constrained expression over $\mathcal{E}$. Then the boolean $\mathrm{Null}(E)$ can be computed.
  \end{corollary}
  
  \begin{corollary}
    Let $\mathcal{E}=(\Sigma,\Gamma,\mathcal{P},\mathcal{F})$ be an expression environment, $E$ be a constrained expression over $\mathcal{E}$ and $w$ be a word in $\Sigma^*$. The membership test of $w$ in $L(E)$ is decidable.
  \end{corollary}

\section{Conclusion and Future Work}

  In this paper, we have extended the expressive power of regular expressions by the addition of two new operators involving the zeroth order boolean formulae leading to the notion of constrained expressions. We have presented a method in order to solve the membership problem in the general case where the interpretation is not fixed and when the alphabet is not unary.
  
  An interesting continuation would be to consider the case of unary alphabets by extending the normalization defined in Subsection~\ref{ssec normalization alph non unaire} with the commutativity of the catenation; indeed, as far as a unary alphabet is considered, two words commute. Hence, any term has to be sorted according to an order (\emph{e.g.} the lexicographic order). We conjecture that Proposition~\ref{prop exist inj cas binaire} still holds for unary case, considering the word $a^p$ with $p>p'$ for any $a^{p'}$ in a term in $T$ instead of $ab^pa$ (see proof of Lemma~\ref{lem exist mot separ}).
  
   We have also shown that the membership problem can be undecidable when the interpretation is fixed. However, we can express a sufficient condition for the membership test to be decidable: whenever the interpretation $I$ is fixed, if it can be decided if a formula $\phi$ is satisfiable (\emph{e.g.} if there exists a realization $r$ such that $\mathrm{eval}_{I,r}(\phi)=1$), then (and trivially) the membership problem can be solved.
   Let us consider the following definition.
   
  \begin{definition}
    Let $\mathcal{E}=(\Sigma,\Gamma,\mathcal{P},\mathcal{F})$ be an expression environment, and $I$ be an expression interpretation over $\mathcal{E}$.
    The interpretation $I$ is \emph{decidable} if for any boolean formula $\phi$ in $\mathcal{P}(\mathcal{F}(\Gamma))$, the existence of a realization $\mathrm{r}$ such that $\mathrm{eval}_{I,\mathrm{r}}(\phi)=1$ is decidable.
  \end{definition}

  According to the previous definition, another perspective is to restrain the $\mid$ operator and the denoted language 
  in order to embed a decidable interpretation in the predicate, yielding the notion of a decidable constrained expression.   
   As an example, the predicate of length equality is decidable, and the membership of an expression using it is decidable (\emph{e.g.} $((x\dashv a^*)(y\dashv b^*)(z\dashv c^*))\mid |x|=|y|=|z|$ is such an expression).
  It is an open question to determine if the decidability of an interpretation is decidable, and how it can be characterized.
  
\bibliography{D:/DocsSyncro/Recherche/Bibliographie/biblio}

\end{document}